\documentclass[11pt]{article}
\usepackage{waingarten}
 
\usepackage{environ}
\newcounter{quote}

\NewEnviron{myquotenumber}{\vspace{3ex}\par
\refstepcounter{quote}%
\hfill\parbox{\dimexpr \textwidth-2cm}
{{\BODY}}
\hfill\llap{(\thequote)}\vspace{2ex}\par}

\title{Testing Unateness Nearly \vspace{0.5cm}Optimally}
\author{Xi Chen \\Columbia University \\ \texttt{{xichen@cs.columbia.edu}}
  \and Erik Waingarten \\Columbia University\\ \texttt{eaw@cs.columbia.edu}\vspace{0.3cm}}

\begin{document}
\maketitle
\begin{abstract}
We present an $\tilde{O}(n^{2/3}/\eps^2)$-query algorithm that tests whether 
  an unknown Boolean function $f\colon\{0,1\}^n\rightarrow \{0,1\}$ is unate
  (i.e., every variable is either non-decreasing or non-increasing) or $\eps$-far
  from unate.
The upper bound 
is nearly optimal given the $\tilde{\Omega}(n^{2/3})$ lower~bound~of~\cite{CWX17}. 
The algorithm 
builds on a novel use of the binary search procedure
  and its analysis over long random paths. %
\end{abstract}
\thispagestyle{empty}

\newpage
\thispagestyle{empty}
\tableofcontents
\thispagestyle{empty}
\newpage

\setcounter{page}{1}


\newcommand{\AESearch}{\textsc{AE-Search}}
\newcommand{\Preprocess}{\mathtt{Preprocess}}

\section{Introduction}

A Boolean function $f \colon \{0,1\}^n \to \{0,1\}$ is \emph{monotone} if every variable is non-decreasing, and \emph{unate} if every variable is either non-decreasing or non-increasing (equivalently, $f$ is unate iff there exists a string $a \in \{0,1\}^n$ such that $g(x) := f(x \oplus a)$ is monotone). Both problems of \emph{testing} monotonicity and unateness were   introduced in \cite{GGLRS00}, where a tester is a 
  randomized algorithm that, given query access to an unknown Boolean function $f \colon \{0,1\}^n \to \{0,1\}$, outputs ``accept'' with probability at least ${2}/{3}$ when $f$ is monotone (or unate) and outputs ``reject'' with probability at least ${2}/{3}$ when $f$ is~$\eps$-far from monotone (or unate).\footnote{Given a property $\calP$ of Boolean functions, 
  we say $f$ is~$\eps$-far from $\calP$ if for every $g \in \calP$, $\Prx_{\bx\sim\{0,1\}^n}[f(\bx) \ne g(\bx)] \geq \eps$ where $\bx \sim \{0,1\}^n$ is sampled uniformly at random.}
The work of \cite{GGLRS00} analyzed the non-adaptive\footnote{An algorithm is non-adaptive 
if queries made cannot depend on answers to previous queries and thus, all queries can be made in
  a single batch.
In contrast a general adaptive algorithm proceeds round by round: 
the point it~queries in each round can depend on answers to previous queries.}, one-sided error\footnote{We say a tester makes one-sided error if it always accepts a function that satisfies the property.} edge tester\footnote{An edge tester keeps drawing edges $(\bx,\by)$ from the hypercube uniformly at random
  and querying $f(\bx)$ and $f(\by)$.} which led
  to the upper bounds of $O(n/\eps)$ and $O(n^{3/2}/\eps)$ for testing monotonicity and 
  unateness, respectively.
These remained the best upper bounds for over a decade.

Recently there have been some exciting developments in understanding the query complexity of both problems. 
Progress made on the upper bound side is due, in part, to 
  new \emph{directed} isoperimetric inequalities on the hypercube. In particular, \cite{CS16} and \cite{KMS15} showed that various isoperimetric inequalities on the hypercube have directed analogues, where the edge boundary is now measured by considering anti-monotone bichromatic edges\footnote{An edge
    $\smash{(x,x^{(i)})}$ (where $\smash{x^{(i)}}$ denotes the point obtained from $x$ by flipping the $i$th bit) in $\{0,1\}^n$
    is bichromatic if $f(x)\ne f(x^{(i)})$,
    is  monotone (bichromatic) if $x_i=f(x)$,
    and is  anti-monotone (bichromatic) if $x_i\ne f(x)$.}. 
These inequalities were then used in the analysis of {non-adaptive} algorithms for testing monotonicity \cite{CS16, CST14, KMS15}. 
For example, to~obtain the $\tilde{O}(\sqrt{n}/\eps^2)$ upper bound, \cite{KMS15}
  used their new inequality to prove the existence of a large and almost regular 
  bipartite graph that consists of anti-monotone bichromatic edges in any function that is $\eps$-far
  from monotone (see Lemma \ref{lem:KMS}).
These upper bounds are complemented with lower bounds for testing monotonicity \cite{FLNRRS02,CST14, CDST15, BB16, CWX17}. For non-adaptive algorithms, the query complexity has been pinned down to $\smash{\tilde{\Theta}(\sqrt{n})}$ for constant $\eps$; for general adaptive algorithms, a gap remains between $\tilde{O}(\sqrt{n}/\eps^2)$ of \cite{KMS15}
  and the best lower bound of $\tilde{\Omega}(n^{1/3})$ \cite{CWX17}.

Given the similarity in their definitions,
  it is natural to expect that the same directed isoperimetric inequalities
  can be used to  test unateness: if $f$ is far from unate, then
  by definition $f(x\oplus a)$ is far from monotone for any $a\in \{0,1\}^n$, on which
  one can then apply these inequalities to obtain rich graph structures. 
This is indeed the approach \cite{CWX17b} followed 
  to obtain an $\smash{\tilde{O}(n^{3/4})}$-query adaptive algorithm for unateness   by leveraging the directed isoperimetric inequality of \cite{KMS15}.
It improved the upper bound of \cite{GGLRS00} as well as recent linear upper bounds 
  for testing unateness  \cite{KS16,CS16,BMPR16,BCPRS17}
  (which turned out to be optimal \cite{CWX17, BCPRS17b} for non-adaptive and one-sided error algorithms).
Shortly before the work of \cite{CWX17b}, an adaptive lower bound of
  $\tilde{\Omega}{(n^{2/3})}$ was obtained in \cite{CWX17} for testing unateness.

\ignore{\red{Erik: I'm not a huge fan of this paragraph, because I don't actually believe that the two rounds are needed for \cite{CWX17b}, somehow I believe that complexity of non-adaptive, two-sided error should be $\Theta(n^{3/4})$. If you think this paragraph sells our result more, I am okay with keeping it, but I think we could potentially remove it? }An observation made in \cite{CWX17b} is that their $\tilde{O}(n^{3/4})$-query
  algorithm only uses two \emph{rounds} \cite{CG17} of adaptivity:
  The algorithm has two stages; the first stage is non-adaptive, and each query in the 
  second stage only depends on answers to queries from the first stage.
Compared to the linear lower bound for non-adaptive algorithms \cite{CWX17, BCPRS17b},
  it is surprising that such a small amount of adaptivity brings
   a polynomial improvement in the query complexity.
\emph{Is there a better algorithm for testing unateness by utilizing more rounds of adaptivity?}}

Our main contribution is an $\tilde{O}(n^{2/3}/\eps^2)$-query, adaptive algorithm for testing unateness.
This essentially settles the problem since it matches
  the $\tilde{\Omega}(n^{2/3})$ adaptive lower bound of \cite{CWX17} 
 up to a poly-logarithmic factor (when $\epsilon$ is a constant). 
\begin{theorem}[Main]\label{thm:main}
There is an $\smash{\tilde{O}(n^{2/3}/\eps^2)}$-query, adaptive algorithm with the following property: 
Given $\eps>0$ and query access to an unknown Boolean function $f \colon \{0,1\}^n \to \{0,1\}$,
  it always accepts when $f$ is unate and rejects with probability at least $2/3$
  when $f$ is $\eps$-far from unate.
\end{theorem}

In addition to the bipartite graph structure implied by the isoperimetric inequality of
  \cite{KMS15},
the algorithm relies on novel applications 
  of the standard binary search procedure on long 
  random paths.

Given a path between two points $x$ and $y$ in the hypercube with $f(x)\ne f(y)$,
  the binary search (see Figure \ref{fig:binarysearch}) returns a bichromatic edge along the path with $\log \ell$ queries where $\ell$ is the length of the path.
The idea of using binary search in Boolean function property testing 
  is not new.
In every application we are aware of in this area 
  (e.g., in testing conjunctions \cite{DolevRon:11,ChenXie:15}, testing juntas \cite{B09, CLSSX18}, unateness \cite{KS16} and
  monotonicity \cite{CS18}),
  one runs binary search to find bichromatic edges (or pairs, as in testing juntas) 
  that can be directly used to form  a violation (or at least part of it)
   to the property being tested.
This is indeed how we use binary search in 
  one of the cases of the algorithm (Case 2) to search for an \emph{edge violation} (i.e., a 
  pair of bichromatic edges along the same variable, one is monotone and the other is anti-monotone). 
However, in the most challenging case (Case 1) of the algorithm, binary search plays a completely different role.
Instead of searching for an edge violation, binary search is used to \emph{preprocess}
  a large set $S_0\subseteq [n]$ of variables to obtain a subset $S\subseteq S_0$. This set $S$ is used to search for bichromatic edges more efficiently
  using a procedure called $\AESearch$ from \cite{CWX17b}.
Analyzing the performance of $S$ for $\AESearch$ is 
  technically the most demanding part of the paper,
  where new ideas are needed for understanding the behavior of 
  binary search running along long random paths in the hypercube.




\subsection{Technical overview}\label{techover} 

In this section we present a high-level overview of the algorithm, focusing on 
  why and how we~use binary search in Case 1 of the algorithm.
For simplicity we assume $\eps$ is a constant.

First our algorithm rejects a function only when an edge violation to unateness 
  is found.
Since an edge violation is a certificate of non-unateness, the algorithm always accepts a 
  function when it is unate and thus, it makes one-sided error. 
As a result, it suffices to show that the algorithm finds an edge violation with high probability
  when the unknown function $f$ is far from unate. 

\ignore{One noteworthy remark is that the algorithm does not leverage any new directed isoperimetric inequalities, aside from those in \cite{KMS15}. The necessary structural results about far-from-unate functions follow by a careful application of the inequality of \cite{KMS15} in all $2^n$ possible orientations of the variables. Instead, we explore new algorithmic primitives which better exploit the directed isoperimetry of far-from-unate functions.}

For simplicity, we explain Case 1 of the algorithm using the following setting:\footnote{The following conditions on the function $f$ are satisfied by the hard functions in \cite{CWX17} used for proving the $\tilde{\Omega}(n^{2/3})$ lower bound.}
\begin{flushleft}\begin{enumerate}
\item[] All edge violations of $f$ are along a hidden set $\calI \subset [n]$ of $\Omega(n)$ variables. 
For 
each variable $i\in \calI$, there are $\Theta(2^n/n)$ monotone edges and $\Theta(2^n/n)$ 
    anti-monotone edges. 
    Let $P_i^+$ denote the set of points incident on monotone edges along $i$ and 
    $P_i^-$ denote the set of points incident on anti-monotone edges along $i$. 
  The sets $P_i^+$'s for $i\in \calI$ are disjoint, so monotone edges along variables in $\calI$ 
  form a matching of size $\Theta(2^n)$;
  similarly, the sets $P_i^-$'s are disjoint and anti-monotone edges along $\calI$ 
  also form a matching of size $\Theta(2^n)$. 
Along each $i\notin \calI$, there are $\Theta(2^n/\sqrt{n})$ bichromatic edges along $i$ which are all either monotone or anti-monotone, but not both.
\end{enumerate}\end{flushleft}


This particular case will highlight some of the novel ideas in the algorithm and the analysis, so we focus on this case for the technical overview. 

An appealing approach for finding an edge violation 
  is to keep running binary search on points $\bx,\by$ that are drawn 
  independently and unifomly at random.
Since a function that is far from unate must  be $\eps$-far from constant as well, 
  $f(\bx)\ne f(\by)$ with a constant probability and 
  when this happens, binary search returns a bichromatic edge.
Now in order to analyze the chance of observing an edge violation by repeating this process, 
  two challenges arise.
First, the output distribution given by the variable of the bichromatic edge
 found by binary search can depend on $f$ in subtle ways, and becomes difficult to analyze formally (partly because of its adaptivity).
Second, 
 since the influence of variables outside $\calI$ is $\Omega(1/\sqrt{n})$, a random path between $\bx$ and $\by$ of $\Omega(n)$ edges may often cross $\Omega(\sqrt{n})$ bichromatic edges along variables outside of $\calI$ and $O(1)$ bichromatic edges along variables in $\calI$. In this case, binary search will likely return a bichromatic edge along a variable outside $\calI$, which is useless for finding an edge violation.


A less adaptive (and thus much simpler to analyze) 
  variant of binary search called $\AESearch$ was introduced \cite{CWX17b}
  to overcome these two difficulties. The subroutine
$\AESearch\hspace{0.05cm}(f, x,S)$\footnote{See Figure~\ref{fig:edge-search} in Appendix~\ref{sec:aes} for a formal description of the $\AESearch$ subroutine.} queries $f$ and takes two additional inputs: $x\in \{0,1\}^n$ and
  a set $S\subseteq [n]$ of variables, uses $O(\log n)$ queries and satisfies the following property:
\begin{flushleft}\begin{myquotenumber}\label{quote:aesearch}
\textbf{Property of \AESearch:} If $(x,x^{(i)})$ is a bichromatic edge with $i\in S$ and 
  both $x$\\ and $x^{(i)}$ are $(S\setminus \{i\})$-\emph{persistent} (which for $x$ informally means that $f(x)=f(x^{(\bT)})$\\with high probability
  when $\bT$ is a uniformly random subset of $S\setminus \{i\}$ of half of its\\ size),
  then $\AESearch\hspace{0.05cm}(f, x,S)$ finds the edge $(x,x^{(i)})$ with probability at least $2/3$.
\end{myquotenumber}\end{flushleft}
In some sense, $\AESearch\hspace{0.05cm}(f, x, S)$ efficiently checks whether there 
  exists an $i \in S$ such that $(x, x^{(i)})$ is bichromatic, whereas the trivial algorithm for this task takes $O(|S|)$ queries.\footnote{See Definition~\ref{def:persistency} and its relation to the performance of $\AESearch$ in Lemma~\ref{lem:aesearch} for a formal description} 

In this simplified setting,
 the algorithm of \cite{CWX17b} starts by drawing a size-$\sqrt{n}$ set $\bS\subseteq [n]$ uniformly 
  at random and runs $\AESearch\hspace{0.03cm}(f, \bx,\bS)$
  on independent samples $\bx$ for $n^{3/4}$ times, hoping to find an edge violation. 
To see why this works we first note that $|\bS\cap \calI|=\Omega(\sqrt{n})$~with
  high probability. 
Moreover, the following property holds for $\bS$: %
\begin{flushleft}\begin{myquotenumber}\label{quote:random-set}
\textbf{Property of the Random Set $\bS$:} With $\Omega(1)$ probability over the randomness\\ of $\bS$,
  most $i\in \bS\cap \calI$ satisfy that most points in $P_{i}^+$ and $P_{i}^-$ are 
  $(\bS\setminus \{i\})$-persistent.
\end{myquotenumber}\end{flushleft}
We sketch its proof since it highlights the technical challenge we will face later.

First we view the sampling of $\bS$ as $\bS'\cup\{\bi\}$,
  where $\bS'$ is a random set of size $\sqrt{n}-1$ and $\bi$ is a random variable in $[n]$.
Since the influence of each variable in $\bS'$ is at most $O(1/\sqrt{n})$, for many points $x \in \{0,1\}^n$ most random paths of length $O(\sqrt{n})$ along variables in $\bS'$ starting at $x$ will not cross any bichromatic edges.  
In other words, most random sets $\bS'$ of size $\sqrt{n}-1$ satisfy that 
  most of points in $\{0,1\}^n$ are $\bS'$-persistent with high constant probability.
Given that $\cup_i P_i^+$ and $\cup_i P_i^-$ are both $\Omega(1)$-fraction of $\{0,1\}^n$,
   most points in $\cup_i P_i^+$ and $\cup_i P_{i}^-$ must be $\bS'$-persistent as well.
On the other~hand,~given that $\bi$ is \emph{independent} from $\bS'$ and that $\calI$ is $\Omega(n)$,
  with probability $\Omega(1)$ many points in $P_{\bi}^+$ and $P_{\bi}^-$ are $\bS'$-persistent.
The property of $\bS$ follows by an argument of expectation.

With properties of both $\bS$ and $\AESearch$ in hand in (\ref{quote:aesearch}) and (\ref{quote:random-set}), as well as the fact that $|\bS \cap \calI| = \Omega(\sqrt{n})$ with high probability, 
we expect to find a bichromatic edge along a variable in $\bS\cap \calI$ after $\sqrt{n}$ executions of $\AESearch$ (since the union of $P_i^+$ and $P_i^-$ for $i\in \bS\cap\calI$ consists of $\Omega(1/\sqrt{n})$-fraction of $\{0,1\}^n$).
Moreover, the variable is (roughly) uniformly over $\bS\cap \calI$
  and (roughly) equally likely to be monotone or anti-monotone.
It follows 
from the birthday paradox that repeating $\AESearch$ for $O(n^{1/4})\cdot \sqrt{n}$ rounds 
  is enough to find an edge violation.

The natural question is whether we can make $\bS$ larger (e.g., of size $n^{2/3}$) without breaking property (\ref{quote:random-set}).
This would lead to an $\tilde{O}(n^{2/3})$-query algorithm (for the simplified setting). However, it is no longer true that many random paths of length $\Omega(n^{2/3})$ do not cross bichromatic edges because the influence of variables along variables in $\bS \setminus \calI$ is $\Omega(1/\sqrt{n})$. Therefore, large $\bS$ may not satisfy property (\ref{quote:random-set}) and as a result, $\AESearch$ may never output bichromatic edges along variables in $\bS \cap \calI$.
This limit to sets of size at most $O(\sqrt{n})$
 was a similar bottleneck in \cite{KMS15}, and 
  the connection between $|\bS|$ and the total influence of $f$ was later explored in \cite{CS18}.
Indeed, if 
 (\ref{quote:random-set}) held for $\bS$ of size larger than $\sqrt{n}$,
  then one could improve on the 
  $O(\sqrt{n})$-query algorithm of \cite{KMS15} for testing
  monotonicity.
Consequently, if one believes that monotonicity testing requires $\Omega(\sqrt{n})$ \emph{adaptive} queries,
  it is natural to conjecture that the algorithm in \cite{CWX17b} is optimal for testing unateness.

The key insight in this work 
is to
\emph{preprocess} the set $\bS$ before using $\AESearch$. 
For our simplified setting, we first
  sample $\bS_0 \subset [n]$ of size $n^{2/3}$ (much larger than what the analysis~in~\cite{KMS15, CWX17b, CS18} would allow) uniformly at random.
Then, we set $\bS=\bS_0$, and repeat the following steps for $n^{2/3}\cdot\polylog(n)$ many iterations:
\begin{flushleft}\begin{quote}
$\Preprocess$:  Sample $\bx\in \{0,1\}^n$ uniformly at random.
Check if $\bx$ is $\bS$-persistent by drawing $\polylog(n)$ many subsets $\bT\subseteq\bS$ of half of its size
 uniformly at random.
If a $\bT$ with $\smash{f(\bx)\ne f(\bx^{(\bT)})}$ is found, run binary search
  on a random path from $f(\bx)$ to $f(\bx^{(\bT)})$  
  to find a bichromatic edge along variable $\bi$ and remove $\bi$ from $\bS$.
\end{quote}\end{flushleft}
At a high level, the analysis of the algorithm would proceed as follows. 
At the end of $\Preprocess$, for every $i\in \bS$, 
  most points in $\{0,1\}^n$ are $(\bS\setminus \{i\})$-persistent. Otherwise, $\Preprocess$ would 
  remove more variables from $\bS$ 
  since 
  points which are not $(\bS\setminus \{i\})$-persistent cannot be very $\bS$-persistent. 
At the same time, most variables in $\bS_0\cap \calI $ at the beginning survive~in~$\bS$~at the end (given that 
  variables in $\calI$ have very low influence).
It may seem that we can now conclude property (\ref{quote:random-set}) holds for $\bS$, and that a violation is found after $O(n^{2/3})$ rounds of $\AESearch(f, \bx, \bS)$ when $\bx$ is uniform.

However, the tricky (and somewhat subtle) problem is that, even though most points in $\{0,1\}^n$ 
  are $(\bS\cap \{i\})$-persistent for every $i\in \bS\cap \calI$, 
  it is not necessarily the case that points \emph{inside} $P_i^+$ and~$P_i^-$ 
  are $(\bS\cap \{i\})$-persistent, since $P_i^+\cup P_i^-$ is only a $O(1/n)$-fraction of the hypercube.
Compared with~the argument from \cite{CWX17b} above for $\sqrt{n}$-sized uniformly random sets, 
  after preprocessing $\bS_0$ (which was a uniform random set) with multiple rounds
  of binary search, the set $\bS$ left can be very far from random.
More specifically, the set $\bS$ obtained from $\bS_0$
  will heavily depend on the function $f$ and, in principle, a clever adversary 
  could 
  design a function so that $\Preprocess$ running on $\bS_0$
  deliberately outputs a set $\bS$  
  that where points in $P_i^+$ and $P_i^-$ are not $(\bS \setminus \{i\})$-persistent. 
   
The main technical challenge is to show that this is not possible when 
  variables in $\calI$~have low~influence,\footnote{In the simplified setting, each variable $i\in \calI$
  has influence only $O(1/n)$; In the real situation, we need to handle the case even when each
  variable has influence as high as $1/n^{2/3}$.} and 
  the desired property for $\bS$ remains valid.
 To this end, we show that for any variable $i$ with low influence, the following two distributions supported on preprocessed sets $\bS$ have small total variation distance. The first distribution samples $\bS_0' \subset [n]$ of size $(n^{2/3}-1)$ and outputs the set $\bS' \cup \{i\}$ obtained from preprocessing $\bS_0'$. The second distribution $\bS_0' \subset [n]$ of size $(n^{2/3}-1)$ and outputs the set $\bS$ obtained from preprocessing $\bS_0' \cup \{i\}$.
Intuitively this means that a low-influence variable $i$ has little impact
  on the result $\bS$ of $\Preprocess$ and thus, 
  $\Preprocess$ is oblivious to $i$ and 
  cannot deliberately 
  exclude $P_i^+$ and $P_i^-$ from the set of $\bS$-persistent points.


To analyze the total variation distance between the results of running $\Preprocess$ on $\bS_0'$~and 
  $\bS_0'\cup\{i\}$,
we need 
to understand how a low-influence variable $i$
  can affect the result~of a binary search on a \emph{long random path} (given that $\Preprocess$
  is just a sequence of calls to binary search). The random paths have length
  $|\bS_0|=n^{2/3}$ at the beginning of $\Preprocess$, and are repeated for $\tilde{O}(n^{2/3})$ rounds. 
Giving more details, we show that a variable $i$ with influence $\Inf_f[i]$ can affect 
  the result of a binary search on a random path of length $\ell$ with probability
  at most $\log \ell\cdot \Inf_f[i]$, instead of the trivial upper bound of $\ell\cdot \Inf_f[i]$, which is the probability that a variable $i$ affects the evaluation of $f$ on vertices of a random path of length $\ell$.
This is proved in Claim \ref{binaryclaim} (although the formal statement 
  is slightly different since we need to introduce a placeholder when running
  binary search on the set without $i$ so that the two paths have the same length; see Section \ref{sec:binarysearch}).
%

In order to go beyond the assumptions on the function given in this overview, the algorithm needs to deal with more general cases: 
(1) Monotone (or anti-monotone) edges of $\calI$ may not form a matching, but rather, a large and almost-regular bipartite graph whose existence follows from the 
    directed isoperimetric inequality of \cite{KMS15}.
 (2) Although \cite{KMS15} implies the existence of such graphs
    with bichromatic edges from $\calI$,
    there may be more bichromatic edges along $\calI$
    outside of these two graphs, which would raise the influence of these variables to the point where $\Preprocess$ is no longer oblivious of these variables. Intuitively, this implies that bichromatic edges which give rise to edge violations are abundant, so finding them becomes easier. This is handled in Case 2, where we give an algorithm (also based on binary search) which finds many bichromatic edges along these high influence variables, and combine it with the techniques from \cite{CWX17b} to find an edge violation. (3) The set $\calI$ can be much smaller than $n$, in which case, the techniques from \cite{CWX17b} actually achieves better query complexity. We formalize this in Case 3 of the algorithm. 

\subsection{Organization}

We review preliminaries, recall the binary search procedure and review the definition of persistency and
  the $\AESearch$ procedure in Section \ref{sec:pre}. 
We  present the preprocessing procedure~in~Section \ref{sec:preprocess}
  and prove that a low-influence variable has small impact on its output.
We use the directed isoperimetric lemma of \cite{KMS15} to establish a so-called 
  Scores Lemma in Section \ref{sec:score}, which roughly speaking helps us
  understand how good the set $\bS$ is after preprocessing (in terms of 
  using it to run $\AESearch$ to find a bichromatic edge along a certain variable).
We separate our main algorithm into three cases in Section \ref{sec:mainalg},
  depending on different combinations of parameters.
Case 1, 2 and 3 of the algorithm are presented and analyzed in 
  Sections \ref{algsec:case1}, \ref{algsec:case2} and \ref{algsec:case3}, respectively.
Section \ref{sec:highinf} presents a procedure used in Case 2 to find bichromatic
  edges of variables with relatively high influence.


\newcommand{\Prune}{\mathtt{Preprocess}}
\newcommand{\TestAB}{\mathtt{CheckPersistence}}
\newcommand{\BinarySearch}{\mathtt{BinarySearch}}
\newcommand{\nil}{\mathtt{nil}}
\newcommand{\uu}{\boldsymbol{u}}
\newcommand{\vv}{\boldsymbol{v}}
\newcommand{\Sub}{\mathtt{Sub}\hspace{0.04cm}}
\newcommand{\I}{\mathbf{I}}

\section{Preliminaries}\label{sec:pre}

We will use bold-faced letters such as $\bT$ and $\bx$ to denote random variables. For $n \geq 1$, we write $[n] = \{ 1, \dots, n\}$. In addition, we write $g = \tilde{O}(f)$ to mean $g = O(f \cdot \polylog(f))$ and $g = \tilde{\Omega}(f)$ to mean $g = \Omega(f / \polylog(f))$.

For $x \in \{0,1\}^n$, and a set $S \subset [n]$, we write $x^{(S)} \in \{0,1\}^n$ as the point given by letting $\smash{x_k^{(S)} = x_k}$ for all $k \notin S$, and $\smash{x_k^{(S)} = 1 - x_k}$ for all $k \in S$
  (i.e., $x^{(S)}$ is obtained from $x$ by \emph{flipping} variables in $S$). When $S = \{ i \}$ is a singleton set, we abbreviate $x^{(i)} = x^{(\{i\})}$ and say that $x^{(i)}$ is obtained from $x$ by  {flipping} the $i$th variable. Throughout the paper, we use $n+1$ as the name of a \emph{placeholder} variable (i.e., a dummy variable). If $x \in \{0,1\}^n$ and $S \subseteq [n+1]$, then $x^{(S)}: = x^{(S \setminus \{n+1\})}$, and in particular, $x^{(n+1)} = x$. We will refer to this as \emph{flipping} variable $n+1$ (see Section~\ref{sec:binarysearch}) although no change is made on $x$. For a subset $S \subseteq [n+1]$ and a variable $i \in [n]$, we let $\Sub(S,i) \subseteq [n+1]$ be the subset obtained by \emph{substituting} $n+1$ with $i$ and 
  $i$ with $n+1$. In other words, 
$$
\Sub(S,i)=\begin{cases} S & \text{if $i,n+1\in S$ or $i,n+1\notin S$} \\
  (S\cup\{n+1\})\setminus \{i\} & \text{if $i\in S$ and $n+1\notin S$}\\
  (S\cup\{i\})\setminus \{n+1\} & \text{if $n+1\in S$ and $i\notin S$.}\end{cases}
$$

 We will at times endow $S \subseteq [n+1]$ with an \emph{ordering} $\pi \colon [|S|] \rightarrow S$ which is a bijection indicating that $\pi(i)$ is the $i$th element of $S$ under $\pi$. When $T \subset S$, the ordering $\tau \colon [|T|] \to T$ obtained from $\pi$ is the unique bijection such that for all $i, j \in T$, $\tau^{-1}(i) < \tau^{-1}(j)$ if and only if $\pi^{-1}(i) < \pi^{-1}(j)$. Moreover, when $S \subseteq [n+1]$ and $\pi$ is an ordering of $S$, the ordering $\pi' $ of $\Sub(S, i)$ obtained from $\pi $ is obtained by 
   substituting $n+1$ with $i$ and $i$ with $n+1$ in the ordering,
   i.e., $\pi'(k)=\pi(k)$ when $\pi(k)\notin \{i,n+1\}$, $\pi'(k)=n+1$ if $\pi(k)=i$
   and $\pi'(k)=i$ if $\pi(k)=n+1$.

Given a Boolean function $f \colon \{0,1\}^n \to \{0,1\}$, and a variable $i \in [n]$, we say that $(x, x^{(i)})$ is a \emph{bichromatic edge} of $f$ along variable $i$ if $f(x) \neq f(x^{(i)})$;
  it is monotone (bichromatic) if $x_i=f(x)$ and 
  anti-monotone (bichromatic) if $x_i\ne f(x)$. The \emph{influence} of variable $i$ in $f$ is defined as $$\Inf_f[i] = \Prx_{\bx \sim \{0,1\}^n}\left[ f(\bx) \neq f(\bx^{(i)})\right],$$ which is twice the number of bichromatic edges of $f$ along $i$ divided by $2^n$. The \emph{total influence} of $f$, $\I_f = \sum_{i\in [n]} \Inf_f[i],$ is twice the number of bichromatic edges of $f$ divided by $2^n$.
Given distributions $\mu_1$ and $\mu_2$ on some sample space $\Omega$, the \emph{total variation distance} between $\mu_1$ and $\mu_2$ is given by $$\dtv(\mu_1, \mu_2) = \max_{S \subseteq \Omega} \big|\mu_1(S) - \mu_2(S) \big|.$$

\subsection{Binary search with a placeholder}\label{sec:binarysearch}

We use the subroutine $\BinarySearch\hspace{0.05cm}(f,x,S,\pi)$ described in Figure \ref{fig:binarysearch},
  where $f\colon\{0,1\}^n\rightarrow \{0,1\}$~is a Boolean function, $x\in \{0,1\}^n$,
  $S$ is a nonempty subset of $[n+1]$, and $\pi$ is an ordering of $S$.

When $S\subseteq [n]$, 
  $\BinarySearch\hspace{0.05cm}(f, x, S, \pi)$ performs as the standard binary search algorithm: 
  $x=x_0,x_1,\ldots,x_{|S|}=x^{(S)}$ is~a path from $x$ to $\smash{x^{(S)}}$ in which $x_t$ is obtained 
  from $x_{t-1}$ by flipping variable $\pi(t)\in S\subseteq [n]$, and when $\smash{f(x) \neq f(x^{(S)})}$, the binary search is done along
  this path to find an edge that is bichromatic.
Now in general $S$ may also contain $n+1$, which we use as the~name of a placeholder variable.
Similarly, when $f(x) \neq f(x^{(S)})$, the binary search is done along the path 
  $x=x_0,x_1,\ldots,x_{|S|}=x^{(S)}$ (recall that $x^{(S)}$ is defined as $x^{(S\setminus \{n+1\})}$
  when $S$ contains $n+1$) where
  $x_t$ is obtained from $x_{t-1}$ by flipping variable $\pi(t)$
  (in particular, when $\pi(t) = n+1$, $x_t=x_{t-1}$).  

Note that even though $n+1$ is a placeholder variable, given $S \subseteq [n+1]$ 
  with $n+1 \in S$ and an ordering $\pi$ of $S$, queries made by $\BinarySearch\hspace{0.05cm}(f, x, S, \pi)$ and $\BinarySearch\hspace{0.05cm}(f, x, S \setminus \{n+1\}, \pi')$ (where $\pi'$ is the ordering of
  $S \setminus \{n+1\}$ obtained from $\pi$) are different, so their results may also be different.
We summarize properties of $\BinarySearch$ in the following lemma.

\begin{lemma}
 $\BinarySearch\hspace{0.04cm}(f, x, S, \pi)$ uses $O(\log n)$ queries and satisfies the following property.
If $f(x)=f(x^{(S)})$, it returns $\nil$; if $f(x)\ne f(x^{(S)})$, it returns
  a variable $i\in S \setminus \{n+1\}$ and a point $y \in \{0,1\}^n$ along the path from
  $x$ to $x^{(S)}$ with ordering $\pi$ such that $(y, y^{(i)})$ is bichromatic.
\end{lemma}

\begin{figure}
\begin{framed}
\noindent Subroutine $\BinarySearch\hspace{0.05cm}(f,x, S, \pi)$
\begin{flushleft}
\noindent {\bf Input:} Query access to $f \colon \{0, 1\}^n \to \{0, 1\}$,  a point 
  $x \in \{0,1\}^n$, a nonempty set $S \subseteq [n+1]$ and an ordering $\pi$ of $S$.\\
{\bf Output:} Either $i \in S$ and a point $y \in \{0,1\}^n$ where $(y, y^{(i)})$ is a bichromatic edge, or $\nil$.

\begin{enumerate}
\item Query $f(x)$ and $f(x^{(S)})$ and return $\nil$ if $f(x)=f(x^{(S)})$.
\item Let $m=|S|$ and $x=x_0,x_1,\ldots,x_{m}=x^{(S)}$ be the sequence of points obtained from $x$ by flipping variables in the order of $\pi(1),\ldots,\pi(m)$: $\smash{x_{i} = x_{i-1}^{(\pi(i))}}$. Let $\ell=0$ and $r=m$.
\item While $r-\ell>1$ do
\item \ \ \ \ \ \ \ \ Let $t=\lceil (\ell+r)/2 \rceil$ and query $f(x_t)$. If $f(x_\ell)\ne f(x_t)$ set $r=t$; otherwise set $\ell=t$.\
\item Return $\pi(r)$ and $y = x_{\ell}$.
\end{enumerate}
\end{flushleft}\vskip -0.14in
\end{framed}\vspace{-0.2cm}
\caption{Description of the binary search subroutine for finding a bichromatic edge.} \label{fig:binarysearch}
\end{figure}
\subsection{Persistency with respect to a set of variables}

We need the following notion of persistency for points and edges with respect to a set of variables.

\begin{definition}\label{def:persistency}
Given a Boolean function $f\colon\{0,1\}^n\rightarrow \{0,1\}$, a set $S\subseteq [n+1]$ of variables and~a point $x\in \{0,1\}^n$, we say that $x$ is \emph{$S$-persistent} if the following two conditions hold:
\[ \Prx_{\substack{\ \\ \bT \subseteq S \\ |\bT| = \left\lfloor |S|/2\right\rfloor}}\left[ f(x) = f(x^{(\bT)})\right] \geq 1 - \frac{1}{\log^2 n} \quad\text{and}\quad \Prx_{\substack{\bT \subseteq S \\ |\bT| = \lfloor |S|/2\rfloor+1}}\left[ f(x) = f(x^{(\bT)})\right] \geq 1 - \frac{1}{\log^2 n} . \]
where $\bT$ is a subset of $S$ of certain size drawn uniformly at random.
Note that when $S=\emptyset$, every point in $\{0,1\}^n$ is trivially $S$-persistent.

Let $e$ be an edge in $\{0,1\}^n$.
We say that $e$ is \emph{$S$-persistent} if both points of $e$ are $S$-persistent.
\end{definition}

The notion of persistency above is useful because it can be used to formulate a clean sufficient condition for $\AESearch\hspace{0.05cm}(f,x,S)$ to find a bichromatic edge $(x,x^{(i)})$ 
  for some $i\in S$ with high~probability.
This is captured in Lemma \ref{lem:aesearch} (see Lemma~6.5 in \cite{CWX17b}) below.
For completeness we include the description of $\AESearch$ \cite{CWX17b} 
  and the proof of Lemma \ref{lem:aesearch} in Appendix 
  \ref{sec:aes}.

\begin{lemma}\label{lem:aesearch}
Given a point $x\in \{0,1\}^n$ and a set $S\subseteq [n+1]$,
$\AESearch\hspace{0.05cm}(f,x,S)$ makes $O(\log n)$ queries to $f$, and returns either an
  $i\in S$ such that $(x,x^{(i)})$ is a bichromatic edge, or ``fail.''
  
Let $(x, x^{(i)})$ be a bichromatic edge of $f$ along $i\in [n]$. 
If $i\in S$ and $(x,x^{(i)})$ is $(S\setminus \{i\})$-persistent, then 
  both $\AESearch\hspace{0.05cm}(f,x, S )$ and $\AESearch\hspace{0.05cm}(f,x^{(i)}, S )$ 
  output $i$ with probability at least $2/3$.
\end{lemma}

\red{Lemma \ref{lem:aesearch} has the following immediate corollary.
\begin{corollary}\label{AESearchCorollary}
Given a set $S\subseteq [n+1]$ and a point $x\in \{0,1\}^n$,
  there exists at most one variable $i \in S$ such that $(x,x^{(i)})$ is both bichromatic and $(S \setminus \{i\})$-persistent.
\end{corollary}
\begin{proof}
If the condition holds for both $i\ne j\in S$,
  then from Lemma \ref{lem:aesearch} $\AESearch\hspace{0.05cm}(f,x,S)$ would return both $i$ and $j$ with probability at least $2/3$, a contradiction.
\end{proof}}

\section{Preprocessing Variables}\label{sec:preprocess}

Our goal in this section is to present a preprocessing procedure called $\Prune$.
Given query~access to a Boolean function
  $f\colon\{0,1\}^n\rightarrow \{0,1\}$, 
  a nonempty set $S_0\subseteq [n+1]$ (again, $n+1$ serves here as a placeholder variable), an ordering $\pi$ of $S_0$ and a parameter $\xi\in (0,1)$,
  $\Prune\hspace{0.03cm}(f ,S_0,\pi,\xi)$ makes $(|S_0|/\xi)\cdot \text{polylog}(n)$  queries
  and returns a subset $\bS$ of $S_0$.
At a high level, $\Preprocess$ keeps running $\BinarySearch$ to
  remove variables from $S_0$ until the set $\bS\subseteq S_0$ left 
  satisfies that at least $(1-\xi)$-fraction of points in $\{0,1\}^n$
  are~$\bS$-persistent (recall Definition \ref{def:persistency}).

In addition to proving the above property for $\Prune$ in Lemma \ref{lem:prune2},
  we show in Lemma~\ref{lem:main} the following: When $i\in S_0\subseteq [n]$ has low influence,
  then the result of running $\Prune$ on $S_0$ is \emph{close} (see Lemma \ref{lem:prune2}
  for the formal statement) to
  that of running it on $\Sub(S_0,i)$ (in which we substitute $i$
  with the placeholder variable $n+1$). 

\subsection{The preprocessing procedure}\label{sec:prune}

\begin{figure}
\begin{framed}
\noindent Subroutine $\TestAB\hspace{0.05cm}(f,S, \pi,\xi)$
\begin{flushleft}
\noindent {\bf Input:} Query access to $f \colon \{0, 1\}^n \to \{0, 1\}$, a nonempty set $S \subseteq [n+1]$, 
an ordering $\pi$ of $S$ and a parameter $\xi\in (0,1)$. \\
\noindent {\bf Output:} Either $\nil$ or a variable $i \in S$. 
\begin{enumerate}
\item Repeat the following steps ${\log^4 n}/{\xi}$ many times:
\begin{enumerate}
\item Sample a point $\bx$ from $\{0,1\}^n$ uniformly at random.\vspace{0.06cm}
\item 
Flip a fair coin and perform one of the following tasks:\vspace{0.17cm}
\begin{itemize}
\item Sample $\bT \subseteq S$ with size $\lfloor  |S|/2 \rfloor$ uniformly.
  Run $\BinarySearch\hspace{0.03cm}(f,\bx,\bT, \bpi')$ where $\bpi'$ is
  the ordering of $\bT$ defined by $\pi$ restricted on $\bT$. If $\BinarySearch(f, \bx ,\bT, \bpi')$ returns a variable $i$ and a point $y$, output $i$.\vspace{0.2cm}
\item Sample $\bT \subseteq S$ with size $\lfloor  |S|/2 \rfloor+1$ uniformly.
  Run $\BinarySearch\hspace{0.03cm}(f,\bx,\bT, \bpi')$ where $\bpi'$ is
  the ordering of $\bT$ defined by $\pi$ restricted on $\bT$. If $\BinarySearch(f, \bx ,\bT, \bpi')$ returns a variable $i$ and a point $y$, output $i$.

\end{itemize} 
\end{enumerate}
\item If $\BinarySearch$ always returned $\nil$, output $\nil$.
\end{enumerate}
\end{flushleft}\vskip -0.14in
\end{framed}\vspace{-0.2cm}
\caption{Description of the  subroutine $\TestAB$.} \label{fig:test-ab}
\end{figure}

\begin{figure}
\begin{framed}
\noindent Procedure $\Prune\hspace{0.05cm}(f,S_0, \pi,\xi)$
\begin{flushleft}
\noindent {\bf Input:} Query access to $f \colon \{0, 1\}^n \to \{0, 1\}$, a nonempty set $S_0 \subseteq [n+1]$, an
  ordering $\pi$ of $S_0$ and a parameter $\xi\in (0,1)$. \\
{\bf Output:} A subset $S \subseteq S_0$.

\begin{enumerate}
\item Initially, let $S= S_0$ and $\tau=\pi$.
\item While $S$ is nonempty do
\item \ \ \ \ \ \ \ \ Run $\TestAB\hspace{0.05cm}(f,S,\tau,\xi)$.
\item \ \ \ \ \ \ \ \ If it returns $\nil$, return $S$; otherwise (it returns an $i\in S$), 
  remove $i$ from $S$ and $\tau$.
\item Return $S$ (which must be the empty set to reach this line).
\end{enumerate}
\end{flushleft}\vskip -0.14in
\end{framed}\vspace{-0.2cm}
\caption{Description of the procedure $\Prune$ for preprocessing a set of variables.} \label{fig:prune}
\end{figure}

The procedure $\Prune\hspace{0.05cm}(f,S_0,\pi,\xi)$ is described in Figure \ref{fig:prune}.
It uses a subroutine $\TestAB\hspace{0.05cm}(f,S,\pi,\xi)$ described in Figure \ref{fig:test-ab}.
Roughly speaking, $\TestAB$ checks if at least $(1-\xi)$-fraction of points in $\{0,1\}^n$
  are $S$-persistent for the current set $S$. 
This is done by sampling points $\bx$ and subsets $\bT$ of $S$ of the right sizes
  uniformly at random, and checking if $\smash{f(\bx)=f(\bx^{(\bT)})}$, for $\log^4 n/\xi$ many rounds.
If $\TestAB$ finds $\bx$ and $\bT$ such that $\smash{f(\bx)\ne f(\bx^{(\bT)})}$,
  it runs binary search on them to find a bichromatic edge along some variable $i\in S$ and 
  outputs $i$;
  otherwise it returns $\nil$.

The main property we prove for $\TestAB$ (see Lemma~\ref{lem:test-lemma} in Appendix \ref{proof:pruneanalysis})
  is that when the fraction of points that are not $S$-persistent is at least $\xi$,
  it returns a variable $i\in S$ 
  with high probability. 
  
The procedure $\Prune\hspace{0.05cm}(f,S_0,\pi,\xi)$ sets $S=S_0$ and $\tau=\pi$
  at the beginning and keeps calling $\TestAB \hspace{0.05cm}(f,S,\tau,\xi)$ and removing
  the variable $\TestAB(f, S, \tau, \xi)$ returns from both $S$ and the ordering $\tau$, until $\TestAB$ returns $\nil$ 
  or $S$ becomes empty in which case $\Prune$ terminates and returns  $S$.
As a result, $\Prune$ makes at most $|S_0|$ calls to $\TestAB$.
Using the property of $\TestAB$ from Lemma~\ref{lem:test-lemma}, it is unlikely that $\xi$-fraction of points are not $S$-persistent
  but somehow $\TestAB\hspace{0.05cm}(f,S,\tau,\xi)$ returns $\nil$.
This implies that at least $(1-\xi)$-fraction of $\{0,1\}^n$ are $\bS$-persistent 
  for $\bS=\Prune\hspace{0.05cm}(f,S_0,\pi,\xi)$ at the end with high probability.

We summarize our discussion above in the following lemma but delay its proof
  to Appendix \ref{proof:pruneanalysis} since it follows from standard applications of Chernoff bounds and union bounds.

\begin{lemma}\label{lem:prune2}
Given a Boolean function $f\colon\{0,1\}^n\rightarrow \{0,1\}$, 
  a nonempty $S_0 \subseteq [n+1]$, an ordering $\pi$ of $S_0$
  and a parameter $\xi\in (0,1)$, $\Prune\hspace{0.05cm}(f,S_0, \pi,\xi)$ makes
  at most  $O(|S_0|\hspace{0.05cm}{\log^\red{5} n}/{\xi})$  queries to $f$ and with probability at least $1 - \exp\left(-\Omega(\log^2 n)\right)$,
  it outputs a subset $\bS \subseteq S_0$ 
  such that 
   at least $(1-\xi)$-fraction of points in $\{0,1\}^n$ are $\bS$-persistent. 
\end{lemma}

\subsection{Low influence variables have low impact on $\Preprocess$}\label{sec:prune2}


In the rest of the section, we show that when $S_0\subseteq [n]$,
  \red{a variable $i\in S_0$ with 
  low influence $\Inf_f[i]$ has low impact on the result
  of $\bS= \Prune\hspace{0.03cm}(f,S_0, \pi,\xi)$. 
More formally, we show that one can substitute $i$ by the placeholder $n+1$ and 
  the result of running $\Prune$ on $\Sub(S_0,i)$ is~almost the same (after substituting
  $n+1$ back to $i$ in the result of $\Prune$).}

This is made more precise in the following lemma:

 
\begin{lemma}\label{lem:main}
Let $f\colon\{0,1\}^n\rightarrow \{0,1\}$ be a Boolean function. 
Let $i\in S_0\subseteq [n]$, 
  $\pi$ be an ordering~of~$S $ and $\xi\in (0,1)$. 
Let $S_0' = \Sub(S_0, i)$ be the subset of $[n+1]$ 
  and~let $\pi'$ be the ordering of $S_0'$ obtained from $\pi$ by
  substituting $i$ with $n+1$.  
Then we have
$$
\dtv\Big(\Prune\hspace{0.05cm}(f,S_0,\pi,\xi),\hspace{0.03cm}\Sub\big(\Prune\hspace{0.05cm}(f ,S_0',\pi',\xi),i\big)\Big)
\le O\left(\frac{|S_0|\hspace{0.01cm}\log^\red{5} n}{\xi}\right)\cdot \Inf_f[i].
$$
\end{lemma}
  
Because $\Prune$ keeps calling $\TestAB$ which keeps calling $\BinarySearch$,
  we start the proof~of Lemma \ref{lem:main}
  with the following claim concerning the binary search procedure.  
  

\begin{claim}\label{binaryclaim}
Let $i\in S\subseteq [n]$ 
  and $\pi$ be an ordering of $S $. 
Let $S' = \Sub(S, i)$, 
  and $\pi'$ be the ordering~of $S'$ obtained from $\pi$ by
  substituting $i$ with $n+1$. 
We let $\bu$ and $\bv$ be the random variables where
\begin{flushleft}  \begin{itemize}
\item $\bu$ is the output of $\BinarySearch\hspace{0.05cm}(f ,\bx,S,\pi)$
  when $\bx$ is drawn from
  $\{0,1\}^n$ uniformly, and \vspace{-0.07cm}
\item $\bv$ is the output of $\BinarySearch\hspace{0.05cm}(f ,\bz,S',\pi')$ 
  when $\bz$ is drawn from $\{0,1\}^n$ uniformly. \vspace{0.06cm}
 \end{itemize}\end{flushleft}
Then, we have $ \dtv(\bu,\bv)\le 
O(\log n)\cdot \Inf_f[i].$ 
\end{claim}
\begin{proof}

Our plan is to show that for every point $x\in\{0,1\}^n$ with a certain property, we have
\begin{equation}\label{multiset1}
\big\{\BinarySearch\hspace{0.05cm}(f ,x,S ,\pi ),\hspace{0.05cm}
\BinarySearch\hspace{0.05cm}(f ,x^{(i)},S,\pi)\big\}
\end{equation}
as a multiset is the same as
\begin{equation}\label{multiset2}
\big\{\BinarySearch\hspace{0.05cm}(f ,x,S' ,\pi' ),\hspace{0.03cm}
\BinarySearch\hspace{0.05cm}(f ,x^{(i)},S',\pi')\big\}.
\end{equation}
It turns out that the property holds for most points in $\{0,1\}^n$. 
  The lemma then follows.

To describe the property we let $m=|S|=|S'|$ and let $k=\pi^{-1}(i)$ (with $\pi'(k)=n+1$). We~let $J \subseteq [0:m]$ denote the  set of indices taken by variables $\ell$ and $r$
  (see Figure~\ref{fig:binarysearch}~for settings~of~$\ell$ and $r$) in an execution of $\BinarySearch$  
  along a path of length $m$ that outputs the $k$th edge at the end. 
\red{For example, ignoring the rounding issue, $J$ always contains 
  $0,m$ and $m/2$: these are indices of the first three points that 
  binary search examines. It contains $3m/4$ if $k> m/2$, or $m/4$ if $k\le m/2$, so on and so forth.
The set $J$ also always contains $k-1$ and $k$: these are indices of the last two points
  that binary search examines before returning the $k$th edge.}
  
Now we describe the property.
Given $x\in \{0,1\}^n$ we let
$x = x_0, \dots, x_{m} = x^{(S)}$ with $\smash{x_{t} = x_{t-1}^{(\pi(t))}}$ for all $t\in [m]$.
We let $\calC(x)$ be the indicator of the condition that:
\begin{align}
f(x_j) = f(x_j^{(i)}), \quad\text{for all }j \in J. \label{eq:condition-x}
\end{align}
\ignore{We then define a set $J$ of $O(\log n)$ many indices in $[0:m]$ as follows.
Imagine that we run binary search on a path of length $m+1$ in which the first
  $k$ points have the same value and the last $m+1-k$ points have the same but opposite value.
So the binary search will find the $k$th edge as the only bichromatic edge at the end,
  and we use $J$ to denote the set of points queried during the process.
As a result, $J$ contains $0,m,\lfloor m/2\rfloor,\ldots$ as well as $k-1$ and $k$.
For each $j\in [0:m]$, we let $S_j$ denote the subset of $S$ that contains the first
  $j$ variables in $S$ with respect to the ordering $\pi$.
Our condition on $x\in \{0,1\}^n$ is that, letting $x_j=x^{(S_j)}$ for each $j\in [0:m]$,
\begin{equation}\label{condition1}
f(x_j)=f(x_j^{(i)}),\quad \text{for all $j\in J$.}
\end{equation}}
We show that $\bx\sim\{0,1\}^n$   satisfies $\calC(\bx)$ with high probability.
Because $\bx$ is drawn uniformly from $\{0,1\}^n$\ignore{(and $S_j$ is a fixed set)},
  $\bx_j$ defined above is also distributed uniformly for each $j \in J$ and thus,
  the probability that a specific $j\in J$ violates the condition above is at most $\Inf_f[i]$. It
  then follows from a union bound over $j\in J$ that
  the fraction of points that violate the condition $\calC(x)$ is at most $\Inf_f[i]\cdot O(\log n)$.

It suffices to prove that when $x\in \{0,1\}^n$ satisfies $\calC(x)$,
  the two multisets in (\ref{multiset1}) and (\ref{multiset2}) are the same.
To this end we write down the two paths in the multiset (\ref{multiset1}) that start with $x$ and $x^{(i)}$ as
$$
x_0,x_1,\ldots,x_m \quad \text{and}\quad y_0,y_1,\ldots,y_m 
$$
in which $x_t=x^{{(\pi(t)}}_{t-1}$ and $y_t=x_t^{(i)}$.
Similarly we write down the two paths for (\ref{multiset2}) as
$$
z_0,z_1,\ldots,z_m\quad \text{and}\quad w_0,w_1,\ldots,w_m,
$$
in which we have $z_t=x_t$ for all $t<k$ and $z_t=y_t$ for all $t\ge k$;
  $w_t=y_t$ for all $t<k$ and $w_t=x_t$ for all $t\ge k$.
It follows from the property (\ref{eq:condition-x}) of $x$ that
\begin{equation}\label{condition2}
f(x_j)=f(y_j)=f(z_j)=f(w_j),\quad\text{for all $j\in J$.}
\end{equation}
Since $0,m\in J$ we have that $f(x_0)=f(x_m)$ implies the same 
  holds for $y,z$ and $w$ in which case (\ref{multiset1}) and 
  (\ref{multiset2}) are trivially the same since they all return $\nil$.
So we assume below that $f(x_0)\ne f(x_m)$ and thus, all four binary searches return a variable
  and a bichromatic edge.

Next, since $k-1,k\in J$ we have 
$$
f(x_{k-1})=f(x_k)=f(y_{k-1})=f(y_k)=f(z_{k-1})=f(z_k)=f(w_{k-1})=f(w_k).
$$
As a result, the $k$th edge is not bichromatic in all four paths and thus, 
  during each run of binary search, $k$ is removed from the interval $[\ell:r]$ (see Figure~\ref{fig:binarysearch}) after
  a certain number of rounds.~Moreover, it follows from the definition of $J$ and (\ref{condition2})
  that in all four runs of binary search, the values of $\ell$ and $r$ are the same
  at the moment when $k$ is removed from consideration (i.e., at the first time when either
  $\ell$ or $r$ is updated so that $k\notin [\ell:r]$.
We consider two cases for the values of $\ell$ and $r$.
\begin{flushleft}\begin{enumerate}
\item $\ell>k$: In this case, $\BinarySearch\hspace{0.05cm}(f,x,S,\pi)$ continues to search
  on the path $x_\ell,\ldots,x_r$ and $\BinarySearch\hspace{0.05cm}(f,x^{(i)},S',\pi')$
  continues to search on the path $w_{\ell},\ldots,w_r$ which is the same as $x_\ell,\ldots,x_r$
  given that $\ell>k$.
As a result, their outputs are the same. Similarly we have that 
  $\BinarySearch\hspace{0.05cm}(f,x^{(i)},S,\pi)$ is the same as 
  $\BinarySearch\hspace{0.05cm}(f,x,S',\pi')$ in this case. See Figure~\ref{fig:binary-search-2} for example executions.
\item $r<k$: In this case, $\BinarySearch\hspace{0.05cm}(f,x,S,\pi)$ continues to search
  on the path $x_\ell,\ldots,x_r$ and $\BinarySearch\hspace{0.05cm}(f,x ,S',\pi')$
  continues to search on the path $z_{\ell},\ldots,z_r$ which is the same as $x_\ell,\ldots,x_r$
  given that $r<k$.
As a result, their outputs are the same. Similarly we have that 
  $\BinarySearch\hspace{0.05cm}(f,x^{(i)},S,\pi)$ is the same as 
  $\BinarySearch\hspace{0.05cm}(f,x^{(i)},S',\pi')$ in this case. See Figure~\ref{fig:binary-search-1} for example executions.
\end{enumerate}\end{flushleft}
As a result, the two multisets are the same when $x$ satisfies the condition $\calC(x)$.
\end{proof}

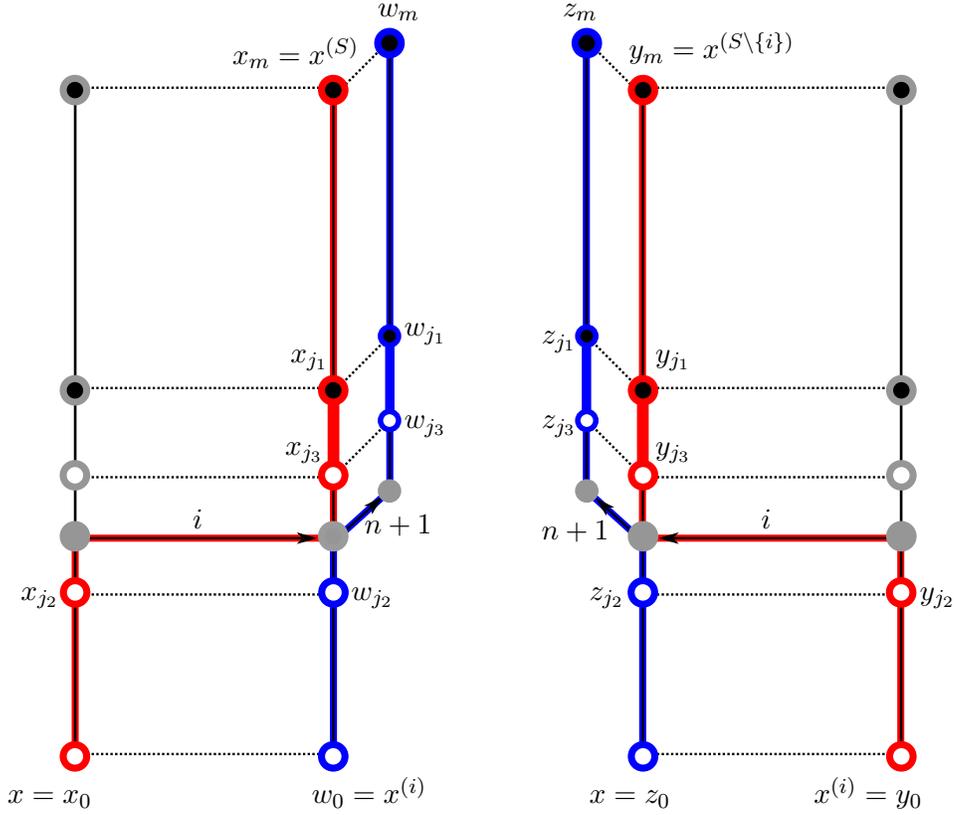
\begin{figure}
\centering
\begin{picture}(400, 300)
\put(0,0){\includegraphics[width=.8\linewidth]{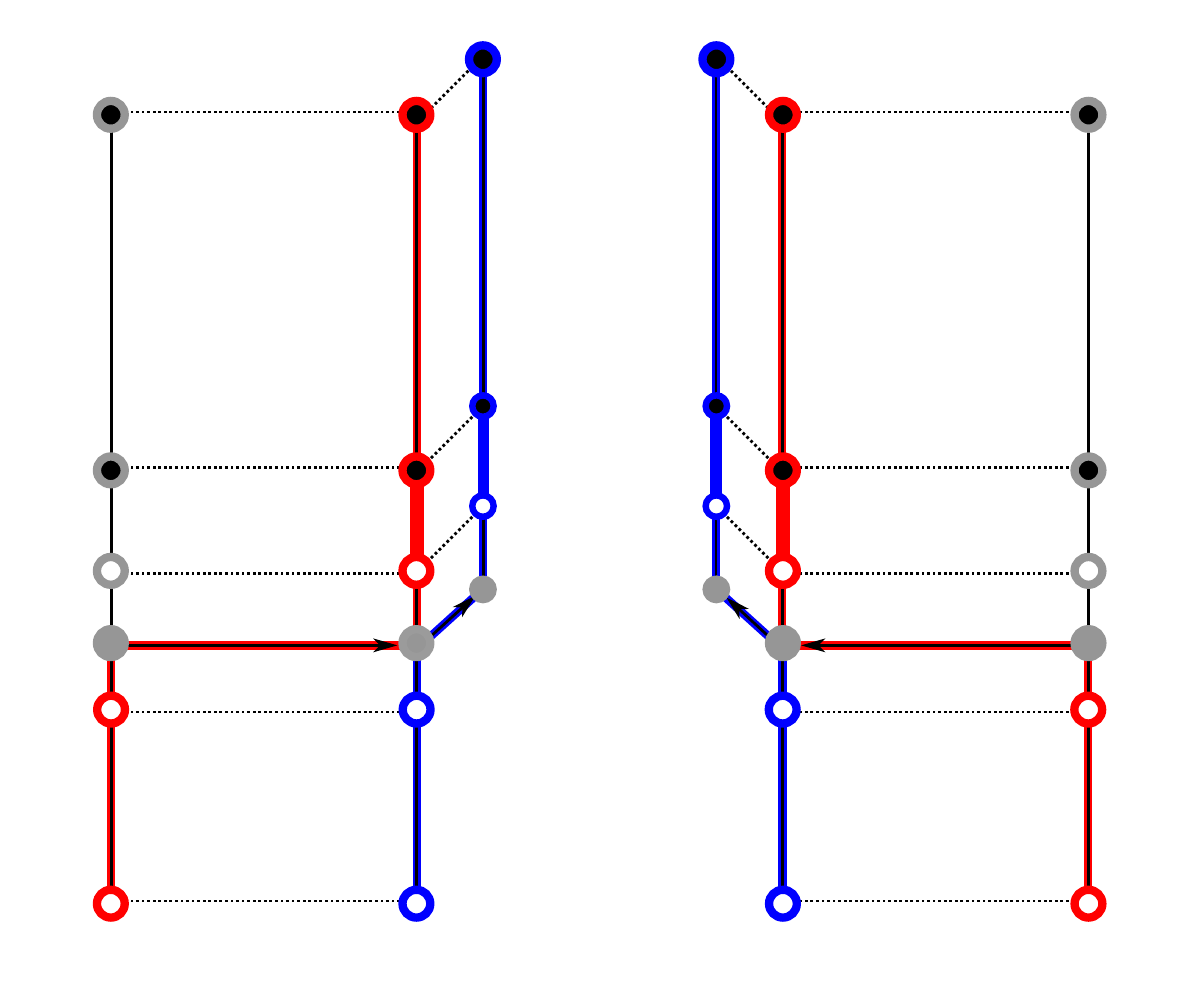}}
\put(10, 8){$x=x_0$}
\put(125, 8){$w_0 = x^{(i)}$}
\put(95, 288){$x_m = x^{(S)}$}
\put(150, 306){$w_m$}
\put(117, 175){$x_{j_1}$}
\put(160, 185){$w_{j_1}$}
\put(15, 85){$x_{j_2}$}
\put(140, 85){$w_{j_2}$}
\put(160, 150){$w_{j_3}$}
\put(115, 140){$x_{j_3}$} 
\put(80, 112){$i$}
\put(145, 110){$n+1$}

\put(230, 8){$x = z_0$}
\put(315, 8){$x^{(i)} = y_0$}
\put(245, 290){$y_m = x^{(S \setminus \{i\})}$}
\put(355, 85){$y_{j_2}$}
\put(255, 140){$y_{j_3}$}
\put(255, 175){$y_{j_1}$}
\put(230, 85){$z_{j_2}$}
\put(212, 182){$z_{j_1}$}
\put(212, 150){$z_{j_3}$}
\put(220, 306){$z_m$}
\put(295, 112){$i$}
\put(212, 108){$n+1$}
\end{picture}
\caption{Example executions of $\BinarySearch(f, x, S, \pi)$ and $\BinarySearch(f, x^{(i)}, S', \pi')$ on the left-hand side, and executions of $\BinarySearch(f, x, S', \pi')$ and $\BinarySearch(f, x^{(i)}, S, \pi)$ on the right-hand side, assuming that $\calC(x)$ is satisfied, and corresponding to the case when $k \leq \ell$. Queries made only during executions of $\BinarySearch(f, x, S, \pi)$ and $\BinarySearch(f, x^{(i)}, S, \pi)$ are displayed by red dots, and the corresponding paths considered are outlined in red; queries made only during executions of $\BinarySearch(f, x, S', \pi')$  and $\BinarySearch(f, x^{(i)}, S', \pi')$ are displayed by blue dots, and the corresponding paths considered are outlined in blue. Points are filled in with black if $f$ evaluates to $1$, and points which are not filled in if $f$ evaluates to $0$. Dotted lines indicates that condition $\calC(x)$ or the fact that $n+1$ is a dummy variable implies points evaluate to the same value under $f$. From the above executions, it is clear to see that $\BinarySearch(f, x, S, \pi)$ on the left-hand side considers the path (drawn in red) between $x_{j_3}$ and $x_{j_1}$, and $\BinarySearch(f, x^{(i)}, S', \pi')$ considers the path between $w_{j_3}$ and $w_{j_1}$ (drawn in blue); since variable $n+1$ represents a dummy variable, $f$ has the same evaluation on both of these paths, so both output the same variable. Similarly, $\BinarySearch(f, x^{(i)}, S, \pi)$ on the right-hand side considers the path (drawn in red) between $y_{j_3}$ and $y_{j_1}$, and $\BinarySearch(f, x^{(i)}, S', \pi')$ considers the same path between $z_{j_3}$ and $z_{j_1}$ (drawn in blue); as a result of $n+1$ being a dummy variable, both output the same variable.}\label{fig:binary-search-2}
\end{figure}

\begin{figure}
\centering
\begin{picture}(400, 300)
\put(0,0){\includegraphics[width=.8\linewidth]{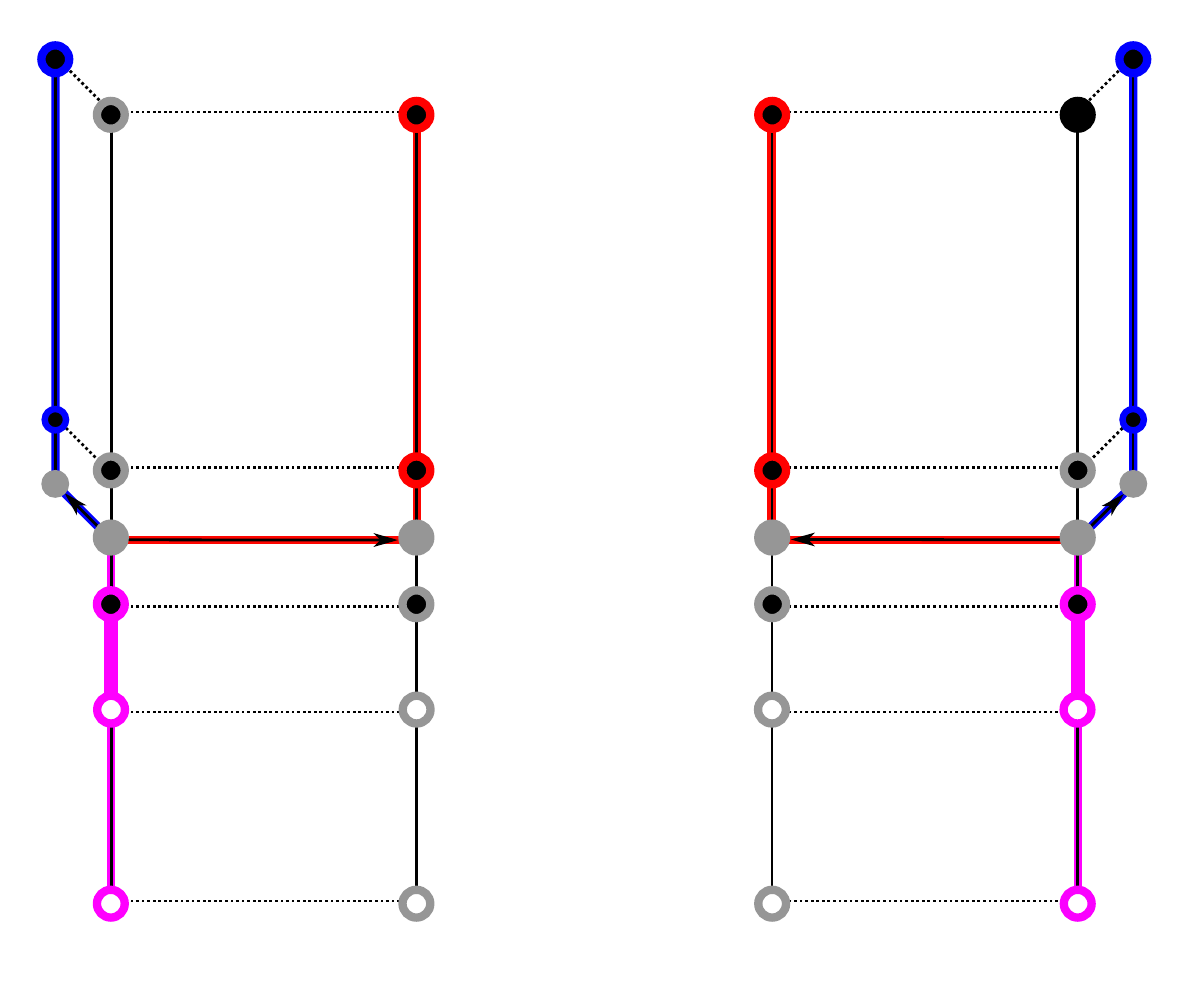}}
\put(10, 8){$x=x_0=z_0$}
\put(125, 8){$x^{(i)}$}
\put(130, 290){$x_m = x^{(S)}$}
\put(10, 308){$z_m$}
\put(143, 163){$x_{j_1}$}
\put(0, 180){$z_{j_1}$}
\put(-15, 85){$x_{j_2} = z_{j_2}$}
\put(-15, 120){$x_{j_3} = z_{j_3}$}
\put(80, 147){$i$}
\put(0, 140){$n+1$}

\put(243, 8){$x$}
\put(310, 8){$x^{(i)} = y_0 = w_0$}
\put(240, 290){$y_m = x^{(S \setminus \{i\})}$}
\put(355, 308){$w_m$}
\put(353, 85){$y_{j_2} = w_{j_2}$}
\put(353, 120){$y_{j_3} = w_{j_3}$}
\put(370, 180){$w_{j_1}$}
\put(227, 163){$y_{j_1}$}
\put(295, 147){$i$}
\put(357, 145){$n+1$}
\end{picture}
\caption{Example executions of $\BinarySearch(f, x, S, \pi)$ and $\BinarySearch(f, x, S', \pi')$ on the left-hand side, and executions of $\BinarySearch(f, x^{(i)}, S, \pi)$ and $\BinarySearch(f, x^{(i)}, S', \pi')$ on the right-hand side, assuming that $\calC(x)$ is satisfied, and corresponding to the case when $k > r$. Queries made only during executions of $\BinarySearch(f, x, S, \pi)$ and $\BinarySearch(f, x^{(i)}, S, \pi)$ are displayed by red dots, and the corresponding paths considered are outlined in red; queries made only during executions of $\BinarySearch(f, x, S', \pi')$  and $\BinarySearch(f, x^{(i)}, S', \pi')$ are displayed by blue dots, and the corresponding paths considered are outlined in blue; queries which are made during both are displayed with purple dots, and the intersection of the paths considered in both are purple. Similarly to Figure~\ref{fig:binary-search-2}, points filled in evaluate to $1$ under $f$, and points which are not filled in evaluates to $0$ under $f$. Dotted lines implies points evaluate to the same value under $f$. Note that $\BinarySearch(f, x, S, \pi)$ considers the path (drawn in purple) between $x_{j_2}$ and $x_{j_3}$, and $\BinarySearch(f, x, S', \pi')$ considers the same path between $z_{j_2}$ and $z_{j_3}$; thus, both output the same variable. Similarly, $\BinarySearch(f, x^{(i)}, S, \pi)$ considers the path (drawn in purple) between $y_{j_2}$ and $y_{j_3}$, and $\BinarySearch(f, x^{(i)}, S', \pi')$ considers the same path between $w_{j_2}$ and $w_{j_3}$; as a result, both output the same variable.}\label{fig:binary-search-1}
\end{figure}

Claim \ref{binaryclaim} gives the following corollary using a union bound:

\begin{corollary}\label{testcorollary}
Let $i\in S\subseteq [n]$ and $\pi$ be an ordering of $S $. 
Let $S' = \Sub(S, i)$
  and $\pi'$ be the ordering of $S'$ obtained from $\pi$ by
  substituting $i$ with $n+1$.
Then we have  
$$
\dtv\Big(\text{$\TestAB\hspace{0.05cm}(f ,S,\pi,\xi),\hspace{0.05cm}\TestAB\hspace{0.05cm}(f ,S',\pi',\xi)$}\Big)\le 
O\left(\frac{\log^\red{5}n}{\xi}\right)\cdot \Inf_f[i].
$$
\end{corollary}
\begin{proof}
We use the following coupling to run $\TestAB\hspace{0.05cm}(f ,S,\pi,\xi)$
  and $\TestAB\hspace{0.05cm}(f ,S',\pi',\xi)$ in parallel.
  
For each round of $\TestAB$  we first flip a fair coin and draw a subset
 $\bT$ of $S$ of the size indicated by the coin uniformly.
Then we couple the binary search on $\bx\sim \{0,1\}^n$ and $\bT$ 
  and the binary search on $\bz\sim \{0,1\}^n$ and $\Sub(\bT,i)$ using the best coupling between them.

  
  
It then follows from Claim \ref{binaryclaim} and a union bound over the $ \log^\red{4}n/\xi$ 
  rounds that the probability of this coupling of 
  $\TestAB\hspace{0.05cm}(f ,S,\pi,\xi)$ and $\TestAB\hspace{0.05cm}(f,S',\pi',\xi)$
  returning different results is at most
$$
(\log^\red{4} n/\xi)\cdot \Inf_f[i]\cdot O(\log n).
$$
This finishes the proof of the corollary.
\end{proof}

Now we prove Lemma \ref{lem:main}.

\begin{proof}[Proof of Lemma \ref{lem:main}]
Let $m=|S|=|S'|$.
For each $j\in [m]$, let $\bX_j$ denote the output of the $j$th call to
  $\TestAB$ in $\Prune\hspace{0.05cm}(f,S_0,\pi,\xi)$ with $\bX_j$ set to $\nil$ by
  default if the procedure terminates before the $j$th call.
Similarly we use $\bY_j$ to denote the output of the $j$th call in $\Prune\hspace{0.05cm}(f,S_0',\pi',\xi)$.
Let $\bX=(\bX_1,\ldots,\bX_m)$ and $\bY=(\bY_1,\ldots,\bY_m)$.
Then $\bX=\bY$ implies that $\bS=\Prune\hspace{0.05cm}(f,S_0,\pi,\xi)$ is the same
  as $\bS'=\Prune\hspace{0.05cm}(f,S_0',\pi',\xi)$.
As a result, it suffices to show that
$$
\dtv(\bX,\bY)\le \frac{m\hspace{0.05cm}\red{\log^5 n}}{\xi}\cdot \Inf_f[i].
$$

To this end, we first note that by Corollary~\ref{testcorollary}
  the total variation distance between $\bX_1$ and $\bY_1$
  is at most $\beta:=O(\log^5 n/\xi)\cdot \Inf_f[i]$.
On the other hand, note that if the outputs from the first $\ell-1$ calls in 
  $\Prune\hspace{0.05cm}(f,S_0,\pi,\xi)$ and $\Prune\hspace{0.05cm}(f,S_0',\pi',\xi)$
  are the same, say $a_1,\ldots,a_{\ell-1}$, then
  before the $\ell$th call,
  the set $\bS$ in the former still contains $i$ and the $\bS'$ in the latter 
  can be obtained by substituting its $i$ with $n+1$.
It follows from Corollary \ref{testcorollary}
  that, for any $\ell>1$ and any $a_1,\ldots,a_{\ell-1}$, 
  the total variation distance between the distribution of $\bX_\ell$
  conditioning on $\bX_1=a_1,\ldots,\bX_{\ell-1}=a_{\ell-1}$ and 
  the distribution of $\bY_\ell$ conditioning on
  $\bY_1=a_1,\ldots, \bY_{\ell-1}=a_{\ell-1}$ is  
  also at most $\beta$.
We prove that these properties together 
  imply that $\dtv(\bX,\bY)\le m\beta$,
  from which the lemma follows.\footnote{We 
  suspect that this is probably known in the literature but were not able to find a reference.}
  
For this purpose we use the following coupling of $\bX$ and $\bY$.
First we use the best coupling for the distribution of $\bX_1$ and the
  distribution of $\bY_1$ to draw $(\ba_1,\bb_1)$.
Then we draw $(\ba_2,\bb_2)$ from the the best coupling for the distribution of 
  $\bX_2$ conditioning on $\bX_1=\ba_1$ and the distribution of 
  $\bY_2$ conditioning on $\bY_1=\bb_1$. 
We then repeat until $(\ba_m, \bb_m)$ is drawn.
It follows from the description that the marginal distribution of $\ba=(\ba_1,\ldots,\ba_m)$
  is the same as $\bX$ and the marginal distribution of $\bb=(\bb_1,\ldots,\bb_m)$ is the same as $\bY$.
Moreover, we have 
\begin{align*}
\dtv(\bX,\bY)&\le 
\Pr\big[\ba\ne \bb\big]\\[0.3ex]&=\Pr\big[\ba_1\ne \bb_1\big]+\Pr\big[\ba_1=\bb_1 \wedge \ba_2\ne \bb_2\big]
+\cdots +\Pr\big[\ba_j=\bb_j\ \text{for $j<m$} \wedge \ba_m\ne \bb_m\big],
\end{align*}
which is at most $m\beta$ by the description of the coupling and properties of $\bX$ and $\bY$.
\end{proof}

\newcommand{\Sample}{\mathtt{Sample}}

\section{The Scores Lemma}\label{sec:score}

By definition when $f$ is $\eps$-far from unate, 
  $f(x\oplus a)$ is $\eps$-far from monotone for every $a\in \{0,1\}^n$.
This means that we can utilize the directed isoperimetric inequality of \cite{KMS15} to show the existence of relatively large and almost-regular bipartite graphs
  that consist of bichromatic edges (see Definition \ref{goodbipartite}
  and Lemma \ref{lem:KMS}). 
The goal of this section is to show that, using these bipartite graphs,
  there exist certain probability distributions over subsets of variables
  such that a set $\bS$ drawn from any of these distributions can be
  used to search for bichromatic edges via $\AESearch$ efficiently. 
  


To this end, we start by introducing three distributions $\calH_{\xi,m},
  \calD_{\xi,m}$ and $\calP_{i,m}$ in Section \ref{dist:sec}.
We then use them to define a \emph{score} for each variable $i\in [n]$   which aims to quantify the chance of finding a bichromatic edge along $i$
  using $\AESearch$ and a set $\bS$ drawn from some of those distributions.
Finally we prove the Scores Lemma in Section \ref{sec:scorelemma},
  which shows that the sum of scores over $i\in [n]$ is large
  when $f$ is $\eps$-far from unate and has total influence $O(\sqrt{n})$.


\subsection{Distributions $\calD_{\xi,m}, \calH_{\xi,m}$ and $\calP_{i,m}$ and the definition of scores}\label{dist:sec}

We start by defining two distributions $\calD_{\xi,m}$ and $\calH_{\xi,m}$.

\begin{definition}
Given $\xi \in(0,1)$ and $m:1\le m\le n$,
  we let $\calD_{\xi,m}$ denote the following distribution supported on subsets of $[n]$:
  $\bS \sim \calD_{\xi,m}$ is drawn by first sampling a subset $\bS_0$ of $[n]$ of size $m$
  and an ordering $\bpi$ of $\bS_0$ uniformly at random.
We then call $\Prune\hspace{0.04cm}(f,\bS_0,\bpi,\xi)$
  to obtain $\bS$.

Similarly, let $\calH_{\xi,m}$ denote the following distribution supported on subsets~of $[n+1] $:
  $\bS\sim \calH_{\xi,m}$ is drawn by first sampling a subset $\bS_0$~of $[n+1] $ of size $m$
  with $n+1$ $\in \bS_0$ and an ordering $\bpi$ of~$\bS_0$ uniformly at random.
We then call $\Prune\hspace{0.04cm}(f,\bS_0,\bpi,\xi)$ to obtain $\bS$.
 Notice that as $n+1$ is just a placeholder, we always have $n+1\in \bS\sim \calH_{\xi,m}$. 
\end{definition}

As it will become clear later, our unateness tester will sample 
  subsets according to the distribution $\calD_{\xi,m}$ and use them to 
  find an edge violation to unateness when $f$ is far from unate.
While this section is mainly concerned about $\calH_{\xi,m}$, it will 
  only be used in the analysis to help us understand how good those samples from $\calD_{\xi,m}$ are 
  in terms of revealing an edge violation to unateness.
  
Let $\Lambda=\lceil 2 \log (n/\eps)\rceil$ in the rest of the paper.
Given $i\in [n]$ and $m:1\le m\le n-1$ we use $\calP_{i,m}$ to denote the uniform distribution
  over all size-$m$ subsets of $[n]\setminus \{i\}$.

Next we use $\calH_{\xi,m}$ and $\calP_{i,m}$ to define strong edges.



\begin{definition}[Strong edges]\label{def:strong-edge}
Let $e$ be a bichromatic edge of $f$ along variable $i\in [n]$.
We say~$e$~is \emph{$\ell$-strong}, for some integer $\ell\in [\Lambda]$,
if the following two conditions hold:
\begin{flushleft}\begin{enumerate}
\item For every $m\le n^{2/3}$ as a power of $2$ and every 
  $\xi=1/2^k$ with $\ell\le k\le \Lambda$,
  the edge $e$ is $\bS$-persistent 
  \emph{(}recall Definition \ref{def:persistency}\emph{)} with probability at least $1-\red{(1/\log n)}$ when $\bS\sim \calH_{\xi,m}$. 
\item 
 The edge $e$ is $\bS$-persistent with probability at least $1-\red{(1/\log n)}$ when $\bS\sim\calP_{i,\lceil \sqrt{n}/2^\ell\rceil}$.
\end{enumerate}\end{flushleft}
\end{definition}

\newcommand{\Score}{\textsc{Score}}
\def\PE{\text{PE}}

For each $i \in [n]$ and $\ell\in [\Lambda]$, we define
\[ 
\Score^{+}_{i,\ell} (f) = \frac{1}{2^n} \cdot \text{number of $\ell$-strong monotone edges
  along variable $i$}.
\]
We analogously define $\Score_{i,\ell}^{-}(f)$ for anti-monotone edges along variable $i$. Finally we define
\begin{align}
\Score_i^{+}(f) &= \max_{ \ell\in [\Lambda]} \left\{ \Score^{+}_{i,\ell} (f) \cdot\frac{1}{2^\ell}\right\},  \label{eq:score-def}
\end{align}
and we analogously define $\Score^{-}_i(f)$.

\subsection{The Scores Lemma}\label{sec:scorelemma}

We state the Scores Lemma:

\begin{lemma}[The Scores Lemma]\label{lem:scores-lemma}
Let 
 $f \colon \{0,1\}^n \to \{0,1\}$ be a Boolean function that is $\eps$-far from unate with 
   total influence $\I_f < 6 \sqrt{n}$. Then we have
\[ \sum_{i\in [n]} \min \Big\{ \Score^+_i(f),\hspace{0.03cm} \Score_i^{-} (f) \Big\} \ge \Omega\left(\frac{\eps^2}{\Lambda^8}\right). \]
\end{lemma}

We note that our Scores Lemma above looks very similar to  Lemma 4.3 from \cite{CWX17b}. 
Thus  the proof  follows a similar trajectory.
The main difference is that we are varying the distributions from which the set $\bS$ of variables is drawn.
Compared to \cite{CWX17b}  we not only consider the~quality of $\bS$ drawn from the $\calP$
  distribution in the definition of strong edges but also those drawn from $\calH_{\xi,m}$
  with a number of possible combinations of $\xi$ and $m$ in the indicated range. 
This makes the proof of the lemma slightly more involved than that of Lemma 4.3 in \cite{CWX17b}.

We prove Lemma~\ref{lem:scores-lemma} by proving the following simpler version, which avoids the minimum. 
\begin{lemma}\label{lem:helper-structural}
Assume that $f$ is 
  $\eps$-far from {monotone} and satisfies $\I_f < 6 \sqrt{n}$. 
Then we have
\[ \sum_{i\in [n]} \Score_{i}^{- }(f) \ge \Omega\left(\frac{\eps^2}{\Lambda^8}\right). \]
\end{lemma}

\begin{proof}[Proof of Lemma~\ref{lem:scores-lemma} assuming Lemma~\ref{lem:helper-structural}]
Let $f \colon \{0,1\}^n \to \{0,1\}$ be a Boolean function which is $\eps$-far from unate. We let $a \in \{0,1\}^n$ be defined by setting, for each $i \in [n]$, 
\[ a_i =\left\{ \begin{array}{cl} 0 & \text{\ if $\Score_{i}^{+}(f) \geq \Score_{i}^{-}(f)$} \\[0.3ex] 
					    1 & \text{\ otherwise} \end{array} \right..\]
Consider the function $g \colon \{0,1\}^n \to \{0,1\}$ defined using $f$ by $g(x) = f(x \oplus a)$. We note that  $\I_f=\I_g$ and $g$ is $\eps$-far from unate and thus, $\eps$-far from monotone. So Lemma \ref{lem:helper-structural} implies that  
\begin{equation}\label{eqeq1}
 \sum_{i\in [n]} \Score_{i}^{-}(g) \ge \Omega\left(\frac{\eps^2}{\Lambda^8}\right).
\end{equation}
Finally, we claim that out choice of $s$ implies that
\begin{equation}\label{eqeq2}
\min\big\{ \Score_{i}^{+ }(f) ,\hspace{0.03cm} \Score_{i}^{- }(f) \big\} = \Score_{i}^{- }(g).
\end{equation}
This can be observed by checking that (1) the distributions $\calH_{\xi,m}$ defined using $f$ and $g$
  are exactly the same; 
  and (2) a point $x\in \{0,1\}^n$ is $S$-persistent for some $S\subseteq [n+1]$ 
  in $f$ if and only if $x\oplus a$ is
  $S$-persistent in $g$.
As a result, an edge $(x,x^{(i)})$ is $\ell$-strong in $f$ if and only if $(x\oplus a,x^{(i)}\oplus a)$
  is $\ell$-strong in $g$ but of course whether they are monotone or anti-monotone may change depending on $a_i$.
(\ref{eqeq2}) follows from these observations and Lemma \ref{lem:scores-lemma} follows
  from (\ref{eqeq1}) and (\ref{eqeq2}).   
\end{proof}

Before proving Lemma \ref{lem:helper-structural},
  we need a definition and a key technical lemma from \cite{KMS15}.


\begin{definition}\label{goodbipartite}
Given a Boolean function $f:\{0,1\}^n\rightarrow \{0,1\}$, we write 
  $G_f^-$ to denote its bipartite graph of anti-monotone edges:
\begin{flushleft}\begin{enumerate}
\item Vertices on the LHS of $G_f^-$ correspond to points $x\in \{0,1\}^n$
  with $f(x)=1$ and vertices on the RHS correspond to points $y\in \{0,1\}^n$ 
  with $f(y)=0$;
\item $(x,y)$ is an edge in $G_f^-$ if and only if $(x,y)$ is an anti-monotone edge in $f$.
\end{enumerate}\end{flushleft}

Let $G=(U,V,E)$ be a subgraph of $G_f$,  
where $U$ is a set of points $x$ with
$f(x) = 1$, $V$ is a set~of points $y$ with $f(y) = 0$, and $E$ 
  consists of all anti-monotone edges between $U$ and $V$ (i.e., 
  $G$ is~the induced subgraph of $G_f^-$ on $(U,V)$. 
We say that $G$ is \emph{right-$d$-good} for some positive integer $d$  
  if the degree of every $y\in V$ lies in $[d:2d]$
and the degree of every $x\in U$ is at most $2d$;
We say that $G$ is \emph{left-$d$-good} if the degree of every $x\in U$ lies in
  $[d:2d]$ and the degree of every $y\in V$ is at most $2d$.
\end{definition}

\begin{lemma}[Lemma 7.1 in \cite{KMS15}]\label{lem:KMS}
If $f:\{0,1\}^n\rightarrow \{0,1\}$ is $\eps$-far from monotone, 
 $G_f^-$ contains a bipartite  subgraph $G = (U, V, E)$ induced on $(U,V)$ that 
  satisfies one of the following conditions:
\begin{flushleft}\begin{enumerate}
\item $G$ is left-$d$-good for some positive integer $d$ and $\sigma = |U|/2^n$
  satisfies
\begin{equation}\label{hehehehe1}
\sigma^2 d=\Theta\left(\frac{\eps^2}{\log^4 n}\right).
\end{equation}
\item $G$ is right-$d$-good for some positive integer $d$ and $\sigma=|V|/2^n$
  satisfies (\ref{hehehehe1}).
\end{enumerate}\end{flushleft}
\end{lemma}

We note that each vertex in $G_f^-$ has degree at most $n$.
As a result, the two parameters $d$ and $\sigma$ in Lemma \ref{lem:KMS}
  always satisfy that $1\le d\le n$ and 
\begin{equation}\label{hehehehe2}
1\ge \sigma\ge \Omega\left(\frac{\eps}{\sqrt{n}\log^2n}\right).
\end{equation}


\begin{proof}[Proof of Lemma~\ref{lem:helper-structural}]
Let $f:\{0,1\}^n\rightarrow \{0,1\}$ be a function that is $\eps$-far from monotone
  with total variance $\I_f\le 6\sqrt{n}$.
It follows from Lemma \ref{lem:KMS} that there is a subgraph $G=(U,V,E)$ of $G_f^-$
  that satisfies one of the two conditions in Lemma \ref{lem:KMS}.
Below we assume without loss of generality that $G$ is left-$d$-good and $\sigma=|U|/2^n$ satisfies (\ref{hehehehe1});
  the proof for $G$ being right-$d$-good is symmetric.

In the rest of the proof we set $\ell$ to be the positive integer such that
$$
\frac{1}{2^{\ell}}< \frac{\sigma}{\Lambda^4}\le 
\frac{1}{2^{\ell-1}}.
$$
So $\ell\in [\Lambda]$ using (\ref{hehehehe2}).
Our goal is to show that at least half of edges in $G$ are $\ell$-strong.
As a result,
$$
\sum_{i\in [n]} \Score_{i}^-(f)
\ge \sum_{i\in [n]} \Score_{i,\ell}^-(f)\cdot \frac{1}{2^\ell} 
\ge \frac{\Omega(|E|)}{2^n}\cdot \frac{1}{2^\ell}
=\Omega(\sigma d)\cdot \frac{1}{2^\ell}
=\Omega\left(\sigma d\cdot \frac{\sigma}{\Lambda^4} \right)
=\Omega\left(\frac{\eps^2}{\Lambda^8}\right).
$$  
The fact that at least half of edges in $G$ are 
  $\ell$-strong follows directly from the next two claims:\vspace{0.07cm}
  
\begin{claim}\label{simpleclaim}
%
At least $(1-o(1))$-fraction of edges in $G$
  satisfy the first condition of being  $\ell$-strong.
\end{claim}

\begin{claim}\label{claim:FOCS}
At least $(1-o(1))$-fraction of edges in 
  $G$ satisfy the second condition\vspace{0.1cm}
  of being $\ell$-strong. 
\end{claim}


\red{The proof of Claim~\ref{claim:FOCS} follows from the arguments in Section~6.2 in \cite{CWX17b}. Specifically, given the definition of \emph{robust sets} for a bichromatic edge $e$ of a certain size in Definition~6.4 of \cite{CWX17b}, Claim~\ref{claim:FOCS} is equivalent to applying Lemma~6.11 and Lemma~6.12 twice. }

We prove Claim \ref{simpleclaim} in the rest of the proof.
To this end, let $m\le n^{2/3}$ and $\xi\le 1/2^\ell$ such that both $m$ and $1/\xi$ 
  are powers of $2$. We consider the quantity $\alpha$
  as the fraction of $e\in E$ such that $e$ is not $\bS$-persistent with probability at~least $1/\log n$ when
  $\bS\sim \calH_{\xi,m}$. 
Using $\alpha$ we have
\[ \Prx_{\substack{\be, \bS}}
\big[\be \text{ is  
  $\bS$-persistent}\big] \leq (1 - \alpha) + \alpha\left(1 - \frac{1}{\log n}\right)=1-\alpha/\log n, \]
where $\be$ is drawn uniformly from $E$ and $\bS\sim \calH_{\xi,m}.$ 
On the other hand, we consider the probability 
 \begin{align*}
 \Prx_{\substack{{\be,\bS}}}
 \big[\be \text{ is not 
   $\bS$-persistent}\big].\end{align*} 
By Lemma~\ref{lem:prune2}, as well as the fact that each vertex of $G$ is incident to at most $2d$ edges in $E$, for~at least $\smash{(1 - \exp (-\Omega(\log^2 n) ) )}$-fraction of $\smash{\bS \sim \calH_{\xi,m}}$, there are at most 
  $\smash{2{\xi} d \cdot 2^n}$ many edges which are not  $\bS$-persistent out of a total of at least $\sigma d \cdot 2^n$ edges in $E$. Therefore, we have
\begin{align*} 
\Prx_{\be,\bS}
\big[  \be \text{ is not 
  $\bS$-persistent}\big]
	&\le 
	\exp\left(-\Omega(\log^2 n)\right)+
	\left(1 - \exp\left(-\Omega(\log^2 n)\right) \right) 
	 \cdot O\left( {\xi}/{\sigma}\right),
\end{align*}
which is $O(1/\Lambda^4)$.
Combining these inequalities we have that $\alpha=O(\log n/\Lambda^4)$.
Claim \ref{simpleclaim} follows from a union bound over $O(\log n)\cdot \Lambda\le O(\Lambda^2)$
  many choices of the two parameters $m$ and $\xi$.
\end{proof}

\subsection{Bucketing scores}

We will now use standard grouping techniques to make Lemma~\ref{lem:scores-lemma} easier to use.

From (\ref{eq:score-def}), we say that $i \in [n]$ is of \emph{type}-$(s,t)$ for some $s,t \in [\Lambda]$ if 
\[ \Score_{i}^{+ } = \Score_{i,  {s}}^{+ } \cdot \frac{1}{2^{s}} 
  \quad\text{and}\quad \Score_{i}^{- } = \Score_{i,t}^{-} \cdot \frac{1}{2^{t}}. \]
From Lemma~\ref{lem:scores-lemma}, there exist $s, t \in [\Lambda]$ such that
\begin{align}
\sum_{\substack{i\in [n] \\ \text{type-$(s,t)$}}} \min \left\{ \Score_{i}^{+ }, \Score_{i}^{- }\right\} &\ge \Omega\left(\frac{\eps^2}{\Lambda^{10}}\right) . \label{eq:bucket-2}
\end{align}
Furthermore, we say a variable $i\in [n]$ has \emph{weight} $k$ for some positive integer $k$ if
\[ \frac{1}{2^{k }} < \min\left\{ \Score_{i}^{+ }, \Score_{i}^{- } \right\} \leq \frac{1}{2^{k-1}}. \]
Therefore, we have
\begin{align*}
\sum_{\substack{i\in [n] \\ \text{type-$(s,t)$}}} \min \left\{ \Score_{i}^{+ }, \Score_{i}^{- }\right\} &= \sum_{k\ge 1} \sum_{\substack{i\in [n] \\ \text{type-$(s,t)$} \\ \text{weight $k$}}} \min \left\{ \Score_{i}^{+ }, \Score_{i}^{- }\right\} \\
	&\leq \sum_{k\in [3\Lambda]}  \sum_{\substack{i\in [n] \\ \text{type-$(s,t)$} \\ \text{weight $k$}}} \min \left\{ \Score_{i}^{+ }, \Score_{i}^{- }\right\} + n \cdot \left(\frac{\eps}{n}\right)^{3},
\end{align*}
which implies by (\ref{eq:bucket-2}) that there exists some $h \in [3\Lambda]$ such that
\begin{align}
\sum_{\substack{i \in [n] \\ \text{type-$(s,t)$} \\ \text{weight $h$}}} \min\left\{ \Score_{i}^{+ }, \Score_{i}^{- } \right\} \ge \Omega\left(\frac{\eps^2}{\Lambda^{11}}\right).  \label{eq:bucket-3}
\end{align}
We let $$\calI^* = \big\{ i \in [n] : i \text{ is of type-$(s, t)$ and weight $h$}\big\},$$
and let $\calI$ be a subset of $\calI^*$ such that 
  $|\calI|$ is the largest power of $2$ that is not larger than $|\calI^*|$.
  
We summarize the above discussion in the following lemma.
\begin{lemma}\label{lem:score-summary}
Let $f \colon \{0,1\}^n \to \{0,1\}$ be $\eps$-far from unate with $\I_f < 6 \sqrt{n}$.
Then there are $s, t \in [\Lambda]$, $h \in [3\Lambda]$ and a set $\calI \subseteq [n]$
  such that $|\calI|$ is a power of $2$, $|\calI|/2^h=\Omega(\eps^2/\Lambda^{11})$ and 
  every $i\in \calI$ has  
\begin{equation}\label{hehehehehehe} \min\left\{\Score_{i,s}^{+ }\cdot \frac{1}{2^s},\hspace{0.06cm}\Score_{i,t}^{- }\cdot 
\frac{1}{2^t}\right\} \geq \frac{1}{2^h}  .
\end{equation}
\end{lemma}

\section{The Main Algorithm}\label{sec:mainalg}

We now describe the main algorithm for testing unateness.
The algorithm rejects a function $f$ only when an edge violation
  has been found. 
As a result, for its correctness it suffices to show that 
  when the input function $f$ is $\eps$-far from unate, the algorithm 
  finds an edge violation with probability at least $2/3$.
For convenience we will suppress  $\text{polylog}(n/\eps)$ factors using $\tilde{O}(\cdot)$
  in the rest of analysis. 

The main algorithm has four cases. Case $0$ is when the input function $f$ satisfies 
  $\I_f> 6\sqrt{n}$.
In this case an $\tilde{O}(\sqrt{n})$-query algorithm is known \cite{BCPRS17} (also
  see Lemma 2.1 of \cite{CWX17b}). 
  

From now on, we assume that $f$ is not only $\eps$-far from unate but also satisfies 
  $\I_f\le 6\sqrt{n}$. 
Then there are parameters $s,t\in [\Lambda]$ and $h\in [3\Lambda]$ and a set $\calI\subseteq [n]$
  with which Lemma \ref{lem:score-summary} holds for~$f$.
We~may assume that the algorithm knows $s,t,h$ and $|\calI|=2^\ell$~(by~trying~all possibilities, which~just incurs an addition factor of $O(\Lambda^4)$ in the query complexity).
We may further assume  
  without loss of generality that
  $s\ge t$ since the case of $s<t$ is symmetric. 


We consider the following three cases of $f$:

\begin{enumerate}
\item[] \textbf{Case 1:} $|\calI|/2^t\ge n^{2/3}$ and 
  and at least half of $i\in \calI$ satisfy
\begin{equation}\label{infcondition}
\Inf_f[i]\le \left(\frac{\eps^2}{\Lambda^{13}}\right)\cdot \frac{n^{1/3}}{|\calI|}\hspace{0.1cm};
\end{equation}
\item[] \textbf{Case 2:} $|\calI|/2^t \ge n^{2/3}$ and 
and at least half of $i\in \calI$ violate (\ref{infcondition}); and 
\item[] \textbf{Case 3:} $|\calI|/2^t\le n^{2/3}$.
\end{enumerate}

We prove the following two lemmas in Section \ref{algsec:case1} and \ref{algsec:case2} 
  which cover the first two cases.

\begin{lemma}\label{lemmaforcase1}
Let $s\ge t\in [\Lambda]$, $h\in [3\Lambda]$, and $\ell\in [\lfloor \log n\rfloor]$ with $2^{\ell}/2^t\ge n^{2/3}$.~There~is a $\tilde{O}(n^{2/3}/\eps^2)$-query algorithm with the following property. 
Given any Boolean function $f\colon\{0,1\}^n\rightarrow \{0,1\}$ that satisfies (i) Lemma \ref{lem:score-summary} holds for $f$
  with $s,t,h$ and a set $\calI\subseteq [n]$ with $|\calI|=2^{\ell}$;
  and (ii) at least half of $i\in \calI$ satisfy (\ref{infcondition}), 
  the algorithm finds an edge violation to unateness with probability at least $2/3$.
\end{lemma}

\begin{lemma}\label{lemmaforcase2}
Let $s\ge t\in [\Lambda]$, $h\in [3\Lambda]$, and $\ell\in [\lfloor \log n\rfloor]$  
  with $2^{\ell}/2^{t}\ge n^{2/3}$.
There is an~\mbox{algorithm} that makes $\smash{\tilde{O}(n^{2/3}/\eps^2)}$ queries and satisfies the following property. 
Given any Boolean function $f\colon\{0,1\}^n\rightarrow \{0,1\}$ that satisfies (i) Lemma \ref{lem:score-summary} holds for $f$
  with $s,t,h$ and a set $\calI\subseteq [n]$ with $|\calI|=2^{\ell}$
  and (ii) at least half of $i\in \calI$ violate (\ref{infcondition}), 
  the algorithm finds an edge violation to unateness with probability at least $2/3$.
\end{lemma}

Case 3 can be handled using an algorithm presented in \cite{CWX17b}.
We include its description and the proof of the following lemma 
  in Section~\ref{algsec:case3} for completeness.
  
\begin{lemma}\label{lemmaforcase3}
Let $s\ge t\in [\Lambda]$, $h\in [3\Lambda]$, and $\ell\in [\lfloor \log n\rfloor]$ with $2^{\ell}/2^{t}\le n^{2/3}$.
There~is~an~\mbox{algorithm} that makes $\tilde{O}(n^{2/3} + \sqrt{n}/\eps^2)$ queries and satisfies 
the following property. 
Given any Boolean func\-tion $f\colon\{0,1\}^n\rightarrow \{0,1\}$ that satisfies Lemma \ref{lem:score-summary}  
  with $s,t,h$ and a set $\calI\subseteq [n]$ of size $|\calI|=2^{\ell}$,
  the algorithm finds an edge violation of $f$ to unateness with probability at least $2/3$.
\end{lemma}

Theorem \ref{thm:main} follows by combining all these lemmas.

\section{The Algorithm for Case 1}\label{algsec:case1}

Let $s\ge t\in [\Lambda]$, $h\in [3\Lambda]$ and $\ell\in [\lfloor \log n\rfloor]$.
In Case 1 the input function $f\colon\{0,1\}^n\rightarrow \{0,1\}$~satisfies
  Lemma \ref{lem:score-summary} with parameters $s,
  t,h$ and $\calI\subseteq [n]$ of size $|\calI|=2^\ell$, with $|\calI|/2^t \geq n^{2/3}$.
At least half of the variables $i\in \calI$
  have low influence~as given in (\ref{infcondition}).
Let $i$ be such a variable. Then by (\ref{infcondition}),
$$
\left(\frac{\eps^2}{\Lambda^\red{13}}\right)\cdot \frac{n^{1/3}}{|\calI|}
>\Inf_f[i]\ge 2\cdot \Score_{i,s}^+
\ge \frac{2^s}{2^h}.
$$
Letting $\xi=1/2^s$ throughout this section, it follows from Lemma \ref{lem:score-summary} that
\begin{equation}\label{boundforxi}
\xi= \Omega\left(\frac{\Lambda^{13}}{\eps^2}\cdot \frac{|\calI|}{2^hn^{1/3}}\right)
=\Omega\left(\frac{\Lambda^{2}}{n^{1/3}}\right).
\end{equation}



\subsection{Informative sets}\label{sec:inform}

We start with the notion of \emph{informative sets}.
Note that we will have different notions of informative sets
  in different cases of the algorithm.
We use the same name because they serve similar purposes. 

Given $i\in [n]$ and a set $S\subseteq [n+1]$ we use 
  $\PE_i^+(S)$ to denote the set of $s$-strong monotone edges along variable 
  $i$ that are $S$-persistent.
We define $\PE_i^-(S)$ similarly for antimonotone edges.

\begin{definition}[Informative Sets]\label{def:informative}
A set $S \subseteq [n+1]$ is $i$-\emph{informative} for monotone edges if 
\begin{equation}\label{bulbul-mon}
\frac{|\emph{\PE}_i^+(S)|}{2^n}\ge \frac{\Score_{i,s}^+}{4}\ge \frac{2^{s-h}}{4} 
\end{equation}
and that $S \subseteq [n+1]$ is $i$-\emph{informative} for anti-monotone edges if
\begin{equation}\label{bulbul-antimon}
\frac{|\emph{\PE}_i^-(S)|}{2^n}\ge \frac{\Score_{i,t}^-}{4}\ge \frac{2^{t-h}}{4}.
\end{equation}
We simply say that $S\subseteq [n+1]$ is $i$-\emph{informative} if $S$ satisfies both (\ref{bulbul-mon}) and (\ref{bulbul-antimon}).
\end{definition}

\begin{lemma}\label{ff}
For each $i\in \calI$ and each positive integer $m\le n^{2/3}$ that is a power of $2$,  
  $\bS\sim \calH_{\xi,m}$~is $i$-informative with probability at least $1 - o(1)$.
\end{lemma}
\begin{proof}
We first show that $\bS\sim\calH_{\xi,m}$ satisfies (\ref{bulbul-mon}) with probability
  at least $1 - o(1)$.
The same argument works to show $\bS \sim \calH_{\xi, m}$ satisfies (\ref{bulbul-antimon}). The lemma then follows from a union bound.~To this end, let $\alpha$ be the probability of $\bS\sim\calH_{\xi,m}$ being
  $i$-informative for monotone edges.
We examine 
$$
\Prx_{\be,\bS} \big[\text{$\be$ is $\bS$-persistent}\big],
$$
where $\be$ is an $s$-strong monotone edge along variable $i$ drawn uniformly at random and 
  $\bS\sim\calH_{\xi,m}$.
It follows from the definition of strong edges that the probability is at least $1-1/\log n$.
On the other hand, we can also upperbound the probability using $\alpha$ (and the definition
  of $i$-informative sets) as
$ 
(1-\alpha)/4+\alpha.
$
Solving the inequality we get $\alpha\ge 1 - o(1)$.
\end{proof}

Next we introduce two new families of distributions that will help
  us connect $\calH_{\xi,m}$ with $\calD_{\xi,m}$.

\begin{definition}
Given $\xi \in(0,1)$, $m\colon1\le m\le n$ and $i\in [n]$,
  we let $\calD_{\xi,m,i}$ denote the following distribution supported on subsets of $[n]$:
  $\bS \sim \calH_{\xi,m,i}$ is drawn by first sampling a subset $\bS_0$ of $[n] $ of size $m$
  with $i\in \bS_0$
  and an ordering $\bpi$ of $\bS_0$ uniformly at random.
We then call $\Prune\hspace{0.04cm}(f,\bS_0,\bpi,\xi)$
  and set $\bS$ to be its~output.

Similarly,  $\calH_{\xi,m,i}$ denotes the following distribution supported on subsets of $[n+1]$:
  $\bS \sim \calH_{\xi,m,i}$ is drawn by first sampling a subset $\bS_0$ of $[n+1]\setminus \{i\}$ of size $m$
  with $n+1\in \bS_0$
  and an ordering $\bpi$ of $\bS_0$ uniformly at random.
We then call $\Prune\hspace{0.04cm}(f,\bS_0,\bpi,\xi)$
  and set $\bS$ to be its~output.
\end{definition}

Using the fact that the total variation distance between the $\bS_0$ used in $\calH_{\xi,m}$ (at the beginning 
  of the process)
  and the $\bS_0$ used in $\calH_{\xi,m,i}$ is at most $m/n$,
  we have $$\dtv\big(\calH_{\xi,m},\calH_{\xi,m,i}\big)\le m/n$$ and the following corollary
  from Lemma \ref{ff}.

\begin{corollary}\label{coro1}
For every $i \in \calI$ and every positive integer $m\le n^{2/3}$ as a power of $2$, we have that
  $\bS\sim \calH_{\xi,m,i}$ is $i$-informative with probability at least $1-o(1)$.
\end{corollary}

The next two lemmas allow us to draw random subsets and still obtain $i$-informative sets. They enable us to use techniques from \cite{CWX17b} for particular cases of our algorithm.

\begin{lemma}\label{lem:informative-for-antimon}
For every $i \in \calI$ we have $\bT \sim \calP_{i, \lceil \sqrt{n}/2^t \rceil}$ is $i$-informative for anti-monotone edges with probability at least $1-o(1)$.
\end{lemma}

\begin{proof}
Similarly to the proof of Lemma~\ref{ff}, we write $\alpha$ to denote the probability of $\bT\sim\calP_{i,\lceil \sqrt{n}/2^t \rceil}$ being
  $i$-informative for anti-monotone edges. We examine 
$$
\Prx_{\be,\bT} \big[\text{$\be$ is $\bT$-persistent}\big],
$$
where $\be$ is a $t$-strong anti-monotone edge along variable $i$ drawn uniformly and 
  $\bT\sim\calP_{i,\lceil \sqrt{n}/2^t\rceil}$.
It follows from the definition of strong edges that this probability is at least $1-1/\log n$.
On the other hand, we can also upperbound the probability using $\alpha$ (and the definition
  of $i$-informative sets for anti-monotone edges) as
$ 
(1-\alpha)/4+\alpha.
$
Solving the inequality we get $\alpha\ge 1-o(1)$.
\end{proof}

Similarly, we may conclude the analogous lemma for monotone edges, whose proof follows similarly to Lemma~\ref{lem:informative-for-antimon}.
\begin{lemma}\label{lem:informative-for-mon}
For every $i \in \calI$ we have that $\bS \sim \calP_{i, \lceil\sqrt{n}/2^s\rceil}$ is $i$-informative for monotone edges with probability at least $1 - o(1)$.
\end{lemma}

\subsection{Catching variables: Relating $\calD_{\xi,m}$ and $\calH_{\xi,m}$}

Now we focus on the variables in $\calI$ that satisfy (\ref{infcondition}).
To this end, we let $\calI^*$ be a subset of $\calI$ of size $\lceil |\calI|/2\rceil$
  such that all variables in $\calI^*$ satisfy (\ref{infcondition}).
Given that the algorithm knows the size of $\calI$, it also knows the 
  size of $\calI^*$ (though not variables within).
Next we use $m$ to denote the largest~power of $2$ that is at most 
$ \xi |\calI|/ n^{1/3}$. In other words, $m$ is the unique power of 2 satisfying
\begin{align}
\dfrac{\xi |\calI|}{2n^{1/3}} < m \leq \dfrac{\xi |\calI|}{n^{1/3}}.  \label{eq:def-m}
\end{align}
Given that $|\calI| \geq |\calI|/2^t \geq n^{2/3}$ and (\ref{boundforxi}), we have 
$m\gg 1$ and $m=
\Theta( {\xi |\calI|}/{n^{1/3}} )$.




\newcommand{\Caught}{\textsc{Caught}}

We now turn to analyzing the distribution $\calD_{\xi,m}$ with the $m$ defined above.

\begin{definition}[Catching Variables]
Let $i\in \calI^*$. We say that a set $S\subseteq [n]$ 
  \emph{catches} the variable $i$ if $i\in S$ and $\Sub(S,i)=(S\cup \{n+1\})\setminus \{i\}$ is $i$-informative (see Definition~\ref{def:informative}). We let 
\[ \Caught (S) = \big\{ i \in \calI^* : S \text{ catches } i\big\}. \]
\end{definition}

Intuitively, if we sample $\bS \sim \calD_{\xi,m}$ and $i \in \Caught (\bS)$, then 
  we have an upper bound for how many samples $\bx$ we need for $\AESearch(f,\bx,\bS\cup \{n+1\})$ 
  to reveal a bichromatic edge along $i$. 

\begin{claim}\label{trickyclaim}
For every $i \in \calI^*$, we have 
$$
\Prx_{\bS\sim \calD_{\xi,m,i}}\big[\bS\ \emph{\text{catches\ }} i\big]
\ge \Prx_{\bT\sim \calH_{\xi,m,i}}\big[\bT\ \emph{\text{is $i$-informative\ }}\big]-o(1).
$$
\end{claim}

\begin{proof}
We show the total variation distance between 
  $\bS\sim \calD_{\xi,m,i}$ and $\Sub(\bT,i)$ over $\bT\sim \calH_{\xi,m,i}$ is  
\begin{equation}\label{tttt}
O\left(\frac{m\log^8 n}{\xi}\cdot \Inf_f[i]\right)
=O\left(\frac{|\calI|\log^8 n}{n^{1/3}}\cdot \Inf_f[i]\right)
=o(1),
\end{equation}
given (\ref{infcondition}) and $i\in \calI^*$. The lemma follows from the observations that
  $\Sub(\bT,i)$ contains $i$ and when $\bT$ is $i$-informative, $\Sub(\bT,i)$ catches $i$.
  
To upperbound the total variation distance between $\bS\sim \calD_{\xi,m,i}$ and $\Sub(\bT,i)$ over $\bT\sim \calH_{\xi,m,i}$, we use the following coupling.
First we draw a subset $\bS_0$ of $[n]$ with $i\in \bS_0$ 
  and an ordering $\bpi$ of $\bS_0$ uniformly at random.
Then we set $\bS_0'=\Sub(\bS_0,i)$ and $\bpi'$ to be the ordering of $\bS_0'$
  obtained from $\bpi$ by replacing $i$ with $n+1$.
Finally we draw the output from the best coupling for $\Prune\hspace{0.04cm}(f,\bS_0,\bpi,\xi)$
  and $\Sub(\Prune\hspace{0.04cm}(f,\bS_0',\bpi',\xi),i)$.
The upper bound in (\ref{tttt}) follows directly from Lemma \ref{lem:main}. 
\end{proof}

\def\cc{\mathbf{c}}
\def\ff{\boldsymbol{f}}
\def\gg{\boldsymbol{g}}

\newcommand{\AlgorithmC}{\mathtt{AlgorithmCase1.1}}

\begin{figure}[t!]
\begin{framed}
\noindent Procedure $\AlgorithmC\hspace{0.04cm} (f)$

\begin{flushleft}
\noindent {\bf Input:} Query access to a Boolean function $f \colon \{0, 1\}^n \to \{0, 1\}$

\noindent {\bf Output:} Either ``unate,'' or two edges constituting an edge violation of $f$ to unateness.\begin{enumerate}
\item Repeat the following $O(1)$ times:
\item \ \ \ \ \ \ \ \ Draw $\bS\sim \calD_{\xi,m}$: First draw a size-$m$ 
  subset $\bS_0$ of $[n]$ and an ordering $\bpi$ of $\bS_0$\\
  \ \ \ \ \ \ \ \  uniformly at random and 
  then call $\Prune\hspace{0.04cm}(f,\bS_0,\bpi,\xi)$.
\item \ \ \ \ \ \ \ \ Repeat $q$ times, where $\smash{q=O\left(n^{2/3}\Lambda^{13}\big/\eps^2\right)}$:
\item \ \ \ \ \ \ \ \ \ \ \ \ \ \ \ \ Draw an $\bx\in \{0,1\}^n$ uniformly and run
  $\AESearch\hspace{0.04cm}(f,\bx,\bS\cup\{n+1\})$
\item \ \ \ \ \ \ \ \ Let $\bA$ be the set of $i\in [n]$ such that an anti-monotone edge along $i$ is found
\item \ \ \ \ \ \ \ \ Repeat $q$ times:
\item \ \ \ \ \ \ \ \ \ \ \ \ \ \ \ \ Draw an $\by\in \{0,1\}^n$ uniformly and run 
  $\AESearch\hspace{0.04cm}(f,\by,\bS\cup \{n+1\})$
\item \ \ \ \ \ \ \ \ Let $\bB$ be the set of $i\in [n]$ such that a monotone edge along variable $i$ is found
\item \ \ \ \ \ \ \ \ If $\bA\cap \bB\ne \emptyset$, output an edge violation of $f$ to unateness.
\item Output ``unate.''
\end{enumerate}
\end{flushleft}\vskip -0.14in
\end{framed}\vspace{-0.2cm}
\caption{Algorithm for Case 1.1}\label{fig:algcase1-1}
\end{figure}

\subsection{Algorithm for Case 1.1}\label{sec:case1-1}

There are two sub-cases in Case 1. Specifically, for the remainder of 
  Section \ref{sec:case1-1}, we assume that 
\begin{equation}\label{case11}
m\ge \frac{n^{1/3}\log^2 n}{2^t},
\end{equation}
and handle the other case in Case 1.2. We let
\begin{align}
r:= \frac{m|\calI^*|}{n}=\Omega(\log^{2}n), \label{eq:r-value}
\end{align}
the expected size of the intersection of a random size-$m$ subset of $[n]$ with $\calI^*$, where we used (\ref{case11}) and $|\calI^*|/2^t = \Omega(n^{2/3})$.
Note that both $m$ and $r$ are known to the algorithm. We prove Lemma \ref{lemmaforcase1} assuming (\ref{case11}) using $\AlgorithmC$ in Figure~\ref{fig:algcase1-1}, with the following query complexity.

 \begin{claim} 
The query complexity of $\AlgorithmC$ is $\tilde{O}(n^{2/3}/\eps^2)$.
 \end{claim} 
\begin{proof}
By Lemma~\ref{lem:prune2}, line 2 of $\AlgorithmC$ requires $\tilde{O}({m}/{\xi})$ queries. By (\ref{eq:def-m}), 
\[ \frac{m}{\xi} = O\left( \frac{\xi |\calI|}{\xi \cdot n^{1/3}}\right) = O\left( \frac{|\calI|}{n^{1/3}}\right)=O(n^{2/3})  \] using the trivial bound of $|\calI|\le n$.
The claim then follows from our choice of $q$ in the algorithm.
\end{proof}


The algorithm for Case 1.1 starts by sampling a 
  set $\bS\sim \calD_{\xi,m}$.
It then keeps drawing points $\bx$ uniformly at random to run $\AESearch\hspace{0.04cm}(f,\bx,\bS\cup \{n+1\})$
  to find bichromatic edges, with the hope to find an edge violation along 
  one of the variables in $\calI^*$.
We break lines 3--8 into the search of anti-monotone edges and the search of monotone edges separately only for 
  the analysis later; algorithm wise there is really no need to do so. 
(The reason why we use $\bS\cup \{n+1\}$ instead of $\bS$ in
  the algorithm will become clear in the proof of Lemma~\ref{lem:catching-lemma};
  roughly speaking, we need it to establish a connection between $\calD_{\xi,m}$ and $\calH_{\xi,m}$ so that
  we can carry the analysis on $\calH_{\xi,m}$ that has been done so far over to $\calD_{\xi,m}$.)  
On the one hand,
  recall from Lemma~\ref{lem:aesearch} that if $i\in S\subseteq [n]$ and 
  a bichromatic edge $e$ along variable $i$ is $\Sub(S,i)=(S\cup\{n+1\})\setminus \{i\}$-persistent,
  then running $\AESearch$ on $S\cup \{n+1\}$ and any of the two points of $e$ would 
  reveal $e$ with high probability.
On the other hand, if a set (e.g., $\Sub(S,i)$) is $i$-informative then it is persistent on
  a large fraction of edges along variable $i$.

Our first goal is to prove Lemma \ref{lem:catching-lemma}, which states that 
  $\bS\sim\calD_{\xi,m}$ catches many variables.

\begin{lemma}\label{lem:catching-lemma}
We have $|\Caught(\bS)|\ge r/6$ with probability $\Omega(1)$. 
\end{lemma}

\begin{proof}
Let $\alpha$ be the probability we are interested in:
\[ \alpha = \Prx_{\bS \sim \calD_{\xi,m}}\left[ \big|\Caught(\bS)\big| \geq \frac{r}{6}\hspace{0.02cm}\right]. \]
For each $i \in \calI^*$, as the $\bS_0$ drawn in $\calD_{\xi,m}$ 
  at the beginning contains $i$ with probability $m/n$, we have
\begin{align*}
\Prx_{\bS \sim \calD_{\xi,m}}\big[  \bS\ \text{catches $i$}\big] &\ge \frac{m}{n}\cdot  \Prx_{\bS \sim \calD_{\xi,m,i}}\big[ \bS\ \text{catches $i$} \big] \\
	&\geq \frac{m}{n} \cdot \left( \Prx_{\bT \sim \calH_{\xi,m,i} }\big[ \bT \text{ is $i$-informative} \big] - o(1)\right) \geq (1-o(1)) \cdot \frac{m}{n}.
\end{align*}
Furthermore, since $\bS \sim \calD_{\xi,m}$ is a subset of $\bS_0$
  drawn at the beginning which is a random subset of $[n]$ of
   size $m$, we have by Lemma~\ref{appendix1} that with probability at least $1 - \exp( -\Omega(r))$ 
   over the draw~of $\bS \sim \calD_{\xi,m}$ that $|\bS \cap \calI^*| \leq 4r$. 
Therefore, we have
\begin{align*}
(1-o(1)) r=(1-o(1)) \cdot \frac{m  |\calI^*|}{n} &\leq \sum_{i \in \calI^*} \Prx_{\bS \sim \calD_{\xi,m}}\big[\bS\ \text{catches $i$}\big] = \Ex_{\bS \sim \calD_{\xi,m}}\Big[ \big| \Caught(\bS)\big|\Big]\\
				&\leq m \cdot \exp\left( -\Omega(\log^2 n)\right) + (1-\alpha) \cdot \frac{r}{6} + \alpha \cdot 4r.
\end{align*}
Solving for $\alpha$ gives the desired bound of $\alpha=\Omega(1)$.
\end{proof}

Given Lemma~\ref{lem:catching-lemma}, a constant fraction of the intersection of $\bS$ and $\calI^*$ will be caught, and therefore, $\AESearch\hspace{0.05cm}(f,\bx, \bS \cup \{ n + 1\})$ will output a bichromatic edge along a variable from these caught coordinates for sufficiently many points $\bx$. 
In the rest of the proof, we fix $S$ to be a set that catches at least $r/6$ many variables in $\calI^*$,
  and prove in the rest of the proof that 
  during this loop, 
  an edge violation is found with probability $1-o(1)$. 

Given $S$, we write $\calJ\subseteq \calI^*$ to denote 
  the set of variables caught by $S$ with $|\calJ|\ge r/6$.
Then by definition we have that $\calJ\subseteq S$ and $\Sub(S,j)$ is $j$-informative for every $j\in \calJ$.

We start by showing that $\bA\cap \calJ$ is large with high probability.

\begin{lemma}\label{lem:a-cap-j}
We have $ |\bA\cap \calJ|\ge \Omega\left( 
{m2^t}/{n^{1/3}}\right)$ 
  with probability at least $1-o(1)$.
\end{lemma}
\begin{proof}
For each $j\in \calJ$, we let $X_j$ be the set of $x\in \{0,1\}^n$
  such that $\AESearch\hspace{0.04cm}(f,x,S\cup\{n+1\})$ returns an anti-monotone edge along 
  $h$ with probability at least ${2/3}$.
Then by definition we have
$$
|X_j|=\Omega\left(\frac{2^t}{2^h}\right)\cdot 2^n,
$$
and the $X_j$'s are disjoint by Corollary \ref{AESearchCorollary}
  so we have $r2^t/2^h=O(1)$.
To analyze $\bA\cap \calJ$,
  we break the rounds on line 3 into $\lceil m2^t/n^{1/3}\rceil$ many phases, each consisting of 
$$
 \left\lceil \frac{2^h}{r 2^t} \right\rceil \cdot \log^2 n
=O\left(\frac{2^h}{r 2^t}\cdot\log^2 n\right)
$$ many iterations of line 4 (using $r2^t/2^h=O(1)$).
The $q$ rounds we have are enough since
  $$
\left\lceil \frac{m2^t}{n^{1/3}}\right\rceil
\cdot O\left(\frac{2^h}{r 2^t}\cdot\log^2 n\right)
=O\left(\frac{m 2^h\log^2 n}{n^{1/3}r}\right)=O\left(n^{2/3}\Lambda^{13}/\eps^2\right) 
$$
using $r=\Omega(m|\calI|/n)$ and $|\calI|/2^h=\Omega(\eps^2/\Lambda^{11})$.
%
Also note $r=\Omega(m2^t/n^{1/3})$ using 
  $|\calI|\ge n^{2/3}2^t$.

At the beginning of each phase, either $\Omega(r)$ anti-monotone edges along different
  variables in~$\calJ$ have 
  already been found, and we are done, or the number of 
  variables in $\calJ\setminus \bA$ at the moment is at least
  $\Omega(r)$.
As a result, their union of $X_j$ for such $j$ is $\Omega(r 2^t/2^h)\cdot 2^n$.
Using the number of rounds in each phase, the probability 
  of not finding any new anti-monotone~edge along $\calJ$ during this phase is at most $1/\poly(n)$, 
  and it remains negligible
  even after a union bound over the number of phases.
The lemma follows from the fact that the number of phases is at least
  $\Omega(m2^t/n^{1/3})$.
\end{proof}

Fix an $A$ such that $C=A\cap \calJ$  satisfies the lower bound of Lemma~\ref{lem:a-cap-j}. 
We finally show that $C\cap \bB$ with high probability.
This finishes the proof of correctness in Case 1.

\begin{lemma}
We have that $C\cap \bB$ is not empty with probability at least $1-o(1)$.
\end{lemma}
\begin{proof}
For each $j\in C$, let $Y_j$ denote the set of $y\in \{0,1\}^n$ such that
  $\AESearch\hspace{0.04cm}(f,y,S\cup\{n+1\})$ returns a monotone edge along variable $j$
  with probability at least $2/3$.
Then we have  
$$
{|Y_j|=\Omega\left(\frac{2^s}{ 2^h}\right) \cdot 2^n.
} 
$$
Since they are disjoint, the fraction of points that $Y_j$, $j\in C$, cover 
  is at least 
$$
\Omega\left(\frac{2^s}{2^h}\cdot \frac{m2^t}{n^{1/3}}\right)\ge
\Omega\left(\frac{2^s}{2^h}\cdot \frac{\xi|\calI|}{n^{1/3}}\cdot \frac{2^t}{n^{1/3}}\right)
=\Omega\left(\frac{\eps^2}{n^{2/3}\Lambda^{11} }\right)
$$
using $s\ge t$ and $m2^t\ge n^{1/3}\log^2n$.
The lemma follows from our choice of $q$ in the algorithm.
\end{proof}

\ignore{\subsection{Algorithm for Case 1.2}

In Case 1.2 of the algorithm we assume that
\begin{equation}\label{case12}
m<\frac{\sqrt{n}}{2^t \log^2 n}.
\end{equation}
Let $q=\lceil \sqrt{n}/2^t\rceil$ as in the definition of strong edges
  and let 
$$
w=\min\left(\xi |\calI|,\hspace{0.05cm} n^{2/3}\right).
$$
Some explanation about $w$, in Lemma 6.8, if we run $\AESearch$ $w$ times
  then we will find 
\begin{equation}\label{heha1}
\ell=\min(\Omega(r),\hspace{0.05cm}|\calI|/n^{2/3})
= \Omega\left(\frac{w|\calI|}{n^{4/3}}\right).
\end{equation}
many variables. So maybe we should make it a separate lemma and use in both subcases.
Also note that since $\xi\ge 1/n^{1/3}$, we have $w\ge m/\xi$.
We use $\AlgorithmCC$ described in Figure \ref{fig:algcase21}. 
It has query complexity $\tilde{O}(n^{2/3})$ (since the number of queries for $\Prune$ 
  is $m/\xi\le w$.

I will be informal in the analysis below.
First I think we can use a union bound to show that $\bT$ and most of the $\bS$
  drawn on line 4 are good :).
For some reason I need the following two cases for analysis but we can try to merge them.

The first case is when $r\ge |\calI|/n^{2/3}$.
In this case we have $w=\xi |\calI|$ and $\ell=\Omega(r)$.
For each good $\bS$, by the choice of $w$ earlier, we get $\ell$ many variables.
So in total we get
$$
N=\min\left(\frac{n^{2/3}\ell}{w},\hspace{0.05cm} \frac{q\ell}{m}\right)
=\min\left(\frac{|\calI|}{n^{2/3}},\hspace{0.05cm}\frac{q|\calI|}{n}\right).
$$
Note that the second component in the $\min$ happens when $n^{2/3}/w$ is bigger than
  $q/m$; when this happens the sets $\bS$ start to overlap with themselves. 
When this many variables have been found in $\bA$, each sample in the second stage finds a collision
  with probability (ignoring polylog and $\eps$):
$$
\Omega\left(\frac{N 2^t}{2^{h}}\right)= 
\Omega\left(\min\left(\frac{2^t}{n^{2/3}},\hspace{0.05cm}\frac{1}{\sqrt{n}}\right)\right)=\Omega\left(\frac{1}{n^{2/3}}\right).
$$

In the other case, we assume that $r<|\calI|/n^{2/3}$.
We have $\ell=|\calI|/n^{2/3}$ and $w=n^{2/3}$.
Plugging in $r=m|\calI|/n$ we also have $m<n^{1/3}$.
Thus,
$$
N=\min\left(\frac{n^{2/3}\ell}{w},\hspace{0.05cm} \frac{q\ell}{m}\right)
\ge \min\left(\frac{|\calI|}{n^{2/3}},\hspace{0.05cm}\frac{q|\calI|}{n}\right).
$$
The rest of the proof is the same.

\begin{figure}[t!]
\begin{framed}
\noindent Procedure $\AlgorithmCC\hspace{0.04cm} (f)$

\begin{flushleft}
\noindent {\bf Input:} Query access to a Boolean function $f \colon \{0, 1\}^n \to \{0, 1\}$

\noindent {\bf Output:} Either ``unate,'' or two edges constituting an edge violation of $f$ to unateness.\begin{enumerate}
\item Repeat the following $O(1)$ times:
\item \ \ \ \ \ \ \ \ Draw a set $\bT\subseteq [n]$ of size $q+1$ uniformly at random.

\item \ \ \ \ \ \ \ \ Repeat $O(n^{2/3}/w)$ times:

\item \ \ \ \ \ \ \ \ \ \ \ \ \ \ \ \ 
Draw $\bS$ by first drawing a subset $\bS_0$ of $\bT$
  of size $m$ and an ordering\\ 
  \ \ \ \ \ \ \ \ \ \ \ \ \ \ \ \ $\bpi$ of $\bS_0$ both uniformly at random and 
  then call $\Prune\hspace{0.04cm}(f,\bS_0,\bpi,\xi)$.
\item \ \ \ \ \ \ \ \ \ \ \ \ \ \ \ \ Repeat $w$ times:
\item \ \ \ \ \ \ \ \ \ \ \ \ \ \ \ \ \ \ \ \ \ \ \ Draw $\bx\in \{0,1\}^n$ uniformly and run 
  $\AESearch\hspace{0.04cm}(f,\bx,\bS\cup \{n+1\}).$
\item \ \ \ \ \ \ \ \ Let $\bA$ be the set of $i\in [n]$ such that a monotone edge in direction $i$ is found.
\item \ \ \ \ \ \ \ \ Repeat $O(n^{2/3})$ times:
\item \ \ \ \ \ \ \ \ \ \ \ \ \ \ \ \ Draw an $\by\in \{0,1\}^n$ uniformly and run
  $\AESearch\hspace{0.04cm}(f,\by,\bT)$
\item \ \ \ \ \ \ \ \ Let $\bB$ be the set of $i\in [n]$ such that an anti-monotone edge in direction $i$ is found
\item \ \ \ \ \ \ \ \ If $\bA\cap \bB\ne \emptyset$, output an edge violation of $f$ to unateness.
\item Output ``unate.''
\end{enumerate}
\end{flushleft}\vskip -0.14in
\end{framed}\vspace{-0.2cm}
\caption{Algorithm for Case 1.1}\label{fig:algcase21}
\end{figure}}


\newcommand{\AlgorithmCC}{\mathtt{AlgorithmCase1.2}}

\subsection{Algorithm for Case 1.2}

\begin{figure}[t!]
\begin{framed}
\noindent Procedure $\AlgorithmCC\hspace{0.04cm} (f)$

\begin{flushleft}
\noindent {\bf Input:} Query access to a Boolean function $f \colon \{0, 1\}^n \to \{0, 1\}$

\noindent {\bf Output:} Either ``unate,'' or two edges constituting an edge violation of $f$ to unateness.\begin{enumerate}
\item Repeat the following $O(1)$ times:
\item \ \ \ \ \ \ \ \ Draw a set $\bT\subset [n]$ of size $p$ uniformly at random.

\item \ \ \ \ \ \ \ \ Repeat $\smash{O\left(n^{2/3}\Lambda^{11} \log^2 n\big/\eps^2\right)}$ times:
\item \ \ \ \ \ \ \ \ \ \ \ \ \ \ \ \ Draw an $\bx\in \{0,1\}^n$ uniformly and run
  $\AESearch\hspace{0.04cm}(f,\bx,\bT)$
\item \ \ \ \ \ \ \ \ Let $\bA$ be the set of $i\in [n]$ such that an anti-monotone edge along $i$ is found.

\item \ \ \ \ \ \ \ \ Repeat  $\smash{O\left({n^{1/3} \log^3 n}\big/({m 2^t}) \right)}$ times:

\item \ \ \ \ \ \ \ \ \ \ \ \ \ \ \ \ 
Draw $\bS$ by first drawing a subset $\bS_0$ of $\bT$
  of size $m$ and an ordering\\ 
  \ \ \ \ \ \ \ \ \ \ \ \ \ \ \ \ $\bpi$ of $\bS_0$ both uniformly at random and 
  then call $\Prune\hspace{0.04cm}(f,\bS_0,\bpi,\xi)$.
\item \ \ \ \ \ \ \ \ \ \ \ \ \ \ \ \ Repeat $\smash{O\left((2^{h}/2^s) \cdot \log n\right)}$ times:
\item \ \ \ \ \ \ \ \ \ \ \ \ \ \ \ \ \ \ \ \ \ \ \ Draw $\by\in \{0,1\}^n$ uniformly and run 
  $\AESearch\hspace{0.04cm}(f,\by,\bS\cup \{n+1\}).$
\item \ \ \ \ \ \ \ \ Let $\bB$ be the set of $i\in [n]$ such that a monotone edge along variable $i$ is found.
\item \ \ \ \ \ \ \ \ If $\bA\cap \bB\ne \emptyset$, output an edge violation of $f$ to unateness.
\item Output ``unate.''
\end{enumerate}
\end{flushleft}\vskip -0.14in
\end{framed}\vspace{-0.2cm}
\caption{Algorithm for Case 1.2}\label{fig:algcase21}
\end{figure}

The second subcase of Case 1 occurs when
\begin{equation}\label{case12}
m<\frac{n^{1/3} \log^2 n}{2^t}, 
\end{equation}
and the crucial difference is that unlike Case 1.1, we cannot conclude with (\ref{eq:r-value}). More specificially, two potential issues which were not present in Section~\ref{sec:case1-1} may arise: 1) sampling a set $\bS \sim \calD_{\xi, m}$ may result in $\Caught(\bS) = \emptyset$, and 2) even if $|\Caught(\bS)|$ is large, too few points $\bx$ may result in $\AESearch\hspace{0.05cm}(f,\bx, \bS)$ returning an anti-monotone edge (for instance, when $2^t = 1$ and $2^h = n$). Thus, we address these two problems with $\AlgorithmCC$, which is described in Figure~\ref{fig:algcase21}. 

At a high level, $\AlgorithmCC$ proceeds by sampling a set $\bT\subseteq [n]$ of size 
$$p:=\lceil \sqrt{n}/2^t \rceil+1=\Theta(\sqrt{n}/2^t)$$
uniformly at random, 
  and directly uses the set $\bT$ to search for anti-monotone edges
  by repeating $\AESearch\hspace{0.04cm}(f, \bx, \bT)$ for $\tilde{O}(n^{2/3}/\eps^2)$ many iterations.
We show at the end the algorithm will obtain anti-monotone edges along 
  at least $\Omega(n^{1/6})$ variables in $\bT \cap \calI$ (as an anti-monotone edge is~found every $\tilde{O}(\sqrt{n}/\eps^2)$ iterations of $\AESearch\hspace{0.04cm}(f, \bx, \bS)$). The algorithm will then sample $\bS_0 \subset \bT$ of~size $m$ (notice that $m\ll p$ by (\ref{case12})), pass it through
  $\Preprocess$ to obtain $\bS\subseteq \bS_0$, and then repeat $\AESearch\hspace{0.04cm}(f, \by, \bS_0)$ in hopes of observing an edge violation. In order to do so, $\bS$ must contain a variable where the algorithm has already observed an anti-monotone edge. Since the number of variables with anti-monotone edges observed may be $O(n^{1/6})$ and 
  $m\cdot n^{1/6}$ could be much smaller than $p$, the algorithm needs to sample the subset $\bS_0$
  from $\bT$ multiple times. 

We state the query complexity of the algorithm for Case 1.2, where we note that the upper bound will follow from (\ref{case12}), (\ref{boundforxi}), (\ref{eq:def-m}), the fact that $2^t \geq 1$ and $2^h/\calI \geq \tilde{\Omega}(\eps^2)$.
\begin{lemma} The query complexity of
$\AlgorithmCC$ is (using (\ref{boundforxi}))
\[ \tilde{O}\left( \frac{n^{2/3}}{\eps^2} \right) + \tilde{O}\left( \frac{n^{1/3}}{m   2^t}\right) \left( \tilde{O}\left(\frac{m}{\xi}\right) + \tilde{O}\left(\frac{2^{h}}{2^s}\right) \right) = \tilde{O}(n^{2/3}/\eps^2). \]
\end{lemma}

In order to analyze $\AlgorithmCC$ we consider one of the main iterations and let 
  $T$ denote the set drawn in line 2.
We define the
  following subset $C \subseteq T \cap \calI^*$ to capture how good $T$ is:
A variable $i\in T\cap \calI^*$ belongs to $C$ if it satisfies both of the following conditions:
\begin{flushleft}\begin{enumerate}
\item[(i)] The set $T \setminus \{ i\}$ is $i$-informative for anti-monotone edges; and
\item[(ii)] 
If we draw $\bS$ by first drawing a subset $\bS_0$ of $T$ of size $m$ conditioning 
  on $i\in \bS_0$ and 
  an ordering of $\bpi$ of $\bS_0$ uniformly at random and then calling 
  $\Prune\hspace{0.04cm}(f,\bS_0,\bpi,\xi)$ to\\ get $\bS$, then 
  the  probability that $\bS$ catches $i$ is at least $1/2$. 
\end{enumerate}\end{flushleft}

We prove that when $\bT$ is a random size-$p$ set, the set $\bC$  
  is large with constant probability. 

\begin{lemma}\label{lem:c-large}
With probability at least $\Omega(1)$ over the draw of $\bT \subset [n]$ in line 2, we have $$|\bC| = \Omega\left(\frac{ p|\calI^*| }{n}\right).$$
\end{lemma}

\begin{proof}
We first note that $p|\calI^*|/ n = \Omega(n^{1/6})$ since $|\calI^*| /2^t = \Omega(n^{2/3})$. By Lemma~\ref{appendix1}, $$|\bC| \leq |\bT \cap \calI^*| \leq 4\hspace{0.05cm}p|\calI^*|/ {n}$$ with probability at least $1 - \exp ( -\Omega(n^{1/6}) )$. We consider the quantity
\[ \alpha = \Prx_{\bT \subset [n]}\left[ |\bC| \geq \frac{1}{10}\cdot \frac{p|\calI^*|}{n} \right], \]
which we will lower bound by $\Omega(1)$ by proving both upper and lower bounds for $\Ex_{\bT}[|\bC|]$. 

For the lower bound, we use the following claim which we prove next. 

\begin{claim}\label{final-touch}
For every $i \in \calI^*$, we have
\[ \Prx_{\bT \subset [n]}\Big[\text{$\bT$ and $i$ satisfy conditions (i) and (ii)}\hspace{0.08cm}\Big|\hspace{0.08cm} i \in \bT \Big] \geq 1-o(1). \]
\end{claim}
It follows from Claim \ref{final-touch} that
\begin{align*}
\Ex_{\bT \subset [n]}\big[ |\bC|\big] &= \sum_{i \in \calI^*} \Prx_{\bT\subset[n]}\big[i \in \bT \big] \cdot \Prx_{\bT\subset [n]}\Big[\text{(i) and (ii) satisfied}
\hspace{0.08cm}\Big|\hspace{0.08cm}  i \in \bT \Big] \geq (1-o(1)) \cdot \frac{p|\calI^*|}{n}.
\end{align*} 
On the other hand, we may upper bound $\Ex_\bT[|\bC|]$ using the definition of $\alpha$,
\[ \Ex_{\bT \subset [n]}\big[ |\bC| \big] \leq n\cdot \exp\left( -\Omega(n^{1/6})\right) + \left(1-\exp\left(-\Omega(n^{1/6} \right)\right) \left( \alpha \cdot \frac{4\hspace{0.03cm}p|\calI^*|}{n} + (1-\alpha) \cdot \frac{p|\calI^*|}{10n}\right), \]
which in turn, implies that $\alpha = \Omega(1)$.
\end{proof}

\begin{proof}[Proof of Claim \ref{final-touch}]
Consider a fixed $i \in \calI^*$, we will show that sampling $\bT$ and conditioning on $i \in \bT$, the probability of $i$ satisfying condition (i) is at least $1 - o(1)$, and likewise, the probability of $i$ satisfying condition (ii) is also at least $1 - o(1)$. The claim then follows from a union bound. 

First, $\bT \subset [n]$ conditioned on $i\in \bT$ is distributed exactly as $\bT' \cup \{i\}$ where $\bT' \sim \calP_{i,\lceil \sqrt{n}/2^t\rceil}$, by our choice of $p$. 
The part on condition (i) follows directly from 
  Lemma~\ref{lem:informative-for-antimon}.
  
Second, let $\beta$ be the probability over $\bT \subset [n]$ conditioning on $i\in \bT$ that (ii) is not satisfied. We consider the following quantity
$$
\Prx_{\bS \sim \calD_{\xi, m, i}}\big[ \bS \text{ catches $i$}\big].
$$
On the one hand, we can we bound it from above by $1-o(1)$ by combining Claim~\ref{trickyclaim} and Corollary \ref{coro1}. On the other hand, $\bS \sim \calD_{\xi, m, i}$ is exactly distributed as first sampling $\bT \subset [n]$, then sampling $\bS_0\subset \bT$ conditioned on $i\in \bS_0$ and 
  finally running $\Preprocess$ using $\bS_0$ to get $\bS$.
Thus, we may use the definition of $\beta$ to upperbound the probability by $\beta/2 + 1-\beta$, which implies $\beta = o(1)$.
%
\end{proof}

Having established Lemma~\ref{lem:c-large}, we consider a fixed iteration of line 2 where the set $C$
  obtained from $T$  
  satisfies the size lower bound from Lemma~\ref{lem:c-large}. Note that since line 2 is executed $O(1)$ times, this will happen with large constant probability. 
After fixing $T$ and $C$, we will show that in this iteration, $\AlgorithmCC$ finds an edge violation to unateness with high probability.

\begin{lemma}\label{lem:a-large}
With probability $1-o(1)$ over the randomness in lines 3--5, $|\bA \cap C| \geq \Omega(n^{1/6})$. 
\end{lemma}

\begin{proof}
We consider breaking up the execution of lines 3--4 into $O(n^{1/6})$ phases, each consists of $$O\left(\frac{\sqrt{n}\hspace{0.04cm} \Lambda^{11} \log^2 n}{\eps^2}\right)$$ execution of line 4 each. We note that 
$$|C| = \Omega\left(\frac{p|\calI^*|}{n} \right) = \Omega(n^{1/6})$$ since $|\calI|/2^t \ge n^{2/3}$. Consider a particular phase of the algorithm, and assume that the number~of anti-monotone edges along variables in $C$ observed is less than $|C|/2$ (otherwise, we are done
  since $|C|=\Omega(n^{1/6})$). In this case, for each $j \in C$ along which the algorithm has not observed an anti-monotone edge, we let $X_j$ be the set of points $x \in \{0,1\}^n$ such that $\AESearch\hspace{0.04cm}(f,x,T)$ will output $j$ with probability at least $2/3$. By condition (i), we have that 
\[ |X_j| = \Omega\left(\frac{2^t}{2^h}\right) \cdot 2^n,\]
and that the $X_j$'s are disjoint by Corollary~\ref{AESearchCorollary}. As a result, there are at least
\[ \Omega\left( \frac{2^t}{2^h} \right) \cdot 2^n \cdot \dfrac{|\calI^*|}{2^t \sqrt{n}} = \Omega\left( \dfrac{|\calI^*|}{2^h \sqrt{n} }\right) \cdot 2^n = \Omega\left(\frac{\eps^2}{\sqrt{n} \cdot \Lambda^{11}} \right) \cdot 2^n \]
many such points $x \in \{0,1\}^n$. By the number of queries we have in each phase, we will observe an anti-monotone edge along some variable in $C$ not yet seen with high probability.
\end{proof}

By Lemma~\ref{lem:a-large} we consider the case when $|\bA \cap C| = \Omega(n^{1/6})$. Fix a particular   $A$ and~a~subset of $\calJ=A\cap C$ such that $|\calJ|=\Theta(n^{1/6})$.

\begin{lemma}
With probability at least $1-o(1)$, 
  at least one of the $\bS$ sampled in line 7 catches 
  at least one variable in $ \calJ$. 
\end{lemma}
\begin{proof}
Let $\beta$ be the probability
  that $\bS$ catches at least one variable in $\calJ$.
Consider 
$$
\Ex_{\bS}\Big[\big\{j\in \calJ: \text{$\bS$ catches $j$}\big\}\Big].
$$
Using condition (ii) on variables in $\calJ$, we have
$$
\Ex_{\bS}\Big[\big\{j\in \calJ: \text{$\bS$ catches $j$}\big\}\Big]=
\sum_{j\in \calJ} \Pr \big[\bS_0\ \text{contains $j$}\big]\cdot 
  \Pr \Big[\bS\ \text{catches $j$}\hspace{0.08cm}\Big|\hspace{0.08cm}\text{$j\in \bS_0$}\Big]
 =\Omega\left(\frac{m2^t}{n^{1/3}}\right).
$$

Since $\bS_0$ is a uniform subset of $\bT$ of size $m$, and $|\calJ| = \Theta(n^{1/6})$, the expectation of $|\bS_0\cap \calJ|$ is $$\Theta\left(\frac{m n^{1/6}}{\sqrt{n}/2^t}\right) = 
\Omega\left(\frac{m 2^t}{n^{1/3}}\right)=O(\log^2n).$$ 
Therefore, the probability of $|\bS\cap \calJ|\le |\bS_0\cap \calJ|$ being at most 
  $O(\log^2 n)$ is at least $1-\exp(-\Omega(\log^2 n))$ by 
  Lemma~\ref{appendix1}. 
  Thus, we may upperbound the above expectation using the definition of $\beta$ by 
$$\beta \cdot O(\log^2 n) + |\calJ|\cdot \exp\big(-\Omega(\log^2 n)\big),$$ which gives the following lower bound on $\beta$:
$$
\beta\ge \Omega\left(\frac{m2^t}{n^{1/3}\log^2 n}\right).
$$
The lemma follows from the number of times we repeat in line 7.
\end{proof}

Finally we show that if $S$ catches a $j\in \calJ$, then
  line 9 finds a monotone edge along $j$.

\begin{lemma}
If $S$ catches $j\in \calJ$ in line 7, then line 9 finds a monotone edge along variable $j$ with probability
  at least $1-o(1)$.
\end{lemma}

\begin{proof}
Since $S$ catches $j$, the number of points $x\in \{0,1\}^n$ such that $\AESearch\hspace{0.05cm}(f,x,S)$ returns
  a monotone edge along $j$ with probability at least $2/3$ is at least
\[ \Omega\left(\frac{2^s}{2^h}\right) \cdot 2^n \]
The lemma then follows from the number of times the algorithm repeats in line 8.
\end{proof}

\ignore{\begin{lemma}
Query complexity is
\[ O(n^{2/3}/\eps^2) + \dfrac{n^{1/3} \log^2 n}{m \cdot 2^t} \left(\frac{m}{\xi} + \frac{2^h}{2^s} \log^2 n \right) = \]
\end{lemma}

Let $q=\lceil \sqrt{n}/2^t\rceil$ as in the definition of strong edges
  and let 
$$
w=\min\left(\xi |\calI|,\hspace{0.05cm} n^{2/3}\right).
$$
Some explanation about $w$, in Lemma 6.8, if we run $\AESearch$ $w$ times
  then we will find 
\begin{equation}\label{heha1}
\ell=\min(\Omega(r),\hspace{0.05cm}|\calI|/n^{2/3})
= \Omega\left(\frac{w|\calI|}{n^{4/3}}\right).
\end{equation}
many variables. So maybe we should make it a separate lemma and use in both subcases.
Also note that since $\xi\ge 1/n^{1/3}$, we have $w\ge m/\xi$.
We use $\AlgorithmCC$ described in Figure \ref{fig:algcase21}. 
It has query complexity $\tilde{O}(n^{2/3})$ (since the number of queries for $\Prune$ 
  is $m/\xi\le w$.

I will be informal in the analysis below.
First I think we can use a union bound to show that $\bT$ and most of the $\bS$
  drawn on line 4 are good :).
For some reason I need the following two cases for analysis but we can try to merge them.

The first case is when $r\ge |\calI|/n^{2/3}$.
In this case we have $w=\xi |\calI|$ and $\ell=\Omega(r)$.
For each good $\bS$, by the choice of $w$ earlier, we get $\ell$ many variables.
So in total we get
$$
N=\min\left(\frac{n^{2/3}\ell}{w},\hspace{0.05cm} \frac{q\ell}{m}\right)
=\min\left(\frac{|\calI|}{n^{2/3}},\hspace{0.05cm}\frac{q|\calI|}{n}\right).
$$
Note that the second component in the $\min$ happens when $n^{2/3}/w$ is bigger than
  $q/m$; when this happens the sets $\bS$ start to overlap with themselves. 
When this many variables have been found in $\bA$, each sample in the second stage finds a collision
  with probability (ignoring polylog and $\eps$):
$$
\Omega\left(\frac{N 2^t}{2^{h}}\right)= 
\Omega\left(\min\left(\frac{2^t}{n^{2/3}},\hspace{0.05cm}\frac{1}{\sqrt{n}}\right)\right)=\Omega\left(\frac{1}{n^{2/3}}\right).
$$

In the other case, we assume that $r<|\calI|/n^{2/3}$.
We have $\ell=|\calI|/n^{2/3}$ and $w=n^{2/3}$.
Plugging in $r=m|\calI|/n$ we also have $m<n^{1/3}$.
Thus,
$$
N=\min\left(\frac{n^{2/3}\ell}{w},\hspace{0.05cm} \frac{q\ell}{m}\right)
\ge \min\left(\frac{|\calI|}{n^{2/3}},\hspace{0.05cm}\frac{q|\calI|}{n}\right).
$$
The rest of the proof is the same.}


\section{Finding Bichromatic Edges of High Influence Variables}\label{sec:highinf}

\newcommand{\FindHiInf}{\texttt{Find-Hi-Inf}}
\newcommand{\CheckRevealing}{\texttt{Check-Revealing}}
\newcommand{\GetRevealingEdges}{\texttt{Get-Revealing-Edges}}
\newcommand{\Good}{\textsc{Good}}
\newcommand{\NoGood}{\textsc{NoGood}}

The goal of this section is to present a procedure \FindHiInf\  
  which will be used in our algorithm for Case 2.
The procedure \FindHiInf\ takes four inputs: query access 
  to a Boolean function $f$, a set $S\subseteq [n]$ of variables,
  a positive integer $m$ and a parameter $\alpha\in (0,1]$.
When $f$ satisfies: 
\begin{equation}\label{cond:many}
\text{There is a hidden
  subset $H\subseteq S$
  such that $|H|\ge m$ and every $h\in H$ satisfies $\Inf_f[h]\ge \alpha$,}
\end{equation}
the procedure \FindHiInf\ efficiently finds a bichromatic edge for 
  \emph{almost all} variables in $H$.
  

At a high level, the procedure \FindHiInf\ (described in Figure~\ref{fig:find-hi-influence}) will exploit  the fact that if there is a hidden set $H$ of variables with high influence, there will be sufficiently many points~in the hypercube which are sensitive along many variables from $H$. As a result, every time one such high sensitivity point is identified, we may use it to observe bichromatic edges along multiple variables. 


\subsection{Revealing points}

We start with the definition of \emph{revealing points}.
These are points that, once identified,
  can be~used to~observe bichromatic edges along multiple variables.
We then present two subroutines that will be~used in  
  \FindHiInf.
Given a revealing point,
  the first subroutine 
  $\GetRevealingEdges$ (see Figure~\ref{fig:reveal-edges}) 
  uses it to find bichromatic edges along a large number of variables. 
While it is expensive to run, we 
  present a second subroutine $\CheckRevealing$ (see Figure~\ref{fig:degree-test})  
  that can check whether a given point is revealing
  so we only run the first one when we are sure that the point is revealing.


\newcommand{\Reveal}{\textsc{Reveal}}

\begin{definition}
Let $\delta \in (0, 1]$ and $T = \{ T_j \}_{j\in [r]}$ be a collection of disjoint
  (but possibly empty)~subsets of $[n]$ for some $r\ge 1$. Given a point $x \in \{0,1\}^n$, we consider the set:
\[ \Reveal\hspace{0.04cm}(x, T, \delta) = \left\{j \in [r] : \Prx_{\bR \subseteq T_j}\left[ f(x) \neq f(x^{(\bR )}) \right] \geq \delta \right\}, \]
where $\bR \subseteq T_j$ is subset of $T_j$ drawn uniformly at random
  (i.e., each element of $T_j$ is included in $\bR$ with probability $1/2$ independently).
Moreover, for $\gamma \in (0, 1]$, we say $x$  is $(\gamma, \delta)$-\emph{revealing} for $T$ if \[ \big|\Reveal\hspace{0.04cm}(x, T,\delta)\big| \geq \gamma \cdot r. \]
\end{definition}

Next we present the first subroutine $\GetRevealingEdges$.
Given $x\in \{0,1\}^n$, $T=\{T_j\}_{j\in [r]},$ and $\delta\in (0,1]$, $\GetRevealingEdges\hspace{0.04cm}(f,x,T,\delta)$ (see Figure~\ref{fig:reveal-edges}) makes $\smash{\tilde{O}(r/\delta)}$ queries and 
  outputs a set of bichromatic edges which contains one bichromatic edge for each variable in $\Reveal\hspace{0.05cm}(x,T,\delta)$.

\begin{figure}[t!]
\begin{framed}
\noindent Subroutine $\GetRevealingEdges\hspace{0.04cm}(f,x, T, \delta)$
\begin{flushleft}
\noindent {\bf Input:} Query access to $f \colon \{0,1\}^n \to \{0,1\}$, a point $x \in \{0,1\}^n$, a collection of disjoint subsets $T = \{ T_j \}_{j\in [r]}$ of $[n]$ for some $r\ge 1$, and $\delta \in (0, 1]$. \\
\noindent {\bf Output:} Two random sets $\bB$ and $\bQ$, where $\bB$ is a set of bichromatic edges of $f$ and $\bQ$ is a subset of the union of $T_j$'s.
\begin{itemize}
\item Initialize $\bB = \bQ = \emptyset$. 
Repeat the following for each $j \in [r]$ for $ {\log^2 n}/{\delta}$ iterations each.
\begin{enumerate}
\item Sample $\bR \subseteq T_j$ uniformly at random, and let $\bpi$ be an arbitrary ordering on $\bR $.\vspace{0.06cm}
\item Run $\BinarySearch\hspace{0.04cm}(f, x, \bT , \bpi)$. If it outputs $\nil$, do nothing. If it outputs a bichromatic edge $\be$ along variable $\bi \in \bT $. Set
$\bB \leftarrow \bB \cup \{ \be\}$ and $\bQ \leftarrow \bQ \cup \{ \bi \}.$
\end{enumerate}
\item Output $\bB$ and $\bQ$.
\end{itemize}
\end{flushleft}\vskip -0.11in
\end{framed}\vspace{-0.2cm}
\caption{The subroutine $\GetRevealingEdges$.} 
\label{fig:reveal-edges}
\end{figure}


\begin{lemma}
\label{lem:revealing}
Let $x\in \{0,1\}^n$, $\delta\in (0,1]$ and $T = \{ T_j \}_{j\in [r]}$ be a collection of $r$
  disjoint subsets of $[n]$ for some $r\ge 1$. 
Then, $\emph{\GetRevealingEdges}\hspace{0.04cm}(f,x, T, \delta)$ uses
  $\smash{\tilde{O}(r/\delta)}$ queries and always returns a set $\bQ$ of variables 
  and a set $\bB$ of bichromatic edges such that $\bQ$ is a subset of the union of $T_j$'s and 
  $\bB$ contains a bichromatic edge along each variable in $\bQ$.
Furthermore, $\bQ$ contains $\Reveal\hspace{0.04cm}(x,T,\delta)$
  with probability at least $1-1/\poly(n)$.
\end{lemma}
\begin{proof}
For each index $j \in \Reveal(x, T,\delta)$, the probability that
  none of the ${\log^2 n}/{\delta}$ many 
  random subsets $\bR$ of $T_j$ satisfies $f(x)\ne f(x^{(\bR)})$ is at most
\[ (1 - \delta)^{\frac{\log^2 n}{\delta}} \ll {1}/{\poly(n)},\]
When this happens
  $\BinarySearch\hspace{0.05cm}(f, x, \bR, \bpi)$ finds a bichromatic edge along a variable in $\bR$ which is also in $T_j$.
The lemma then follows from a union bound over all $j\in \Reveal(x,T,\delta)$.
\end{proof}

Next we describe the second subroutine $\CheckRevealing$ in Figure~\ref{fig:degree-test},
   which we use to 
  check whether a given point $x$ is $(\gamma,\delta)$-revealing with respect to $T$. 

\begin{figure}[t!]
\begin{framed}
\noindent Subroutine $\CheckRevealing\hspace{0.04cm}(f,x, T, \gamma, \delta)$
\begin{flushleft}
\noindent {\bf Input:} Query access to  $f \colon \{0,1\}^n \to \{0,1\}$, a point $x \in \{0,1\}^n$, a collection of disjoint subsets $T = \{ T_j \}_{j\in [r]}$ of $[n]$ for some $r\ge 1$, and $\gamma, \delta \in (0, 1]$. \\
 {\bf Output:} Either ``accept'' or ``reject.''

\begin{enumerate}
\item Initialize $\bb = 0$.
\item Repeat the following steps for ${\log^2 n}/{\gamma}$ iterations:
\begin{itemize} 
\item Sample an index $\bj\sim [r]$ uniformly at random and initialize $\bc= 0$.
\item Repeat the following ${\log^2 n}/{\delta}$ iterations:\vspace{0.06cm}
\begin{itemize}
\item Sample a subset $\bR \subseteq T_{\bj}$ uniformly and increment $\bc$
  if $f(x^{(\bR)}) \neq f(x)$. \vspace{0.1cm}
\end{itemize}
\item If $\bc \geq  {3\log^2 n}/{4}$, increment $\bb$.\vspace{0.05cm}
\end{itemize}
\item If $\bb \geq {3 \log^2 n}/{4}$, output ``accept;'' otherwise, output ``reject.''
\end{enumerate}
\end{flushleft}\vskip -0.11in
\end{framed}\vspace{-0.2cm}
\caption{The subroutine $\CheckRevealing$.}
\label{fig:degree-test}
\end{figure}

\begin{lemma}\label{lem:check}
Let $\gamma, \delta \in (0, 1]$, $x \in \{0,1\}^n$ and $T = \{ T_j \}_{j\in [r]}$ be a collection of disjoint subsets of~$[n]$ for some $r\ge 1$. Then, $\emph{\CheckRevealing}\hspace{0.04cm}(f,x, T, \gamma, \delta)$ makes $\smash{\tilde{O}(1/(\gamma \delta))}$ many queries and satisfies the following two conditions:
(1) If $x$ is $(\gamma, \delta)$-revealing with respect to $T$, then it outputs ``accept" with probability at least $1 - 1/{\poly(n)}$; (2)
If $x$ is not $({\gamma}/2, {\delta}/2)$-revealing with respect to $T$, then it outputs ``reject" with probability at least $1 - 1/{\poly(n)}$.
\end{lemma}

\begin{proof}
Suppose that $x$ is $(\gamma, \delta)$-revealing.
It follows from Chernoff bound that with probability at least $1-1/\poly(n)$,
  the number of iterations in which the index $j\in [r]$ lies in $\Reveal(x,T,\delta)$
  is at least $3\log^2 n/4$.
Moreover, for every such index $j$, the counter $b$ is incremented with
  probability at least $1-1/\poly(n)$.
By a union bound, the subroutine accepts with probability $1-\poly(n)$.


Suppose that $x$ is not $({\gamma}/{2},{\delta}/{2})$-revealing. Then we have that $|\Reveal(x, T, {\delta}/{2})| \leq  {\gamma r}/{2}$. The claim follows similarly by applying Chernoff bounds and then a union bound.
\end{proof}

We can combine Lemma~\ref{lem:revealing} and \ref{lem:check} to conclude that when
  $ {\CheckRevealing}\hspace{0.05cm}(x, T, \gamma, \delta)$ returns ``accept,''
  it is safe to run ${\GetRevealingEdges}\hspace{0.05cm}(x, T, {\delta}/{2})$
  and we should expect the latter to find at least $\gamma r/2$ many bichromatic edges along
  different variables from the union of $T_j$'s.


\newcommand{\FindRevealing}{\texttt{Find-Revealing}}

\subsection{The $\FindRevealing$ procedure}

Recall that we are given an 
  $S \subseteq [n]$ and two parameters $m$ and $\alpha$, and we 
  are interested in the case when $S$ contains a hidden set $H\subseteq S$
  that satisfies (\ref{cond:many}).

Assuming this is the case, we prove in Corollary \ref{coro5} that 
  there exist many 
  points $z$ such that, with probability $\Omega(1)$ over the draw of a
  random partition $\bT$ of $S$, $z$ is $(\Omega(1),\Omega(1))$-revealing for $\bT$.
We then present a procedure called $\FindRevealing$ that takes
  advantage of this property to find bichromatic edges along 
  many different variables in $S$.
Note that in establishing Corollary \ref{coro5}, we assume that 
  there is a subset $H$ of $S$ that satisfies (\ref{cond:many}), and $H$ is used in the analysis (and known to us in the proofs of these lemmas), even though $H$ is  hidden from the procedure $\FindRevealing$.

We start by using the set $H$ to define the following 
  bipartite graph $G_0= (U_0, V_0, E_0)$ consisting of bichromatic edges of $f$ along variables in $H$: 
\begin{itemize}
\item $U_0 \subseteq \{0,1\}^n$ is the set of left vertices that contain all $x\in \{0,1\}^n$ with $f(x) = 0$.
\item $V_0 \subseteq \{0,1\}^n$ is the set of right vertices that contain all $x\in\{0,1\}^n$ with $f(x) = 1$.
\item $E_0$ is the set of bichromatic edges of $f$ connecting vertices 
  between $U_0$ and $V_0$ along\\ variables in $H$. 
Given that $\Inf_f[i]\ge \alpha$, we have $|E_0| \geq ( \alpha |H|/2) \cdot 2^n
\ge (\alpha m/2)\cdot 2^n$. 
\end{itemize}
Recall the definition of left-$d$-good and
  right-$d$-good bipartite graphs.
The next lemma, similar to Lemma~6.5 of~\cite{KMS15}, shows the existence of an induced subgraph of $G_0$ 
  that has roughly the same number of edges as $G_0$ but is either left-$d$-good or right-$d$-good.

\begin{lemma}\label{lem:reg-graph}
There exist a positive integer $d\le |H|$ that is a power of $2$  as well as subsets $U \subseteq U_0$ and $V \subseteq V_0$ such that the subgraph $G=(U,V,E)$ of $G_0$ induced by vertices $(U,V)$ satisfies:
\begin{itemize}
\item $G$ is either left-$d$-good or right-$d$-good (letting $\sigma=|U|/2^n$ or $|V|/2^n$ accordingly), and\vspace{-0.1cm}
\item $\sigma$ and $d$ satisfy $\sigma d\ge \alpha m/(6\log n)$.\vspace{0.1cm}
\end{itemize}
\end{lemma}

\begin{proof}
Note that $G_0$ has degree at most $|H|$.
Consider the following procedure:\vspace{0.1cm}
\begin{flushleft}\begin{itemize}
\item Iterate through $d = 2^k, 2^{k-1},\ldots,1$, where $2^k$ is the largest power of $2$ that is at most $|H|$, while maintaining a set of vertices $D$ (as vertices deleted so far) which is initially empty:
\begin{enumerate}
\item Consider the subgraph $G'$ of $G_0$ induced by $(U_0\setminus D,V_0\setminus D)$ and let
\[ U^* = \big\{ x \in U_0 \setminus D: \deg_{G'}(x) \geq d \big\} \quad\text{and}\quad V^* = \big\{ y \in V_0 \setminus D : \deg_{G'}(y) \geq d\big\}. 
\] 
\item Terminate if either $U^*$ or $V^*$ has size at least $$\frac{\alpha m}{6d \log n} \cdot 2^n.$$ 
 
\item Otherwise, we update $D \leftarrow D \cup U^* \cup V^*$ (i.e., delete all vertices in
  $U^*$ and $V^*$). \vspace{0.1cm}
\end{enumerate}
\end{itemize}\end{flushleft}
By induction, we have that at the beginning of each round, every 
  vertex in $G'$ has degree at most $2d$.
As a result, if the procedure terminates during one of the iterations,
  we have that the subgraph of $G_0$ induced by either $(U^*,V_0\setminus D)$ or 
  $(U_0\setminus D,V^*)$ satisfies both properties and we are done.

So it suffices to show that the procedure terminates.
Assume for contradiction that it does not terminate.
Then the number of edges deleted (i.e., those adjacent to vertices in $U^*\cup V^*$)
  during each iteration is at most $(\alpha m/(3\log n))\cdot 2^n$.
Since there are no more than $\log n$ iterations, the total number
  of edges deleted is at most $(\alpha m/3)\cdot 2^n$.
This contradicts with the fact that all edges of $H_0$ will be deleted at the end
  if the procedure does not terminate, and that $|E_0|\ge (\alpha m/2)\cdot 2^n$.
\end{proof}

Consider such a subgraph $G = (U, V, E)$ of $G_0$ with parameters 
  $d$ and $\sigma$, given by Lemma~\ref{lem:reg-graph}.
We assume without loss of generality $G$ is left-$d$-good since the proof
  of Corollary \ref{coro5} 
  is symmetric  when $G$ is right-$d$-good by considering $f'$  given by $f'(x) = f(x)\oplus 1$.
For each $z\in U\cup V$, we define
$$
N(z)=\big\{i\in H:
(x, x^{(i)})\in E\big\}.  
$$
Thus, $|N(z)|=\deg_G(z)$.
Corollary \ref{coro5} follows from two technical lemmas (Lemma \ref{lem:good-points}
  and Lemma \ref{lem:Lleftlarge}).
Before stating them we need the following definition.

\begin{definition}
Let $\gamma_0=\delta_0=\eta_0=0.1$ be three constants and $r=100\hspace{0.03cm}d$. 
We say that a point $y\in V$ is \emph{good for $r$-partitions} of $S$ if with probability at least $\eta_0$ over a uniformly random $r$-partition $\bT =$ $ \{ \bT_j \}_{j\in [r]}$ of $S \setminus N(y)$, $y$ is $(\gamma_0, \delta_0)$-revealing with respect to $\bT$.
(To be more formal, such a partition $\bT$ is drawn by first sampling a map $\bg$ from $S\setminus N(y)$
  to $[r]$ uniformly at random, i.e., each element is mapped to an index in $[r]$
  uniformly and independently, and then setting $\bT_j=\bg^{-1}(j)$.)

We use $\Good(S)$ to denote the set of points $y\in V$ that are good
  for $r$-partitions of $S$.
\end{definition}


We 
  delay the proof of the following two technical lemmas.

\begin{lemma}\label{lem:good-points}
For each point $y \in \Good(S)$, with probability at least ${\eta_0}/{2}$ over the draw of a random $r$-partition $\bT$ of $S$, $y$ is $\left( {\gamma_0}/{2}, {\delta_0}/{2}\right)$-revealing for $\bT$.
\end{lemma}

\begin{lemma}\label{lem:Lleftlarge}
If $|\Good(S)| \leq  ({\sigma}/{100}) \cdot 2^n$, there exist $\sigma 2^n/2$
  points $x \in U$ such that with~probability at least $ {1}/{3}$ over the draw of a random $r$-partition $\bT$ of $S$, $x$ is $(1/4000, ({1-\delta_0})/{2})$-revealing for $\bT$.
\end{lemma}

\begin{figure}[t!]
\begin{framed}
\noindent Procedure $\FindRevealing\hspace{0.05cm}(f,S, m,\alpha)$
\begin{flushleft}
\noindent {\bf Input:} Query access to $f \colon \{0,1\}^n \to \{0,1\}$, $S\subseteq [n]$,
  $m\ge 1$ and $\alpha\in (0,1]$. 
  
\noindent {\bf Output:} Either a set of bichromatic edges $\bB$ 
  and a subset $\bQ$ of $S$, or ``fail.''
\begin{flushleft}\begin{itemize}
\item Let $\eta_1=\min\big(\eta_0/2,1/3\big)$, $\gamma_1=\min\big(\gamma_0/2,1/4000\big)$ and 
  $\delta_1=\min\big(\delta_0/2,(1-\delta_0)/2\big)$.
\item Repeat the following for all $d^*$ being a power of $2$ less than $|S|$ from small to large:
\begin{enumerate}
\item Let $r^*=100d^*$ and $\sigma^* = \alpha m/(6d^* \log n)$.
Skip this iteration if $\sigma^*>1$.\vspace{0.05cm}
\item Sample $t= {\log^2 n}/{\sigma^*}$ points $\bx_1, \dots, \bx_t$ from $\{0,1\}^n$ uniformly at random.\vspace{0.05cm}
\item For each point $\bx_j$, for ${\log^2 n}/{\eta_1}$ times, let $\bT$ be an $r^*$-partition of $S$ drawn uniformly at random, and run $\CheckRevealing\hspace{0.05cm}(f,\bx_j, \bT, \gamma_1, \delta_1)$.\vspace{0.08cm}
\begin{itemize}
\item If it outputs ``reject,'' do nothing. \vspace{0.06cm}
\item If it outputs ``accept,'' run $\GetRevealingEdges\hspace{0.05cm}\left(f,\bx_j, \bT, {\delta_1}/{2}\right)$ to get $\bB$ and $\bQ$. If $|\bQ| < {\gamma_1 r^*}/{2}$, output ``fail'' and terminate; otherwise, output $(\bQ, \bB)$.\vspace{0.06cm}
\end{itemize}
\end{enumerate}
\item When this line is reached (i.e. all calls to $\CheckRevealing$ return ``fail'') output ``fail.''\vspace{0.1cm}
\end{itemize}\end{flushleft}
\end{flushleft}\vskip -0.14in
\end{framed}\vspace{-0.2cm}
\caption{The procedure $\FindRevealing$.}\label{fig:find-revealing}
\end{figure}

Using constants $\eta_1,\gamma_1$ and $\delta_1$ defined in Figure \ref{fig:find-revealing},
  we get the following corollary: 
  
\begin{corollary}\label{coro5}
Assume that there is an $H\subseteq S$ that satisfies (\ref{cond:many}).
Then there are at least $\Omega(\sigma 2^n)$ many points $z\in \{0,1\}^n$ such that, 
  with probability at least $\eta_1$ over the draw of a random $r$-partition $\bT$ of $S$,
  $z$ is $(\gamma_1,\delta_1)$-revealing for $\bT$.
\end{corollary}  
\begin{proof}
We consider two cases: $|\Good(S)| \geq ({\sigma}/100)\cdot 2^n$ and
  $|\Good(S)| \leq ({\sigma}/100)\cdot 2^n$.
Then the statement follows from Lemma \ref{lem:good-points} and Lemma \ref{lem:Lleftlarge}
 for these two cases respectively.
\end{proof}
  
The property stated in Corollary \ref{coro5}
  inspires our next procedure $\FindRevealing$ described in Figure \ref{fig:find-revealing}.
We use Corollary \ref{coro5} to prove the following lemma:  
  
\begin{lemma}\label{lem:many-good}
Let $f:\{0,1\}^n\rightarrow \{0,1\}$, $S\subseteq [n]$, $m\ge 1$,
  and $\alpha\in (0,1]$.
$\emph{\FindRevealing}\hspace{0.05cm}(f,S,m,\alpha)$ outputs either 
  ``fail'' or a pair $(\bQ,\bB)$ such that $\bQ\subseteq S$ is nonempty
  and $\bB$ contains one bichromatic edge for each variable in $\bQ$. 
The number of queries it uses can be bounded from above by
$$
\begin{cases}
\tilde{O}(|\bQ|)
  +\tilde{O}\big( |\bQ|\big/(\alpha m)\big) & \text{when it outputs a pair $(\bB,\bQ)$; and}\\[1 ex]
\tilde{O}(|S|)
  +\tilde{O}\big( {|S|}\big/{\alpha m}\big) & \text{when it outputs ``fail.''}  
\end{cases}
$$
Furthermore, when $S$ contains a subset $H$ that satisfies (\ref{cond:many}), 
  $\emph{\FindRevealing}\hspace{0.05cm}(f,S,m,\alpha)$  outputs~a pair $(\bB,\bQ)$ with probability at least $1-1/\poly(n)$.
\end{lemma}
\begin{proof}
It follows from Lemma \ref{lem:revealing} that
  when $\FindRevealing$ outputs a pair $(\bQ,\bB)$,
  $\bQ$ is a subset of $S$ and $\bB$ contains one bichromatic edge for each 
  variable in $\bQ$.
We also have that $\bQ$ is nonempty since we compare $|\bQ|$ with $\gamma_1 r^*/2>0$
  before returning $(\bB,\bQ)$.
  
The number of queries used  when $\FindRevealing$ outputs ``fail'' follows 
  from its description.
Next we bound the number of queries used when it outputs a pair $(\bQ,\bB)$.

Let $d'$ be the value of $d^*$ when $\FindRevealing$ terminates.
Then we have $|\bQ|\ge \gamma_1r'/2$, where $r'=100d'$.
The number of queries used by $\CheckRevealing$ in each iteration is at most
$$
\log^2 n\cdot \frac{6d^*\log n}{\alpha m} \cdot O(\log^2 n)\cdot  \tilde{O}(1) 
=\tilde{O}\left(\frac{d^*}{\alpha m}\right).
$$
As a result, the total number of queries used by $\CheckRevealing$ is at most
$$
\tilde{O}\left(\frac{1+2+\cdots+d'}{\alpha m}\right)=\tilde{O}\left(\frac{d'}{\alpha m}\right)=\tilde{O}\left(\frac{|\bQ|}{\alpha m}\right).
$$
On the other hand, we only make one call to $\GetRevealingEdges$ during the last iteration of $d'$.
So the number of queries it uses is at most $\tilde{O}(r')=\tilde{O}(|\bQ|)$.
It then follows that the total number of queries used can be
  bounded using the expression given in the lemma. 

It is left to show that $\FindRevealing$ returns a pair $(\bB,\bQ)$ with high probability.
Note~that~it returns ``fail'' for two cases.
Either all calls to $\CheckRevealing$ return ``fail'' or one of these~calls 
  returns ``accept'' but the next call to $\GetRevealingEdges$ returns ``fail.'' 
The first event happens with small probability because when $d^*=d$,
  we have $\sigma\ge \sigma^*$ from Lemma \ref{lem:reg-graph} and therefore,~with probability at least $1-\exp(-\log^2n)$ we~get a point $\bx$ that satisfies Corollary \ref{coro5}
   on line 2.
For~this $\bx$ with  probability $1-\exp(-\log^2 n)$ we get a $\bT$ such that 
  $\bx$ is $(\gamma_1,\delta_1)$-revealing for $\bT$ on line~3.
When this happens, it follows from Lemma \ref{lem:check} that the subroutine $\CheckRevealing$ outputs
  ``accept'' with probability at least $1-1/\poly(n)$. 
  
For the second event to happen, either one of the calls to $\CheckRevealing\hspace{0.05cm}(f,\bx,\bT,\gamma_1,\delta_1)$ returns
  ``accept'' while $\bx$ is not $(\gamma_1/2,\delta_1/2)$-revealing for $\bT$,
  or the only call to $\GetRevealingEdges$ $(f,\bx,\bT,\delta_1/2)$
  fails to find enough bichromatic edges while $\bx$ is indeed $(\gamma_1/2,\delta_1/2)$-revealing for $\bT$.
It follows from Lemma \ref{lem:revealing} and \ref{lem:check} and a union bound that this happens with low probability. 
\end{proof}

Next we prove Lemma \ref{lem:good-points}
  and Lemma \ref{lem:Lleftlarge}. We start with Lemma \ref{lem:good-points}.


\begin{proof}[Proof of Lemma \ref{lem:good-points}]
Consider the following procedure to draw an $r$-partition $\bT = \{ \bT_j \}_{j\in [r]}$ of $S$:
\begin{enumerate}
\item Sample $\gg_0 \colon S \setminus N(y) \to [r]$ uniformly at random; let $\bT^{(0)} = \big\{ \bT_j^{(0)} \big\}_{j\in [r]}$ with $\bT^{(0)}_j = \gg_0^{-1}(j)$.
\item Sample $\gg_1 \colon N(y) \to [r]$ uniformly at random; let $\bT^{(1)} = \big\{ \bT_{j}^{(1)} \big\}_{j\in [r]}$ with  $\bT^{(1)}_j = \gg_1^{-1}(j)$.
\item Let $\bT = \big\{ \bT_j \big\}_{j\in [r]}$ be given by $\bT_j = \bT_j^{(0)} \cup \bT_{j}^{(1)}$.
\end{enumerate}
We note that the above procedure samples a uniformly random $r$-partition $\bT$ of $S$. Consider the following set $\bP \subseteq [r]$ defined using $\bT^{(0)}$ and $\bT^{(1)}$:
\[ \bP = \left\{ i \in [r] : i \in \Reveal(y, \bT^{(0)}, \delta_0) \text{ and } \big|\bT_i^{(1)}\big| \leq 1\right\}. \]
We note that every $i \in \bP$ satisfies
\begin{align*} 
\Prx_{\bR \subseteq \bT_i} \left[ f(y) \neq f(y^{(\bR )})\right] &\geq \Prx_{\bR \subseteq \bT_i}\left[ \bR \cap \bT_i^{(1)} = \emptyset \right]\cdot \Prx_{\bR \subseteq \bT_i}\left[f(y) \neq f(y^{(\bR)})\hspace{0.06cm} \big|\hspace{0.06cm} \bR \cap \bT_i^{(1)} = \emptyset \right] \\[0.5ex]
&\geq 0.5 \cdot \Prx_{\bR \subseteq \bT_i^{(0)}} \left[ f(y) \neq f(y^{(\bR)}) \right] \geq \delta_0/2,
\end{align*}
where we used the fact that $i \in \bP$ implies $|\bT_i^{(1)}| \leq 1$ so that  $\bR \cap \bT_i^{(1)} = \emptyset$ with probability at least $1/2$, and the fact that $i \in \Reveal(y, \bT^{(0)}, \delta_0)$.

Therefore, it suffices to show that with probability at least $ {\eta_0}/{2}$ over the draw of $\bT$, $|\bP| \geq  \gamma_0 r/2$. Towards this goal, we consider the event $E$ which occurs when $y$ is~$(\gamma_0, \delta_0)$-revealing with respect to $\smash{\bT^{(0)}}$. Since $y \in \Good(S)$, the event $E$  occurs with probability at least $\eta_0$.

Fix an $r$-partition $T^{(0)}$ of $S\setminus N(y)$ such that $E$ occurs, and 
  let $\smash{\Reveal=\Reveal(y,T^{(0)},\delta_0)}$.
Then for each $\smash{j\in \Reveal}$ the probability
  of $\smash{|\bT^{(1)}_j|\ge 2}$ is at most (using $r=100d$)
$$
 \binom{2d}{2} \cdot \frac{1}{r^2}\le \frac{1}{5000}.
$$
So
  the number of $\smash{j\in \Reveal}$ with $\smash{|\bT^{(1)}_j|\ge 2}$ being more than
  $|\Reveal|/2$ can only happen with probability at most $1/2500$, and when this does not happen, the size of $\bP$ is at least
  $|\Reveal|/2\ge\gamma_0 r/2$.
Overall, this happens with probability at least 
  $\eta(1-1/2500)\ge \eta_0/2$.
\end{proof}

Finally we prove Lemma \ref{lem:Lleftlarge}.
We start with some notation.
Given a point $x \in U$, we let
\[ Y(x) = \big\{ x^{(i)} \in V : i \in N(x) \big\} \]
denote the set of neighbors of $x$ in $G$, where $ |Y(x)|=\deg_G(x)\in [d: 2d]$. 
Given a $y \in Y(x)$ and an $r$-partition $T=\{T_i\}$ of $S$,
  let $\smash{T^{(0,y)}=\{T_i^{(0,y)}\}}$ denote the $r$-partition of $S\setminus N(y)$
  with $$T_i^{(0,y)}=T_i\setminus N(y).$$


\begin{proof}[Proof of Lemma \ref{lem:Lleftlarge}]
Let $U' = \big\{ x \in U : |Y(x) \cap \Good(S)| \ge (3/4)\cdot |Y(x)|\big\}$. We note that:
\begin{align*}
|U'|\cdot d\le 
\sum_{x \in U'} \deg_G(x)\le |\Good(S)|\cdot 2d.
\end{align*}
Using $|\Good(S)| \leq  ({\sigma}/{100}) \cdot 2^n$
  and $|U|\ge \sigma 2^n$,
  at least $\sigma 2^n/2$ many points $x\in U$
  satisfy 
$|Y(x) \cap \Good(S)| \le  (3/4)\cdot |Y(x)|.$
We prove in the rest of the proof that every such $x\in U$
  is $(1/4000,(1-\delta_0)/2)$-revealing for $\bT$ with probability at least 
  $1/3$ over the draw of a random $r$-partition $\bT$ of $S$.

Given a partition $T = \{ T_j \}_{j=1}^r$ of $S$ we consider the following 
  subset $A(T)$ of $Y(x)\setminus \Good(S)$:
\[\left\{x^{(i)} = y \in Y(x) \setminus \Good(S) : \text{for the $k \in [r]$ with $i\in T_k$}, \hspace{-0.05cm}\begin{array}{cl} \text{(i)} &\hspace{-0.15cm} T_k \cap (N(x)\cup N(y)) = \{ i \} \\[0.6ex]
					 \text{(ii)} & \hspace{-0.15cm}k \notin \Reveal(y, T^{(0, y)}, \delta_0) \end{array} \hspace{-0.1cm} \right\}.\]
For each $x^{(i)} = y \in A(T)$, let $k \in [r]$ be the index with $i\in T_k$. Then we have 
\begin{align}
\Prx_{\bR \subseteq T_k}\left[ f(x) \neq f(x^{(\bR)})\right] &\geq \Prx_{\bR \subseteq T_k}\big[ i \in \bR\big] \cdot \Prx_{\bR \subseteq T_k}\left[ f(x) \neq f(x^{(\bR)})\hspace{0.08cm}\big|\hspace{0.08cm} i \in \bR\right] \nonumber \\[0.5ex]
&= 0.5 \cdot \Prx_{\bR' \subseteq T_k^{(0, y)}} \left[ f(x) \neq f(x^{(\bR' \cup \{ i\})}) \right] \label{eq:1} \\
	&= 0.5 \cdot \Prx_{\bR' \subseteq T_k^{(0, y)}} \left[ f(y) = f(y^{(\bR')})\right] \label{eq:2}\geq  (1 - \delta_0)\big/2.
\end{align}
Here (\ref{eq:1}) follows from the fact $i \in \bR$ with probability $1/2$; In that case, letting $\bR' = \bR \setminus \{i \}$~gives $x^{(\bR)} = y^{(\bR')}$, and because $T_k \cap N(y) = \{ i \}$ we have that $\bR \subseteq T_k$ conditioning on $i \in \bR$ is distributed as $\bR' \cup \{ i \}$ with $\bR' \subseteq \smash{T^{(0, y)}_k}$. Additionally, (\ref{eq:2}) follows from the fact that $f(x) \neq f(y)$ since $(x, y)$~is a bichromatic edge of $f$, and the fact that $k \notin \Reveal(y, T^{(0, y)}, \delta_0)$. 

Since $T_k \cap N(x) = \{ i\}$ the number of $k \in [r]$ for which (\ref{eq:2}) holds is at least $|A(T)|$. As a result, $x$ is $({|A(T)|}/{r}, ({1-\delta_0})/{2})$-revealing for $T$. Therefore, in order to show the lemma it suffices to show that 
$|A(\bT)| \geq {d}/{40}$ with probability at least $1/3$ over the draw of a random $r$-partition $\bT$ of $\bS$. 

For this purpose, we 
  consider a fixed $x^{(i)} = y \in Y(x) \setminus \Good(S)$ and show that
  $y\in A(\bT)$ with probability at least $1/2$. 
Similarly to the proof of Lemma~\ref{lem:good-points}, 
we can equivalently draw
a uniformly random $r$-partition $\bT = \{ \bT_{j} \}_{j\in [r]}$ of $S$ by first drawing
  two uniformly random $r$-partitions 
$$\bT^{(0,y)} = \big\{ \bT^{(0, y)}_j\big\}_{j\in [r]}\quad\text{and}\quad
\bT^{(1, y)} = \big\{ \bT^{(1,y)}_j\big\}_{j\in [r]}$$
of $S \setminus N(y)$ and $N(y)$, respectively, and then setting $\bT_j=\bT_j^{(0,y)}\cup \bT_j^{(1,y)}$ for each $j\in [r]$.

Then we have
\begin{align*}
&\hspace{-0.6cm}\Prx_{\bT} \big[ y \notin A(\bT)\big]\\ &\le  \Prx_{\bT}\big[ \text{the $k \in [r]$ with\ }i\in \bT_k\ \text{satisfies}\ \bT_k \cap (N(x)\cup N(y)) \neq \{i\}\big] \\
					     &\qquad + \Prx_{\bT}\big[ y \text{ is $(\gamma_0, \delta_0)$-revealing for } \bT^{(0,y)}\big] \\
					     &\qquad + \Prx_{\bT} \left[ y \text{ is not $(\gamma_0, \delta_0)$-revealing for } \bT^{(0,y)} \text{ and } \bT_k \ni i \text{ has } k \notin \Reveal(y, \bT^{(0, y)}, \delta_0)\right] \\[0.8ex]
					     &\le \left(\dfrac{|N(x) \cup N(y)| - 1}{r} \right) + \eta_0 + (1-\eta_0) \gamma_0 \le \frac{4d}{r} + \eta_0 + \gamma_0 <  {1}/{2},
\end{align*}
using $r=100d$.
Therefore, the expected size of $A(\bT)$ is at least $|Y(x) \setminus \Good(S)|/{2}.$ 
Writing $$\beta = \Prx_{\bT}\left[ |A(\bT)| \geq \dfrac{|Y(x) \setminus \Good(S)|}{10} \right],$$ we have
${1}/{2} \leq \beta + (1-\beta)/10$ and thus, 
$\beta \geq ({1}/{2}) - ({1}/{10}) >1/3.$
So with probability at least $1/3$,  
$$
A(\bT)\ge \frac{|Y(x) \setminus \Good(S)|}{10} \geq \frac{1}{10} \cdot \frac{|Y(x)|}{4} \ge \frac{d}{40}$$
since $|Y(x) \cap \Good(S)| \le  (3/4)\cdot |Y(x)|.$
This finishes the proof of the lemma.
\end{proof}

\subsection{The $\FindHiInf$ procedure}

\begin{figure}
\begin{framed}
\noindent Subroutine $\FindHiInf\hspace{0.05cm}(f,S, m,\alpha)$
\begin{flushleft}
\noindent {\bf Input:} Query access to $f \colon \{0,1\}^n \to \{0,1\}$, $S\subseteq [n]$,
  $m\ge 1$ and $\alpha\in (0,1]$. 

\noindent {\bf Output:} A set of bichromatic edges $\bB$ of $f$ whose variables form a subset $\bQ \subseteq S$.

\begin{enumerate}
\item Initialize $\bQ = \bB = \emptyset$, and $\bS^* = S$.
\item Repeatedly run $\FindRevealing\hspace{0.05cm}(f,\bS^*, m,\alpha)$:
\begin{itemize}
\item If it outputs ``fail,'' terminate and output $\bQ$ and $\bB$.
\item Otherwise, if $\FindRevealing$ outputs $(\bB', \bQ')$, update the sets as 
\[ \bQ \leftarrow \bQ \cup \bQ', \quad \bB \leftarrow \bB \cup \bB'\quad \text{and} \quad \bS^* \leftarrow \bS^* \setminus \bQ'. \]
\end{itemize}
\end{enumerate}
\end{flushleft}\vskip -0.14in
\end{framed}\vspace{-0.2cm}
\caption{The procedure $\FindHiInf$.}\label{fig:find-hi-influence}
\end{figure}

Finally we describe the procedure $\FindHiInf$ in Figure \ref{fig:find-hi-influence}
  and prove the following lemma.


\begin{lemma}
\label{lem:correct-hiinf}
Let $S \subseteq [n]$, $m\ge 1$ and $\alpha\in (0,1]$.
Then $\emph{\FindHiInf}\hspace{0.05cm}(f,S, m,\alpha)$ makes at most 
\[ \tilde{O}\left(|S|+ \frac{|S| }{\alpha m}\right) \]
queries to $f$.
When $S$ contains a subset $H$ that satisfies (\ref{cond:many}), 
  with probability at least $1-1/\poly(n)$
  $\emph{\FindHiInf}\hspace{0.05cm}(f,S, m,\alpha)$ outputs $(\bQ,\bB)$
  such that $$|\bQ \cap H| \geq |H| - m$$
  and $\bB$ contains one bichromatic edge for each variable in $\bQ$.
\end{lemma}
\begin{proof}
First note that the procedure terminates once $\FindRevealing$ outputs ``fail.''
On the other hand, whenever $\FindRevealing$ outputs a pair $(\bB', \bQ')$, $|\bS^*|$ decreases by $|\bQ'|$.
It then follows~from Lemma \ref{lem:many-good}
  that the total number of queries used by calls to $\FindRevealing$
  except the last one that outputs ``fail'' 
  is $\tilde{O}(|S|)+\tilde{O}(|S|/(\alpha m))$ since 
  the total size of $\bQ'$'s they output is at most $|S|$.
By Lemma \ref{lem:many-good}, the number of queries used
  by the last call can be bounded by the same expression.


Suppose that at some moment of the execution, we have
  $|\bQ \cap H| < |H| - m$. Then $\bS^* = S \setminus \bQ$
  satisfies $|\bS^* \cap H|>m$. 
This implies that, for $|\bQ\cap H|<  |H|-m$ to happen at the end,
  a necessary condition is that one of the calls to $\FindRevealing$
  has $|\bS^*\cap H|>m$ in input but outputs~``fail.''
It follows from Lemma \ref{lem:many-good} and a union bound on
  at most $|S|\le n$ many calls to $\FindRevealing$ that this happens with probability 
  at most $1/\poly(n)$. 
\end{proof}


\section{The Algorithm for Case 2}\label{algsec:case2}

\newcommand{\SecondAlg}{\mathtt{AlgorithmCase2}}

\begin{figure}[t!]
\begin{framed}
\noindent Procedure $\SecondAlg(f)$

\begin{flushleft}
\noindent {\bf Input:} Query access to a Boolean function $f \colon \{0, 1\}^n \to \{0, 1\}$

\noindent {\bf Output:} Either ``unate,'' or two edges constituting an edge violation of $f$ to unateness.\begin{enumerate}
\item Repeat the following $O(1)$ times:
\item \ \ \ \ \ \ \ \ Draw $\bS \subset [n]$ of size $n^{2/3}$ uniformly at random, and let $k = \lceil \frac{|\calI^*|}{n^{1/3} \log n} \rceil$. 
\item \ \ \ \ \ \ \ \ Let $\left(\bQ, \bB\right) \leftarrow \FindHiInf\left(f, \bS, k, \alpha\right)$.
\item \ \ \ \ \ \ \ \ Repeat $O\left(\sqrt{n} \Lambda^{14}/\eps^2\right)$ times:
\item \ \ \ \ \ \ \ \ \ \ \ \ \ \ \ \ Sample $\bT \subset \bS$ uniformly at random of size $\left\lceil \frac{\sqrt{n}}{2^s} \right\rceil$.
\item \ \ \ \ \ \ \ \ \ \ \ \ \ \ \ \ Sample $\bx\in \{0,1\}^n$ uniformly at random and run 
  $\AESearch(f,\bx,\bT)$
\item \ \ \ \ \ \ \ \ Let $\bA_{+}$ be the set of $i\in [n]$ such that a monotone edge along variable $i$ is found
\item \ \ \ \ \ \ \ \ Repeat $O\left(\sqrt{n} \Lambda^{14}/\eps^2\right)$ times:
\item \ \ \ \ \ \ \ \ \ \ \ \ \ \ \ \ Sample $\bT \subset \bS$ uniformly at random of size $\left\lceil \frac{\sqrt{n}}{2^t} \right\rceil$.
\item \ \ \ \ \ \ \ \ \ \ \ \ \ \ \ \ Sample an $\bx\in \{0,1\}^n$ uniformly at random and run
  $\AESearch(f,\bx, \bT)$
\item \ \ \ \ \ \ \ \ Let $\bA_{-}$ be the set of $i\in [n]$ such that an anti-monotone edge along variable $i$ is found
\item \ \ \ \ \ \ \ \ Output an edge violation of $f$ to unateness if one is found in $\bB, \bA_{+}$ and $\bA_{-}$.
\item Output ``unate.''
\end{enumerate}
\end{flushleft}\vskip -0.14in
\end{framed}\vspace{-0.2cm}
\caption{Algorithm for Case 2}\label{fig:algcase2}
\end{figure}

Below, we study Case 2 of the algorithm as described in Section~\ref{sec:mainalg}. Let $s \geq t \in [\Lambda]$, $h \in [3\Lambda]$, and $\ell \in [\lfloor \log n \rfloor]$. We assume in Case 2 that the input function $f \colon \{0,1\}^n \to \{0,1\}$ satisfies Lemma~\ref{lem:score-summary} with parameters $s, t, h, \ell$ on a set $\calI \subset [n]$ of size $|\calI| = 2^{\ell}$. Every variable $i \in \calI$ satisfies
\begin{align} 
\Score_{i, s}^{+} \geq 2^{s-h} \qquad\qquad \Score_{i, t}^{-} \geq 2^{t-h}, \qquad\text{and}\qquad \frac{|\calI|}{2^h} = \Omega\left( \frac{\eps^2}{\Lambda^{11}}\right) .\label{eq:score-remind}
\end{align}
We assume 
\begin{align}
|\calI| \cdot \frac{1}{2^s} \geq |\calI| \cdot \frac{1}{2^t} \geq n^{2/3}, \label{eq:size-I-lb}
\end{align}
and that for al least half of the variables $i \in \calI$, 
\begin{align} 
\Inf_f(i) \geq \alpha := \frac{\eps^2 \cdot n^{1/3}}{|\calI| \cdot \Lambda^{13}}. \label{eq:high-inf}
\end{align}

Similarly to case 1 in Section~\ref{algsec:case1}, we consider the subset $\calI^* \subseteq \calI$ of size at least $\lceil \frac{|\calI|}{2}\rceil$ satisfying (\ref{eq:high-inf}) for all $i \in \calI^*$. 

This algorithm $\SecondAlg$, show in Figure~\ref{fig:algcase2}, finds an edge violation with high probability. At a high level, the algorithm first samples a uniformly random set $\bS$ of siz e $n^{2/3}$ and uses $\FindHiInf$ to find bichromatic edges along almost all variables in $\bS \cap \calI^*$. Suppose first, that most of the bichromatic edges from $\bS \cap \calI^*$ found during $\FindHiInf$ are anti-monotone edges. Then via a similar analysis to \cite{CWX17b}, after running $\AESearch(f, \bx, \bS)$ on a uniform random sets $\bT \subset \bS$ of size $\lceil \sqrt{n}/2^s\rceil$ and a uniform point $\bx \sim \{0,1\}^n$ for $\widetilde{O}(\sqrt{n}/\eps^2)$ many iterations, we expect to find a monotone edge along some direction in $\bS \cap \calI^*$. Since the algorithm had already found an anti-monotone edge along many variables in $\bS \cap \calI^*$, the algorithm will very likely find a violation to unateness. Similarly, if most of the bichromatic edges along variables in $\bS \cap \calI^*$ during $\FindHiInf$ are monotone, then $\AESearch(f, \bx, \bT)$ will likely find an edge violation when $\bT$ is a random set of $\bS$ of size $\lceil \sqrt{n}/2^t\rceil$ after $\tilde{O}(\sqrt{n}/\eps^2)$ iterations.

\begin{lemma}[Query complexity of $\SecondAlg$]
$\SecondAlg(f)$ makes 
\[ \tilde{O}\left( n^{2/3}/\eps^2\right)\] 
queries to $f$.
\end{lemma}

\begin{proof}
The query complexity of $\SecondAlg(f)$ follows from the description in Figure~\ref{fig:algcase2}. In particular, we lines 4--12 make a total of $\tilde{O}(\sqrt{n} / \eps^2)$, so it remains to upper-bound the query complexity of line $3$ invoking $\FindHiInf$. By Lemma~\ref{lem:correct-hiinf}, $\FindHiInf(f, \bS, k, \alpha)$, where 
\[ |\bS| = n^{2/3} \qquad\text{and} \qquad k = \lceil\frac{|\calI^*|}{n^{1/3} \log n}\rceil, \]
makes
\[ \tilde{O}\left( |\bS| + \dfrac{|\bS|}{\alpha k}\right) = \tilde{O}\left( n^{2/3}/\eps^2 \right), \]
using the assumptions of Case 2 in (\ref{eq:score-remind}), (\ref{eq:size-I-lb}), and (\ref{eq:high-inf}). 
\end{proof}

We start the analysis of $\SecondAlg$ by defining the notion of \emph{informative sets} for case 2. Recall that, in Section~\ref{algsec:case1}, Definition~\ref{def:informative} gave a different definition of informative sets for case 1. Since these definitions serve very similar purposes in the analysis of the algorithm, we use the same name. Furthermore, for $T \subset [n]$ of size $\lceil \frac{\sqrt{n}}{2^s} \rceil$, $\PE_i^+(T)$ is the set of $s$-strong monotone edges along variable $i$ which are $T$-persistent. Similarly, for $T \subset [n] \setminus \{i\}$ of size $\lceil \frac{\sqrt{n}}{2^t} \rceil$, $\PE_i^-(T)$ is the set of $t$-strong anti-monotone edges along variable $i$ that are $T$-persistent. Note that the definitions of these sets are slightly different than in Subsection~\ref{sec:inform}. 

\begin{definition}\label{def:inform-2}
We say that a set $S \subset [n] \setminus \{ i \}$ of size $(n^{2/3} -1)$ is $i$-informative if the following two conditions hold:
\begin{itemize}
\item[i.] with probability at least $1/10$ over the draw of $\bT \subset S$ of size $\lceil\sqrt{n}/2^s\rceil$, $|\PE_{i}^+(\bT)| \geq \frac{2^{s-h}}{10} \cdot 2^n$.
\item[ii.] with probability at least $1/10$ over the draw of $\bT \subset S$ of size $\lceil \sqrt{n}/2^t\rceil$, $|\PE_i^-(\bT) \geq \frac{2^{t-h}}{10} \cdot 2^n$. 
\end{itemize}
\end{definition}

\ignore{We note the following two observations which we use in this section. First, by (2) in Definition~\ref{def:strong-edge} and (\ref{eq:score-def}), 
\begin{align} 
\frac{|\PE_i^+(T)|}{2^n} \leq \Score_{i, s}^+\text{ for all } T \in \calP_{i,\lceil \sqrt{n}/2^s \rceil}, \qquad
\frac{|\PE_i^-(T)|}{2^n} \leq \Score_{i, t}^- \text{ for all } T \in \calP_{i,\lceil\sqrt{n}/2^t \rceil}. \label{eq:pe-upper}
\end{align} 
Second, by (2) in Definition~\ref{def:strong-edge}, 
\begin{align}
\Ex_{\bT \sim \calP_{i, \lceil \frac{\sqrt{n}}{2^s}\rceil}}\left[ \dfrac{|\PE_i^+(\bT)|}{2^n}\right] \geq (1-o(1)) \Score_{i, s}^+ \text{,  }\Ex_{\bT \sim \calP_{i, \lceil \frac{\sqrt{n}}{2^t}\rceil}}\left[ \dfrac{|\PE_i^-(\bT)|}{2^n}\right] \geq (1-o(1)) \Score_{i, t}^-. \label{eq:pe-lower}
\end{align}
}
\begin{lemma}\label{lem:informative-2}
For every $i \in \calI^*$, when sampling $\bS \subset [n]$ uniformly of size $n^{2/3}$, we have
\[ \Prx_{\bS \subset [n]}\left[i\in \bS \text{ and } \bS \setminus \{ i \} \text{ is $i$-informative} \right] \geq \frac{1}{2n^{1/3}}. \]
\end{lemma}

\begin{proof}
Recall that for $m \in \N$, $\calP_{i, m}$ is the uniform distribution over subsets of $[n] \setminus \{i\}$ of size $m-1$. For $m = n^{2/3}$, we define the quantity:
\[ \gamma = \Prx_{\bS \sim \calP_{i, m}}\left[ \bS \text{ does not satisfy (i) in Definition~\ref{def:inform-2}}\right]. \]
For $m_1 = \left\lceil \frac{\sqrt{n}}{2^s} \right\rceil \leq m$, consider the quantity
\[ \Prx_{\be, \bT}\left[ \be \in \PE_{i}^+(\bT)\right],\]
where $\be$ is an $s$-strong monotone edge sampled uniformly at random, and $\bT \sim \calP_{i, m_1}$. By the definition of $s$-strong monotone edges, we have the above probability is at least $1-o(1)$. On the other hand, we may use the definition of $\gamma$ to upper bound the above probability by $\gamma (\frac{1}{10} + \frac{9}{10} \cdot \frac{1}{10}) + (1-\gamma)$, which implies $\gamma = o(1)$. Analogously, letting $m_2 = \lceil \frac{\sqrt{n}}{2^t}\rceil$, we define the quantity:
\[ \beta = \Prx_{\bS \sim \calP_{i, m_2}}\left[\bS \text{ does not satisfy (ii) in Definition~\ref{def:inform-2}} \right]. \]
Similarly as above, we considering $\Pr[\be \in \PE_i^-(\bT)]$, for $\bT \sim \calP_{i, m_2}$ and $\be$ a uniformly random $t$-strong anti-monotone edge, to  conclude $\beta = o(1)$.
\ignore{we may now bound $\Ex_{\bT}\left[ \frac{|\PE_i^+(\bT)|}{2^n}\right]$, where $\bT \sim \calP_{i, m_{1}}$ in two ways:
\begin{enumerate}
\item By sampling $\bT \sim \calP_{i, m_{1}}$, and
\item By first sampling $\bS \sim \calP_{i, m}$, and then considering $\bT \subset \bS$ uniformly of size $m_{1}$. 
\end{enumerate}
The above methods, in turn, correspond to the following inequalities:
\begin{align}
\Ex_{\bT \sim \calP_{i, m_{1}}}\left[ \dfrac{|\PE_{i}^+(\bT)|}{2^n}\right] &\geq (1-o(1))\cdot \Score_{i, s}^{+}, \label{eq:pe-lb} \\
\Ex_{\bT \sim \calP_{i, m_{1}}}\left[ \dfrac{|\PE_{i}^+(\bT)|}{2^n} \right] &\leq \gamma \left( \frac{1}{10} \cdot \Score_{i, s}^{+} + \frac{9}{10} \cdot \frac{1}{10} \cdot \Score_{i, s}^+\right) + (1 - \gamma) \cdot \Score_{i, s}^+, \label{eq:pe-ub}
\end{align}
where (\ref{eq:pe-lb}) follows from (\ref{eq:pe-upper}), and (\ref{eq:pe-ub}) follows from first sampling $\bS \sim \calP_{i, m}$ and conditioning on whether $\bS$ satisfies (i) in Definition~\ref{def:inform-2} or not, as well as the upper bound in (\ref{eq:pe-upper}). By combining (\ref{eq:pe-lb}) and (\ref{eq:pe-ub}), we conclude $\gamma \leq o(1)$. 

We may analogously define
\[ \beta = \Prx_{\bS \sim \calP_{i, m}}\left[ \bS \text{ does not satisfy (ii) in Definition~\ref{def:inform-2}}\right].\]
When $m_2 = \left\lceil \frac{\sqrt{n}}{2^t} \right\rceil \leq m$, we may similarly bound $\Ex_{\bT}\left[ \frac{|\PE_{i}^{-}|}{2^n} \right]$, where $\bT \sim \calP_{i, m_2}$, which similarly gives $\beta \leq o(1)$.} 
Therefore, the probability over $\bS\sim\calP_{i,m}$ that $\bS$ is $i$-informative is at least $1 - \gamma - \beta \geq 1-o(1)$. 
Finally,
\begin{align*}
\Prx_{\bS \subset [n]}\left[ i \in \bS \text{ and } \bS \setminus \{i\} \text{ is $i$-informative}\right] &\geq \Prx_{\bS \subset [n]}\left[ i \in \bS\right] \cdot \Prx_{\bS' \sim \calP_{i, m}}\left[ \bS' \text{ is $i$-informative}\right] \geq \frac{1}{2n^{1/3}}.
\end{align*}
\end{proof}

\newcommand{\bcalJ}{\boldsymbol{\calJ}}

When sampling a set $\bS \subset [n]$ of size $n^{2/3}$ in line 2 of $\SecondAlg(f)$, we may define the set
\begin{align} 
\bcalJ &= \left\{ i \in \calI^* : i \in \bS \text{ and } \bS \setminus \{i \} \text{ is $i$-informative} \right\}. \label{eq:def-j}
\end{align}
By Lemma~\ref{lem:informative-2}, we immediately obtain
\begin{align} 
\Ex_{\bS \subset [n]}\left[|\bcalJ|\right] \geq \dfrac{|\calI^*|}{2n^{1/3}} = \Omega\left(\dfrac{2^h \cdot \eps^2}{n^{1/3} \cdot \Lambda^{11}}\right), \label{eq:ex-j-lb}
\end{align}
where the second inequality follows from (\ref{eq:score-remind}). We consider the event $\calbE$, defined over the randomness of sampling $\bS$ in line 2, which occurs when $|\bS \cap \calI^*| \leq \frac{4|\calI^*|}{n^{1/3}}$ and $|\bcalJ| \geq \frac{|\calI^*|}{10n^{1/3}}$.

\begin{lemma}\label{lem:good-sample}
In every iteration of line 2 of $\SecondAlg(f)$, event $\calbE$ occurs with probability at least $\Omega(1)$.
\end{lemma}

\begin{proof}
By (\ref{eq:size-I-lb}) and the fact that $\frac{1}{2^s} \leq 1$, $\Ex_{\bS}[|\bS \cap \calI^*| ] \geq |\calI^*|/n^{1/3} \geq n^{1/3}$. Thus, by Lemma~\ref{appendix1}, $|\bS \cap \calI^*| \leq 4|\calI^*|/n^{1/3}$ with probability at least $1- \exp\left( -\Omega(n^{1/3})\right)$.
Let $\beta$ be the probability over $\bS \subset [n]$ that $|\bcalJ| \geq |\calI^*|/(10n^{1/3})$. 
We may then upper bound $\Ex[|\bcalJ|]$ using the definition of $\beta$ by
\[ \exp(-\Omega(n^{1/3})) \cdot n^{2/3} + \beta \cdot \frac{4|\calI^*|}{n^{1/3}} + (1-\beta) \cdot \frac{|\calI^*|}{10n^{1/3}}, \]
which, combined with the lower bound in (\ref{eq:ex-j-lb}), implies $\beta \geq \Omega(1)$. Thus, the probability $\calbE$ occurs is at least $\beta - \exp(-\Omega(n^{1/3})) = \Omega(1)$.
\end{proof}

Thus, consider some iteration of $\SecondAlg(f)$ where event $\calbE$ occurs, and by Lemma~\ref{lem:good-sample} there exists such an iteration with high constant probability. The rest of this section is devoted to showing that the iteration of $\SecondAlg(f)$ where $\calbE$ occurs will find a violation to unateness with high probability. 

\begin{lemma}
Suppose that a particular iteration of $\SecondAlg(f)$, event $\calbE$ occurs. At that iteration, $\SecondAlg(f)$ finds a violation with high probability.
\end{lemma}

\begin{proof}
Fix the particular iteration where event $\calbE$ occurs. We apply Lemma~\ref{lem:correct-hiinf} with the hidden set $H = \bcalJ$ to conclude that at line 3 of $\SecondAlg$, the set of variables $\bQ$ for which bichromatic edges are observed during $\FindHiInf$ satisfies 
\begin{align}
 |\bQ \cap \bcalJ| \geq |\bcalJ| - k \geq \frac{|\calI^*|}{20n^{1/3}}, \label{eq:sizelb}
 \end{align}
 with probability at least $1 - 1/\poly(n)$, where in the last inequality, we used the fact that $k \leq |\bcalJ|/2$, as well as the fact that $|\bcalJ| \geq |\calI^*|/(10n^{1/3})$ when $\calbE$ occurs. Assume that for most of the variables in $\bQ \cap \bcalJ$, $\bB$ contains \emph{anti-monotone} edges in these variables, and let $\bC \subset \bQ \cap \bcalJ$ be the set of these variables, which by assumption, 
\begin{align} 
|\bC| &\geq \frac{|\calI|}{40 n^{1/3}}. \label{eq:c-sizelb}
\end{align}
The case when most variables in $\bQ \cap \bcalJ$ contain monotone edges will follow by a symmetric argument. We will now show that during the execution of lines 4--7 of $\SecondAlg(f)$, $\bA_{+}$ will contain a monotone edge along some variable in $\bC$ with high probability. 

Towards this goal, consider a particular execution of line 5 which samples a set $\bT\subset \bS$ of size $m_1 = \left\lceil \frac{\sqrt{n}}{2^s} \right\rceil$ uniformly at random, and let
\[ \bD = \left\{ i \in \bC \cap \bT : \dfrac{|\PE_{i}^{+}(\bT \setminus \{i\})|}{2^n} \geq \frac{2^{s-h}}{10}\right\}. \]
We note that since $\bC \subset \bcalJ$, every $i \in \bC$ satisfies
\begin{align*}
\Prx_{\bT \subset \bS}\left[ i \in \bD \right] &= \Prx_{\bT \subset \bS}\left[ i \in \bT \right] \cdot \Prx_{\bT' \sim \calP_{i, m_1}}\left[ \dfrac{|\PE_{i}^+(\bT')|}{2^n} \geq \frac{2^{s-h}}{10}\right]\\
	&\geq \frac{m_1}{n^{2/3}} \cdot \frac{1}{10} \geq \frac{1}{10 \cdot 2^s \cdot n^{1/6}}. 
\end{align*}
which implies that the parameter
\begin{align} 
\beta = \Ex_{\bT \subset \bS}\left[ |\bD| \right] \geq |\bC| \cdot \frac{1}{10 \cdot 2^s \cdot n^{1/6}} \geq \dfrac{|\calI^*|}{400 \cdot 2^s \cdot \sqrt{n}} \label{eq:alpha-lb}
\end{align}
by (\ref{eq:c-sizelb}), and similarly, note that $\Ex[|\bC \cap \bT|] = \Theta(\beta)$.

Suppose first that $\beta \leq \log^2 n$, so that $\Ex[|\bC \cap \bT|] = O(\log^2 n)$. Then, by Lemma~\ref{appendix1}, $|\bD| \leq O(\log^2 n)$ with probability at least $1 - \exp\left(-\Omega(\log^2 n)\right)$.
This implies that during the execution of line 5 and 6, 
\begin{align}
\Prx_{\substack{\bT \subset \bS \\ \bx \sim \{0,1\}^n}}\left[\exists i \in \bD \text{ and } \bx \in \PE_{i}^+(\bT \setminus \{i\}) \right] &\geq \Prx_{\bT\subset \bS}\left[ |\bD| \geq 1\right] \cdot \frac{2^{s-h}}{10} = \Omega\left( \frac{\beta}{\log^2 n} \cdot \frac{2^{s-h}}{10}\right) \label{eq:alpha-small-2}\\
	&\geq \Omega\left(\frac{1}{\log^2 n} \cdot \frac{|\calI^*|}{400 \cdot 2^s \cdot \sqrt{n}} \cdot \frac{2^{s-h}}{10}\right) \label{eq:alpha-small-3}\\
	&\geq \Omega\left( \frac{|\calI^*|}{2^{h} \cdot \sqrt{n} \cdot \log^2 n}\right) \geq \Omega\left(\frac{\eps^2}{\Lambda^{13} \sqrt{n}}\right). \label{eq:alpha-small-4}
\end{align}
We note (\ref{eq:alpha-small-2}) follows from the fact that $i \in \bD$ implies $\PE_{i}^{+}(\bT \setminus \{i \}) \geq 2^n \cdot \frac{2^{s-h}}{10}$ and the fact that $\Prx[| \bD|\geq 1] =\Omega(\frac{\beta}{\log^2 n})$; (\ref{eq:alpha-small-3}) follows from (\ref{eq:alpha-lb}); and, (\ref{eq:alpha-small-4}) follows from (\ref{eq:score-remind}). Thus, at least one iteration of the $O(\frac{\sqrt{n} \cdot \Lambda^{14}}{\eps^2})$ iterations of lines 5 and 6 will output a monotone edge in some variable in $\bC$ with high probability.

Suppose that $\beta \geq \log^2 n$. In this case, with probability at least $1 - \exp\left( - \Omega(\log^2 n) \right)$, 
\[ |\bD| \geq \frac{\beta}{4}, \]
and when this occurs,
\begin{align*} 
\Prx_{\bx \sim \{0,1\}^n}\left[\exists i \in \bD \text{ and }\bx \in \PE_{i}^{+}\left( \bT \setminus\{i\} \right) \right] \geq \frac{\beta}{4} \cdot \frac{2^{s - h}}{10} = \Omega\left( \dfrac{\eps^2}{\Lambda^{11} \sqrt{n}}\right)
\end{align*}
by (\ref{eq:alpha-lb}) and (\ref{eq:score-remind}). Thus, in this case again, we may conclude that at least one iteration of lines 5 and 6 will output a monotone edge from $\bC$ with high probability.
\end{proof}


\newcommand{\ThirdAlg}{\texttt{AlgorithmCase3}}

\section{The Algorithm for Case 3}\label{algsec:case3}

Below, we prove correctness of $\ThirdAlg(f)$, which covers Case 3 of the algorithm. We let $s \geq t \in [\Lambda]$, $h \in [3\Lambda]$, and $l \in [\lceil \log n \rceil]$ be parameters, so that $f \colon \{0,1\}^n \to \{0,1\}$ is $\eps$-far from unate and satisfies Lemma~\ref{lem:score-summary} with $s, t, h, \ell$, for a hidden set $\calI \subset [n]$ of size $|\calI| = 2^{\ell}$. Similarly to case 1 and case 2, we assume the parameters $s, t, h, \ell$, and the set $\calI$ satisfy that every $i \in \calI$,
\begin{align}
\Score_{i, s}^+ \geq 2^{s-h} \qquad\qquad \Score_{i, t}^{-} \geq 2^{t-h} \qquad\text{and}\qquad \frac{|\calI|}{2^h} = \Omega\left( \frac{\eps^2}{\Lambda^{11}}\right) . \label{eq:score-remind-3}
\end{align}
Lastly, we assume that $\calI$ is not too large, i.e.,
\begin{align} 
|\calI| \cdot \frac{1}{2^s} \leq |\calI| \cdot \frac{1}{2^t} \leq n^{2/3}.
\label{eq:small-I-3}
\end{align}

Instantiating the algorithm of \cite{CWX17b} with the additional assumptions corresponding to Case 3 from Section~\ref{sec:mainalg} would give the desired upper bound on the query complexity. Specifically, one may derive from (\ref{eq:small-I-3}), that $\Ex_{\bS \subset [n]}\left[ |\bS \cap \calI| \right] \leq n^{1/6}$, which corresponds to the parameter $\alpha$ in \cite{CWX17b}. As a result of Fact~5.4 and Fact 5.12, the query complexity of the algorithm in \cite{CWX17b} is
\[ \tilde{O}(\sqrt{\alpha n}/\eps^2) \leq \tilde{O}(n^{7/12}/\eps^2) \qquad\text{and}\qquad \tilde{O}(\sqrt{n}/\eps^2),\]
 which are $\tilde{O}(n^{2/3}/\eps^2)$. However, there is a (minor) technical caveat in the different definitions for strong edges and $\Score$ in Definition~\ref{def:strong-edge} and the analogous definitions in \cite{CWX17b}. For the sake of completeness, we include a simple algorithm achieving an $\tilde{O}(n^{2/3}/\eps^2)$-query upper bound in Figure~\ref{fig:algcase3}.
 
\begin{figure}[t!]
\begin{framed}
\noindent Procedure $\ThirdAlg(f)$

\begin{flushleft}
\noindent {\bf Input:} Query access to a Boolean function $f \colon \{0, 1\}^n \to \{0, 1\}$

\noindent {\bf Output:} Either ``unate,'' or two edges constituting an edge violation of $f$ to unateness.\begin{enumerate}
\item Repeat the following $O\left(\lceil \frac{2^t \sqrt{n} \log^2 n}{|\calI|}\rceil\right)$ times:
\item \ \ \ \ \ \ \ \ Draw $\bT \subset [n]$ of size $\lceil \frac{\sqrt{n}}{2^t}\rceil$ uniformly at random.
\item \ \ \ \ \ \ \ \ Repeat $O\left(\frac{2^{h}}{2^t} \log^2 n\right)$ times:
\item \ \ \ \ \ \ \ \ \ \ \ \ \ \ \ \ Sample $\bx\in \{0,1\}^n$ uniformly at random and run 
  $\AESearch(f,\bx,\bT)$
\item \ \ \ \ \ \ \ \ Let $\bA$ be the set of $i\in [n]$ such that an anti-monotone edge along variable $i$ is found.
\item \ \ \ \ \ \ \ \ Repeat $O\left( 2^s / 2^t\right)$ times:
\item \ \ \ \ \ \ \ \ \ \ \ \ \ \ \ \ Draw $\bS \subset \bT$ of size $\lceil \frac{\sqrt{n}}{2^s}\rceil$ uniformly at random.
\item \ \ \ \ \ \ \ \ \ \ \ \ \ \ \ \ Repeat $O\left(\frac{2^h}{2^s} \log^2 n\right)$ times:
\item \ \ \ \ \ \ \ \ \ \ \ \ \ \ \ \ \ \ \ \ \ \ \ Sample an $\by\in \{0,1\}^n$ uniformly at random and run
  $\AESearch(f,\by, \bS)$
\item \ \ \ \ \ \ \ \ Let $\bB$ be the set of $i\in [n]$ such that a monotone edge along variable $i$ is found
\item \ \ \ \ \ \ \ \ Output an edge violation of $f$ to unateness if one is found in $\bB$.
\item Output ``unate.''
\end{enumerate}
\end{flushleft}\vskip -0.14in
\end{framed}\vspace{-0.2cm}
\caption{Algorithm for Case 3}\label{fig:algcase3}
\end{figure}

\begin{lemma}[Query complexity of $\ThirdAlg$]
$\ThirdAlg(f)$ makes at most
\[ O\left( \left\lceil\frac{2^t \sqrt{n} \log^2 n}{|\calI|} \right\rceil\right) \left( \tilde{O}\left( \frac{2^h}{2^t}\right) +  O\left(\frac{2^s}{2^t}\right) \cdot \tilde{O}\left( \frac{2^h}{2^s}\right)\right) =  \tilde{O}(\sqrt{n}/\eps^2) + \tilde{O}(n^{2/3}). \]
queries to $f$
\end{lemma}

\begin{proof}
The query complexity upper bound is divided into two cases. If $2^t \sqrt{n}\log^2 = \Omega(|\calI|)$, then the query complexity is $\tilde{O}\left(\sqrt{n} / \eps^2 \right)$. Otherwise, the query complexity is $\tilde{O}(\frac{2^h}{2^t}) = \tilde{O}(|\calI|/2^t)$ and the bound follows from (\ref{eq:small-I-3}).
\end{proof}

We will use the definition of $i$-informative for monotone and anti-monotone edges given in Definition~\ref{def:informative}. The following lemma simply follows from applying Lemma~\ref{lem:informative-for-antimon} and Lemma~\ref{lem:informative-for-mon}, and taking a union bound.

\begin{claim}
With probability $1 - o(1)$ over the draw of $\bT$ and $\bS$ in lines 2 and 3 conditioned on $i \in \bS \subset \bT$, $\bS\setminus\{i\}$ is $i$-informative for monotone edges and $\bT\setminus\{i\}$ is $i$-informative for anti-monotone edges.
\end{claim}

Furthermore, assuming that $i \in \bS \subset \bT$ is sampled in line 2 and 3, where $\bT \setminus \{i\}$ is $i$-informative for anti-monotone edges, and $\bS \setminus \{i\}$ is $i$-informative for monotone edges, it follows that there exists two sets of points $X_i$ and $Y_i$ of size $\Omega(\frac{2^{t}}{2^h}) \cdot 2^n$ and $\Omega(\frac{2^s}{2^h}) \cdot 2^n$, respectively, such that if $\bx \sim X_i$ and $\by \sim Y_i$ are sampled in lines 5 and 8, an edge violation is found with probability at least $\Omega(1)$. Since lines 5 and 8 are repeated sufficiently many times, such points $\bx$ and $\by$ will be sampled from $X_i$ and $Y_i$, respectively. Therefore, since $\bS$ is sampled at least $\tilde{O}(\max\{2^s/2^t, 2^s\sqrt{n}/|\calI|\})$ times, it suffices to prove the following claim.
\begin{claim}
With probability $\Omega( \max\{ 1, \frac{|\calI|}{2^s \sqrt{n} \log^2 n} \})$ over the draw of $\bT$ and $\bS$ in lines 2 and 3, $\bS \cap \calI \neq \emptyset$.
\end{claim}

\begin{proof}
Let $\gamma$ be the probability over $\bS \subset [n]$ that $\bS \cap \calI \neq \emptyset$, which we will lower bound in the remainder of the proof. Consider the quantity
\[ \Ex_{\bS\subset[n]}\left[ |\bS \cap \calI| \right], \]
and note that since $\bS$ is a uniform random subset of $[n]$ of size $\lceil\sqrt{n}/2^s\rceil$, the above expectation is at least $|\calI|/(2^s \sqrt{n})$.
By Lemma~\ref{appendix1}, $|\bS \cap \calI| \leq 4\max\{ \log^2 n, |\calI| /(2^s \sqrt{n})\}$ with probability at least $1 - \exp\left( -\Omega(\log^2 n)\right)$. As a result, we may upper bound the above expectation using the definition of $\gamma$ by
\[ n \exp\left( -\Omega(\log^2 n)\right) + \gamma \cdot 4\max\{\log^2 n, |\calI|/(2^s \sqrt{n})\} \]
which gives the desired lower bound on $\gamma$. 
\end{proof}


\section*{Acknowledgements} We would like to thank Jinyu Xie for useful discussions. This work is supported in part by the NSF Graduate Research Fellowship under Grant No. DGE-16-44869 and NSF awards CCF-1703925 and IIS-1838154.

\begin{flushleft}
\bibliographystyle{alpha}
\bibliography{waingarten}
\end{flushleft}
\appendix

\section{Adaptive Edge Search}\label{sec:aes}

For completeness we present the adaptive edge search algorithm, $\AESearch$ in Figure~\ref{fig:edge-search}, which first appeared in \cite{CWX17b}. 
We recall the lemma and include its proof below.
\begin{figure}[t!]
\begin{framed}
\noindent Subroutine \AESearch$\hspace{0.05cm}(f,x, S)$
\begin{flushleft}\noindent {\bf Input:} Query access to $f \colon \{0, 1\}^n \to \{0, 1\}$, $x \in \{0, 1\}^n$, 
 and a nonempty set $S \subseteq [n]$.

\noindent {\bf Output:}  Either a variable $i \in S$ with $f(x^{(i)}) \neq f(x)$, or ``fail.''

\begin{enumerate}
\item Query $f(x)$ and set $b \leftarrow f(x)$.  

%
\item Draw $L=\lceil 4\log n\rceil$ subsets $\bT_1, \dots, \bT_L \subseteq S$ of size 
$\red{t=\left\lfloor (|S|-1) /{2} \right\rfloor+1}$ uniformly.
\item Query $f(x^{(\bT_\ell)})$ and set the output to be $\bb_\ell$ for each $\ell\in [L]$. Let $\bC \subseteq S$ where\vspace{-0.06cm}
	\[ \bC = \bigcap_{\ell \in [L] \colon\hspace{-0.02cm}\bb_\ell \neq b} \bT_\ell\ \ \ \ \ \ \text{($\bC=\emptyset$ by default if $\bb_\ell=b$ for all $\ell$)}. \vspace{-0.2cm}\]
\item If $\bC = \{ \bi \}$ for some $\bi$, query $f(x^{(\bi)})$ and return $\bi$ if $f(x^{(\bi)}) \neq b$;
otherwise return ``fail.''
\end{enumerate}
\end{flushleft}\vskip -0.14in
\end{framed}\vspace{-0.2cm}
\caption{Description of the adaptive edge search subroutine.\vspace{-0.15cm}} \label{fig:edge-search}
\end{figure}



\begin{lemma} 
Given a point $x\in \{0,1\}^n$ and a set $S\subseteq [n+1]$,
$\AESearch\hspace{0.05cm}(f,x,S)$ makes $O(\log n)$ queries to $f$, and returns either an
  $i\in S$ such that $(x,x^{(i)})$ is a bichromatic edge, or ``fail.''
  
Let $(x, x^{(i)})$ be a bichromatic edge of $f$ along $i$. 
If $i\in S$ and $(x,x^{(i)})$ is $(S\setminus \{i\})$-persistent, then 
  both $\AESearch\hspace{0.05cm}(f,x, S )$ and $\AESearch\hspace{0.05cm}(f,x^{(i)}, S )$ 
  output $i$ with probability at least $2/3$.
\end{lemma}

The first part of the lemma follows directly from the description of $\AESearch$.
For the second part, we prove it for $\AESearch\hspace{0.03cm}(f,x, S)$ 
   since the proof for $\AESearch\hspace{0.03cm}(f,x^{(i)}, S)$ is symmetric.
  The proof proceeds by exactly the same as Claim~6.6 and  6.7 of \cite{CWX17b}, except for some minor notational differences. We present a proof of the claims but adapted to the notation of this paper.
\begin{claim}
Let $(x, x^{(i)})$ be a bichromatic edge and $i\in S\subseteq [n+1]$.
If $(x, x^{(i)})$ is $(S\setminus \{i\})$-persistent, then, in line 3 of $\AESearch\hspace{0.05cm}(f,x, S)$, $i \in \bC$ with probability at least $1 - o(1)$. 
\end{claim}
\begin{proof}
Let $\bT_1, \dots, \bT_{L} \subseteq S $ be subsets sampled in line 2 of $\AESearch$. 
The parameter $t$ satisfies
$$
t=\left\lfloor \frac{|S|-1}{2}\right\rfloor +1\ge \frac{|S|}{2}.
$$
We note that there are two events where $i \notin \bC$: (1) either all $\bT_{\ell}$ satisfy $f(x^{(\bT_{\ell})}) = b$, or (2) there is an $\ell \in [L]$ with $i \notin \bT_{\ell}$ and $f(x^{(\bT_{\ell})}) \neq b$. We show that the probability of either event occurring is at most $o(1)$, so the claim follows by a union bound. 

For the first event, a single sample of a random set $\bT \subseteq S$ of size $t$ satisfies
\begin{align*}
\Prx_{\substack{\bT \subseteq S \\ |\bT| = t}}\left[ f(x^{(\bT)}) \neq b\right] &\geq \Prx_{\substack{\bT\subseteq S \\ |\bT| = t}}\big[i\in \bT \big]\cdot \Prx_{\substack{\bT' \subseteq S\setminus \{i\} \\ |\bT'| = t-1}}\left[ f(x^{(i \cup \bT')}) \neq b \right] \geq \frac{t}{|S| } \cdot \left( 1 - \frac{1}{\log^2 n}\right) \geq \frac{1}{2} - o(1),
\end{align*}
where we used the fact that $f(x^{(i)}) \neq b$ since $(x, x^{(i)})$ is bichromatic, and the assumption
  that $x^{(i)}$ is $(S\setminus \{i\})$-persistent. Thus the probability that all $\bT_{\ell}$ satisfy $f(x^{(\bT_{\ell})}) = b$ is $(0.5 + o(1))^{L} = o(1)$.

For the second event, using the assumption that $x$ is $(S\setminus \{i\})$-persistent, we have that
\begin{align*} 
\Prx_{\bT_1, \dots, \bT_{L}}\big[\exists \ell \in [L]: f(x^{(\bT_{\ell})}) \neq b \text{ and } i \notin \bT_{\ell}\big] &\leq L \cdot \Prx_{\substack{\bT \subseteq S \\ |\bT| = t}}\left[ f(x^{(\bT )}) \neq b \hspace{0.1cm}\big|\hspace{0.1cm} i \notin \bT\right] \leq \dfrac{L}{\log^2 n} = o(1).
\end{align*}
This finishes the proof of the claim.
\end{proof}
\begin{claim}
Let $(x, x^{(i)})$ be a bichromatic edge and $i\in S\subseteq [n+1]$. If $(x, x^{(i)})$ is $(S\setminus\{i\})$-persistent, then in line 3 of $\AESearch\hspace{0.05cm}(f,x, S)$, $\bC$ does not 
  contain any $j\ne i$ with probability at least $1 - o(1)$.
\end{claim}
\begin{proof}
Fix a $j\in S$ but $j\ne i$.
Note that in order for $j \in \bC$, every $\ell \in [L]$ with $f(x^{(\bT_{\ell})}) \neq b$ satisfies $j \in \bT_{\ell}$. However, since $x^{(i)}$ is $(S\setminus \{i\})$-persistent we have 
\begin{align*} 
\Prx_{\substack{\bT \subseteq S \\ |\bT| = t}}\left[ j \notin \bT \text{ and } f(x^{(\bT )}) \neq b\right] &\geq \Prx_{\substack{\bT \subseteq S \\ |\bT| = t}}\big[ i \in \bT\big]\cdot \Prx_{\substack{\bT' \subseteq S\setminus\{i \} \\ |\bT'| = t-1}}\left[ j \notin \bT' \text{ and } f(x^{(\bT'\cup \{i\})}) \neq b \right] \\[0.5ex]
	&\geq \dfrac{t}{|S| } \cdot \left(1-\frac{t-1}{|S|-1} - \frac{1}{\log^2 n}\right) \geq \frac{1}{4} - o(1).
\end{align*}
Therefore, $j\in \bC$ with probability $(3/4 + o(1))^{L} = o(1/n)$.
The claim follows from a union bound.
\end{proof}

It follows by combining these two claims that $\bC=\{i\}$ in line 4 with probability at least $1-o(1)$.
This finishes the proof of the lemma.

\section{Analysis of the Preprocessing Procedure}\label{proof:pruneanalysis}

 
We start with the following property of $\TestAB$: 
 


\begin{claim}\label{lem:test-lemma}
Given a nonempty set $S \subseteq [n+1]$, an ordering $\pi$ of $S$ and $\xi\in (0,1)$,
  $\TestAB$ $(f,S,\pi,\xi)$ makes $O({\log^\red{5} n}/{\xi})$ many queries to $f$.
Furthermore, if the fraction of points that are not $S$-persistent is at least $\xi$,
  $\TestAB\hspace{0.05cm}(f,S,\pi,\xi)$ returns a variable $i\in S$ with probability at 
  least $1-\exp(-\Omega(\log^2 n))$.
\end{claim}
\begin{proof}
The first part of the claim follows from the description of $\TestAB$.
For the second part we note that if the fraction of points that are not $S$-persistent is at least $\xi$,
  then the probability of $\bx$ and $\bT$ with
  $f(\bx)\ne f(\bx^{(\bT)})$ is $\Omega(\xi/\log^2n)$.
It then follows from the number of times we repeat in $\TestAB$.
\end{proof}

We recall the lemma we need for $\Prune$:


\begin{lemma}
Given a Boolean function $f\colon\{0,1\}^n\rightarrow \{0,1\}$, 
  a nonempty $S_0 \subseteq [n+1]$, an ordering $\pi$ of $S_0$
  and a parameter $\xi\in (0,1)$, $\Prune\hspace{0.05cm}(f,S_0, \pi,\xi)$ makes
  at most  $O(|S_0|\hspace{0.05cm}{\log^\red{5} n}/{\xi})$  queries to $f$ and with probability at least $1 - \exp\left(-\Omega(\log^2 n)\right)$,
  it returns a subset $\bS \subseteq S_0$ 
  such that 
   at least $(1-\xi)$-fraction of points in $\{0,1\}^n$ are $\bS$-persistent. 
\end{lemma}


\begin{proof}
The query complexity follows from the fact that $\Prune$ makes at most
  $|S_0|$ many calls to $\TestAB$.
In addition, for each call, it follows from Claim \ref{lem:test-lemma} 
  that the probability of $\TestAB$ returning $\nil$ while $S$ is actually persistent over
  less than $(1-\xi)$-fraction of points is at most $\exp(-\Omega(\log^2 n))$.
The lemma follows from a union bound over $|S_0|\le n$ calls.
\end{proof}

\section{Overlap of Two Random Sets of Certain Sizes}
\label{apphehe}

Let $k,\ell\in [n]$ be two positive integers with $\alpha=k\ell/n$.
We are interested in the size of $|\bS\cap \bT|$ where $\bS$ is a random $k$-sized subset of $[n]$
  and $\bT$ is a random $\ell$-sized subset of $[n]$, both drawn uniformly.

\begin{lemma}\label{appendix1}
For any $t\ge 4\alpha$, the probability of $|\bS\cap\bT|\ge  t$ is at most $\exp(-\Omega(t))$.
\end{lemma}
\begin{proof}
We assume without loss of generality that $k\ge \ell$.
If $\ell> n/2$, the claim is trivial as $\alpha> n/4$ and $t\ge 4\alpha> n$.
We assume $\ell\le n/2$ below.

We consider the following~process.
We draw $\bS$ first. Then we add random (and distinct) indices of $[n]$ to  $\bT$ round by round for $\ell$ rounds.
In each round we pick an index uniformly at random from those that
  have not been added to $\bT$ yet.
Clearly this process generates the same distribution of $\bS$ and $\bT$ that we are interested in.


For each $i\in [\ell]$, we let $\bX_i$ be the random variable that is set to $1$ if
  the index in the $i$th round belongs to $\bS$ and is $0$ otherwise.
Although $\bX_i$'s are not independent, the probability~of $\bX_i=1$ is at most
  $k/(n-\ell)\le 2k/n$ using $\ell\le n/2$,
  for any fixed values of $\bX_1,\ldots,\bX_{i-1}$.
Thus, the expectation of $\sum_{i\in [\ell]} \bX_i$ is at most $2k\ell/n=2\alpha$.
The lemma follows directly from the Chernoff bound (together with a standard coupling argument).
\end{proof}

\end{document}